%% file: main.tex
\documentclass[journal, final]{IEEEtran}

\IEEEoverridecommandlockouts

\input{preamble_IEEE.tex}

\begin{document}

\title{Tail-Calibrated Soft-Output GRAND for\\Finite-Memory Noise-Effect Posteriors}
\author{Behrooz~Razeghi~$\!{\orcid{0000-0001-9568-4166}}$,~
\IEEEmembership{Senior~Member,~IEEE}%
\thanks{B.~Razeghi is with Harvard University, Cambridge, MA, USA (e-mail: behroozrazeghi@seas.harvard.edu).}
\thanks{This work was supported by the Swiss National Science Foundation (SNSF) under Grant No.~222339.}%
}

\maketitle

\begin{abstract}
In guessing random additive noise decoding (GRAND), memory in the hard-decision noise effect changes the likelihood order of candidate noise effects. In soft-output decoding, the same memory also affects the finite-block quantity determining the missing-list probability: the codebook-restricted posterior mass outside the current list. Existing correlation-aware GRAND methods exploit local dependence without interleaving, but their stopping and soft-output rules are not derived from finite-memory posterior tails. Soft-output GRAND (SOGRAND) derives random-codebook a posteriori probability (APP) estimates for GRAND lists, but does not provide finite-memory algorithms for posterior weights, partition functions, tail masses, or bitwise tail marginals for correlated noise-effect posteriors. We introduce Tail-Calibrated SOGRAND for binary additive channels whose ambient hard-decision noise-effect posterior, conditioned on received soft information, is represented by a finite-memory energy. The decoder enumerates candidate noise effects in nondecreasing posterior energy, queries codebook membership as in GRAND, computes posterior weights and tail masses by finite-state recursions, and estimates the unqueried codebook-restricted denominator as $p_qT_q$, where $T_q$ is the ambient posterior tail mass and $p_q$ is the remaining random-codebook occupancy probability. With exact enumeration and no abandonment, the first listed codeword is ML under the likelihood model defining the posterior energy. We also prove an ambient posterior-tail abandonment bound and, separately, conditional unbiasedness, variance, and concentration bounds for the random-codebook missing-list estimator. The same posterior-tail decomposition gives blockwise APP estimates, missing-list probabilities, and bitwise APP log-likelihood ratios (LLRs) for finite-memory noise-effect posteriors.
\end{abstract}


\vspace{26pt}

\section{Introduction}

\IEEEPARstart{M}{any} soft-input channel decoders are derived for, or are commonly used with, a memoryless channel law in which the likelihood factors as $P_{R^n|X^n}(r^n|x^n) = \prod_{i=1}^n P_{R_i|X_i}(r_i|x_i)$. This factorization enters through the channel evidence supplied to the decoder, either explicitly through a memoryless-channel model or implicitly through per-symbol reliability metrics used to initialize or order the decoding procedure. Examples include belief propagation on low-density parity-check (LDPC) code factor graphs, iterative turbo decoding, cyclic-redundancy-check-aided successive-cancellation list decoding of polar codes, and ordered-reliability GRAND-type decoding \cite{gallager1962low, kschischang2001factor,berrou1993near, arikan2009channel, niu2012crc, tal2015list, duffy2022ordered}. The memoryless likelihood model can be mismatched when the physical channel or the receiver front end induces memory in the effective noise seen by the decoder. Time-correlated fading, residual equalization error, and burst noise can produce temporally correlated hard-decision noise effects and correlated reliability values. A common system-level response is interleaving, which reduces local dependence at the decoder input and makes memoryless reliability metrics more appropriate. Interleaving, however, adds delay and disperses local dependence across separated decoder coordinates, thereby obscuring structure that a correlation-aware decoder could exploit \cite{an2022keep, duffy2023using}. In low-latency regimes, this delay is part of the reliability--latency tradeoff rather than a negligible preprocessing cost \cite{an2022keep, duffy2023using}.

\vspace{1pt}

GRAND reformulates maximum-likelihood decoding as a search over noise effects rather than over codewords \cite{duffy2019capacity}. Given a hard-decision received word $y^n$, GRAND queries candidate noise effects $z^n$ in nonincreasing probability order and tests whether $y^n\oplus z^n\in C_n$. For a binary additive channel whose noise effect is independent of the transmitted codeword, the first successful query returns an ML codeword when the noise-effect ordering is exact, with ties among equally likely noise effects resolved by a fixed rule. The search rule is code-agnostic except for the codebook membership test: it requires a membership oracle for $C_n$ and an ordering of the noise effects, but it does not otherwise require a code-specific decoding graph or trellis. The same ML argument applies to additive channels with memory whenever the candidate noise effects can be ordered according to the likelihood induced by the true noise law. Approximate orderings can reduce implementation cost, but they do not, in general, preserve finite-block ML decoding. GRANDAB imposes a preassigned query threshold and declares abandonment if no codeword is found within that threshold \cite{duffy2019capacity}. Although GRAND with abandonment is not, in general, a finite-block ML decoder, it remains capacity-achieving for random codebooks when the abandonment threshold is chosen so that the realized noise effect is queried with probability tending to one \cite{duffy2019capacity}.

\vspace{1pt}

Soft-input variants of GRAND extend the noise-effect search principle to settings in which reliability information is available at the decoder.
Soft GRAND (SGRAND) uses received-word-dependent reliabilities to construct a noise-effect query order corresponding to the soft-input ML rule under the assumed channel model \cite{solomon2020soft}. Ordered Reliability Bits GRAND (ORBGRAND) gives a structured approximation to this ordering: it sorts bit positions by increasing reliability and generates candidate noise effects in nondecreasing logistic weight \cite{duffy2022ordered}. Under this rule, pattern generation can be formulated as an ordered integer-partition problem and implemented by the Landslide algorithm \cite{duffy2022ordered}. Hardware studies have reported high-throughput, energy-efficient implementations of ORBGRAND-style orderings, including integrated universal GRAND decoders and a $40$-nm ORBGRAND chip with measured sub-pJ/bit energy \cite{riaz2021multi, riaz2024sub}. ORBGRAND-AI further showed that local dependence in the noise effect can be incorporated without interleaving by forming substitution metrics over local symbol neighborhoods and applying an ORBGRAND-style ordering to the resulting substitutions \cite{duffy2023using}.

\vspace{1pt}

A complementary line of work, soft-output GRAND (SOGRAND), focuses on soft-output reliability estimation for GRAND lists rather than finite-memory noise-effect modeling \cite{yuan2025soft}. Iterative decoding schemes for product codes, generalized low-density parity-check (GLDPC) codes, and concatenated constructions require component decoders to produce reliability information that can be exchanged across decoding iterations \cite{pyndiah1998near, galligan2023block, riaz2025iterative, yuan2025soft}. Related list-based reliability estimates appear in Forney's erasure/list decoding analysis and in Pyndiah's turbo-product-code decoder \cite{forney1968exponential, pyndiah1998near}. An earlier GRAND reliability formulation introduced per-decoding correctness and missing-list estimates for GRAND-generated lists \cite{galligan2023upgrade}. SOGRAND subsequently developed random-codebook APP estimates for listed candidates and for the probability that the transmitted codeword is absent from the current list \cite{yuan2025soft}. This missing-list term is used to calibrate blockwise and bitwise reliability estimates and allows GRAND to be used as a component decoder in iterative decoding of product and GLDPC codes. In the additive white Gaussian noise (AWGN) and fading-channel settings evaluated in \cite{yuan2025soft}, the resulting decoders were reported to outperform the selected 5G LDPC baselines.
 
\vspace{1pt}

Correlation-aware GRAND orderings and SOGRAND leave open a finite-block posterior-tail computation for correlated noise-effect posteriors. ORBGRAND-AI incorporates local dependence in the noise-effect ordering \cite{duffy2023using}. More recently, soft GRAND for intersymbol interference (SGRAND-ISI) derived a matched noise-effect ordering for linear Gaussian intersymbol-interference channels, together with ORBGRAND-type approximations \cite{li2026grand}. These methods address correlated-channel noise-effect ordering, but do not derive finite-memory computations of the ambient posterior tail mass or bitwise tail marginals. SOGRAND derives random-codebook APP estimates for GRAND lists, whereas the present work develops the finite-memory posterior computations needed to obtain posterior weights, partition functions, posterior tail masses, and bitwise tail marginals for correlated noise-effect posteriors \cite{yuan2025soft}. Thus, the GRAND-type constructions mentioned above do not provide, in a single finite-block procedure, both the ambient posterior tail probability $\Pr\{Z^n\notin \mathcal A_q\mid R^n=r^n\}$ and the codebook-restricted missing-list estimate for $\Pr\{X^n\notin \mathcal L_q\mid R^n=r^n,C_n\}$, where $\mathcal A_q$ is the set of noise effects queried after $q$ steps, and $\mathcal L_q$ is the corresponding list of codewords found by membership queries.

\vspace{1pt}

We address this problem for binary additive channels. In this setting, the hard-decision vector satisfies $Y^n=X^n\oplus Z^n$, so a candidate
noise effect $z^n$ induces the candidate codeword $y^n\oplus z^n$. We model the ambient full-space conditional law of $Z^n$, given the received
soft information and before imposing the codebook constraint, by a finite-range posterior energy. Specifically, for a received soft vector $R^n=r^n$, define the ambient full-space posterior weights $\pi_{r^n}(z^n) = \exp\{-E_{r^n}(z^n)\}/\mathsf Z(r^n)$, $\mathsf Z(r^n) = \sum_{u^n\in\{0,1\}^n}\exp\{-E_{r^n}(u^n)\}$. Under the assumed finite-memory model, these weights define the conditional distribution of the hard-decision noise effect before the codebook membership constraint is imposed. The codebook-conditioned APP is obtained by restricting these weights to the codeword-consistent noise effects $z^n$ satisfying $y^n\oplus z^n\in C_n$, and then renormalizing over that restricted set. Our formulation includes memoryless posterior factorizations as a special case, and includes first-order Markov and finite-state burst models as
finite-memory instances. The decoder enumerates noise effects in nondecreasing energy, queries whether $y^n\oplus z^n\in C_n$, and maintains the accumulated ambient posterior mass of the queried prefix together with the remaining ambient posterior tail mass. For finite-range energies with fixed memory order, the partition function is computed by a finite-state forward recursion, and the ambient tail mass outside the queried prefix is obtained by subtracting the accumulated queried mass
from one. Under the fixed-size random-codebook occupancy model, this ambient tail mass gives the estimator $\widehat U_q=p_qT_q$, $p_q= \frac{M-|\mathcal L_q|}{2^n-q}$, $ M=|C_n|$, for $q<2^n$, where $T_q$ is the ambient posterior mass of the unqueried tail. The resulting stopping criterion uses the plug-in missing-list
estimate $\widehat P_{\rm miss}(q)=\widehat U_q/(S_q+\widehat U_q)$ rather than a fixed membership-query limit.

\vspace{1pt}

We prove finite-block guarantees for the proposed procedure. With exact energy ordering and no abandonment, the first listed codeword is ML under
the likelihood model that defines the posterior energy. With ambient-tail abandonment, the additional probability of failing to query the realized
noise effect is bounded by the abandoned ambient posterior tail mass, provided that the assumed ambient posterior model $\pi_{r^n}$ equals the
true full-space conditional law of the hard-decision noise effect given $R^n=r^n$. Separately, under the fixed-size random-codebook model, we prove that $p_qT_q$ is the conditional-mean estimator of the unqueried codebook-restricted denominator contribution, and derive variance and concentration bounds controlled by the effective support size of the unqueried posterior tail. The same posterior-tail decomposition yields blockwise APP estimates and bitwise APP LLRs under a fixed $0$-over-$1$ convention.

\vspace{1pt}

We evaluate the decoder using three primary metrics. The block error rate (BLER) measures decoding accuracy. Calibration curves compare the plug-in
missing-list estimate with the empirical frequency of the event $X^n\notin\mathcal L_q$, using binned reliability estimates as in
SOGRAND \cite{yuan2025soft}. Query complexity is measured by the number of codebook membership tests. When comparing different enumerators, we
also report a software-work proxy that includes posterior-metric evaluations and priority-queue removals. Membership-query counts are
directly relevant to GRAND decoding because each queried noise effect requires a codebook membership test, as in the complexity analyses of GRAND and ORBGRAND \cite{duffy2019capacity, duffy2022ordered, riaz2024sub}.

\clearpage
 
Our contributions are as follows.
\begin{enumerate}[leftmargin=15pt]
\item We formulate Tail-Calibrated SOGRAND for binary additive channels whose hard-decision noise-effect posterior, conditioned on the received soft observation $R^n=r^n$, is represented by a finite-memory energy. Our formulation separates the ambient posterior law of the hard-decision noise effect from the codebook-restricted APP denominator induced by the membership constraint $y^n\oplus z^n\in C_n$.
\item 
We derive a tail-calibrated soft-output rule for GRAND searches under a fixed-size random-codebook ensemble. After $q$ queries, the unqueried codebook restricted denominator contribution $U_q$ differs from the ambient posterior tail mass $T_q$ because only unqueried noise effects whose associated candidate words are codewords contribute to the APP denominator. Conditional on the received soft information, the queried prefix, and its membership outcomes, we show that, for the fixed-size random-codebook ensemble, the conditional mean of the unqueried codebook-restricted denominator contribution is $p_qT_q$, where $p_q=(M-|\mathcal L_q|)/(2^n-q)$.

\item We prove finite-block reliability guarantees for energy-ordered decoding and ambient-tail abandonment. With exact enumeration and no abandonment, the first codeword encountered in nondecreasing posterior energy is ML under the likelihood model that defines the posterior energy. With ambient-tail abandonment, the additional probability of failing to query the realized noise effect is bounded by the abandoned ambient posterior tail mass, provided that the assumed posterior equals the true conditional law of the hard-decision noise effect. We also bound the error in the plug-in missing-list probability in terms of the relative error in $\widehat U_q$.

\item We derive blockwise and bitwise APP estimates from the same posterior-tail decomposition. The blockwise rule combines the posterior masses of listed codewords with the estimated unqueried codebook-restricted contribution $\widehat U_q$. The bitwise rule decomposes the unqueried contribution according to bit value and yields APP LLRs under a fixed $0$-over-$1$ convention.

\item We evaluate the method in a random-linear-code setting over Gauss--Markov and hard-decision correlated noise models. The experiments compare random-interleaved ORBGRAND, ORBGRAND-AI local-block orderings, ExactBlockProduct ablations using the same nonoverlapping block structure, and Exact-Markov Tail-Calibrated SOGRAND. In the first-hit membership-oracle comparison, matched finite-memory posterior ordering gives lower BLER point estimates and smaller average membership-query counts than the tested memoryless, rank-weighted approximate-independence, and block-product approximate-independence orderings, with an enumeration-cost tradeoff in the present implementation.
\end{enumerate}


\section{Related Work}
\label{sec:related_work}

\subsection{Code-specific soft decoding}

Many conventional soft decoders exploit algebraic or graphical structure specific to the code family. Low-density parity-check (LDPC) decoders use sparse parity-check constraints and message passing on factor graphs \cite{gallager1962low,kschischang2001factor}. Turbo-code decoders use the concatenated convolutional structure of the code and iteratively exchange extrinsic information between constituent decoders \cite{berrou1993near}. CRC-aided polar decoding combines polar codes \cite{arikan2009channel} with CRC-aided list decoding \cite{niu2012crc,tal2015list}. Product-code and turbo-product-code decoders build on Elias product codes \cite{elias1954error} and Pyndiah's reliability-output Chase-type component decoding \cite{pyndiah1998near}.  These methods can achieve low error rates at finite blocklengths in their target regimes by using decoding rules tailored to the code structure. This coupling between code family and decoder, however, limits code-family agnosticism, since changing the code family typically changes the decoding algorithm or architecture.

GRAND-based decoders separate the codebook membership test from the noise-effect ordering rule. As a result, the noise-search procedure is code-agnostic up to the membership test: the decoder requires an ordering of candidate noise effects and an oracle for membership in the codebook, but not a code-specific decoding graph, trellis, or algebraic decoder \cite{duffy2019capacity,riaz2021multi,riaz2024sub}.

\subsection{GRAND and GRANDAB}

For a codebook $C_n\subseteq\mathcal A^n$, an additive noise effect $N^n$, and a received word $Y^n=X^n\oplus N^n$, GRAND queries candidate noise sequences $z^n\in\mathcal A^n$ in nonincreasing order of $P(N^n=z^n)$. For each candidate $z^n$, it tests whether $Y^n\ominus z^n\in C_n$. Under the additive invertible channel model, with $N^n$ independent of
the transmitted codeword,
\begin{equation}
P(Y^n=y^n\mid X^n=c^n) = P(N^n=y^n\ominus c^n).
\end{equation}
Hence
\begin{equation}
\arg\!\max_{c^n\in C_n} \!P(Y^n\!=\!y^n\mid X^n\!=\!c^n)
\!=\! \arg\max_{c^n\in C_n} P(N^n\!=\!y^n\ominus c^n).
\end{equation}
Therefore, when the noise-effect ordering is exact, the first successful query returns an ML codeword, with ties among equally likely noise effects resolved by a fixed deterministic rule.

GRAND was shown to be capacity-achieving for random codebooks under suitable assumptions on the noise process, with error, success, and query complexity characterized through guesswork large deviations \cite{christiansen2012guesswork,duffy2019capacity}. GRANDAB adds abandonment: if no codeword is found within a prescribed query threshold, the decoder declares abandonment. GRAND with abandonment is not, in general, a finite-block ML decoder, because it may stop before reaching the noise effect associated with an ML codeword. Nevertheless, it remains capacity-achieving for random codebooks when the query threshold is chosen so that the realized noise effect is queried with probability tending to one while the probability of an earlier nontransmitted-codeword hit tends to zero. In particular, if the noise entropy rate is $H$, measured in base $|\mathcal A|$, then for any rate $R<1-H$ one may choose $\delta>0$ such that $H+\delta<1-R$ and use a threshold of order $|\mathcal A|^{n(H+\delta)}$ \cite{duffy2019capacity}.

Universal noise-guessing decoders have also been studied for unknown discrete and finite-state additive channels \cite{miyamoto2025universal, miyamoto2025universal_journal}. Those works address universality with respect to an unknown channel law and provide random-coding error and complexity guarantees. The present work instead assumes a specified finite-memory posterior model and develops finite-block posterior-tail and soft-output computations under that model.

\subsection{Soft-input GRAND and ORBGRAND}

Soft GRAND (SGRAND) and SRGRAND extend the GRAND noise-effect search to soft-input decoding by using received reliability information to construct received-word-dependent noise-effect orderings \cite{solomon2020soft, duffy2021guessing}. Under the assumed soft-input channel model, this ordering gives a soft-input ML benchmark. The exact SGRAND ordering, however, is less directly suited to highly parallel hardware implementations because the complete ordering of noise effects depends on the full reliability vector of the received word.

Ordered Reliability Bits GRAND (ORBGRAND) replaces this exact received-word-dependent ordering with a structured reliability-based approximation suitable for hardware implementation \cite{duffy2022ordered}. It sorts bit positions by increasing reliability, so that the least reliable position has rank one, and generates putative noise effects in nondecreasing reliability weight. In the zero-intercept logistic-weight model, the cost assigned to a binary noise pattern $z^n$, after this reliability sorting, is proportional to $w_L(z^n)=\sum_{i=1}^n i z_i$. Under this logistic-weight rule, noise-effect generation can be organized as an ordered integer-partition problem and implemented by the Landslide algorithm \cite{duffy2022ordered}. ORBGRAND has been shown empirically to approach SGRAND performance for moderate-redundancy short codes, and subsequent information-theoretic analysis has characterized ORBGRAND as almost capacity-achieving over the AWGN channel with antipodal input \cite{duffy2022ordered, liu2022orbgrand}.

\subsection{Correlation-aware GRAND and ORBGRAND-AI}

Many soft-input decoders are derived under a memoryless likelihood model, in which the channel likelihood factors across coded bits or symbols. Channels with memory violate this factorization through mechanisms such as burst noise, time-correlated fading, residual equalization error, and temporally correlated interference. A common system-level response is interleaving, which makes the decoder input more nearly consistent with a memoryless likelihood model. Interleaving, however, increases latency and spreads local statistical dependence across separated positions, thereby removing structure that a matched decoder could exploit.

GRAND-MO and related Markov-order GRAND methods showed that correlated noise effects can be modeled directly rather than dispersed by interleaving \cite{an2022keep}. ORBGRAND-AI extends this idea to soft detection by forming local blocks of received symbols, computing reliability metrics for alternative block substitutions, and applying an ORBGRAND-style ordering under an approximate-independence factorization \cite{duffy2023using}. Related interleaving-free ORBGRAND-AI decoding for channels with intersymbol interference was studied in \cite{duffy2025decoding}. The simulations reported in \cite{duffy2023using} show that exploiting local channel correlation can improve BLER and avoid the delay associated with interleaving, relative to the interleaved benchmark considered in that work. These correlation-aware methods address noise-effect ordering, but do not compute finite-memory ambient posterior tail masses, codebook-restricted missing-list estimates, or bitwise tail marginals.

In our earlier LP-GRAND work, we developed an exact likelihood-ordered enumerator for BPSK transmission over correlated Gaussian noise when the interaction graph induced by the specified precision matrix admits a bounded-width path decomposition~\cite{razeghi2026lpgrand}. The same graphical representation also supports exact computation, under the assumed Gaussian metric, of the full-space partition function and the normalized posterior weights of the ambient noise-effect patterns. In the present work, we use these finite-memory posterior quantities to construct soft outputs: we compute the posterior mass of the unqueried patterns, estimate their codebook-restricted contribution to the APP denominator under the random-codebook model, and derive blockwise and bitwise posterior quantities.

\subsection{Soft-output GRAND}

A complementary line of GRAND-related work addresses soft-output generation and soft-input/soft-output use. An early GRAND reliability formulation introduced per-decoding correctness and missing-list estimates for GRAND-generated lists \cite{galligan2023upgrade}. Related developments include soft-input soft-output joint data detection, soft-output guessing-codeword decoding, and SOGRAND \cite{sarieddeen2023soft,duffy2024soft,yuan2025soft}. Iterative decoding of product codes, generalized low-density parity-check (GLDPC) codes, and concatenated constructions often requires component decoders to produce reliability information that can be exchanged across component-decoding steps \cite{elias1954error, pyndiah1998near, galligan2023block, riaz2025iterative, yuan2025soft}. Related
list-based reliability measures trace back to Forney's erasure/list decoding analysis \cite{forney1968exponential} and to Pyndiah's turbo-product-code decoder \cite{pyndiah1998near}. Many list-based reliability rules condition their estimates on the event that the transmitted codeword is contained in the list. SOGRAND develops random-codebook APP estimates for listed candidates and for the probability that the transmitted codeword is absent from the current list \cite{yuan2025soft}. This missing-list term was shown in \cite{yuan2025soft} to improve blockwise and bitwise soft-output calibration and to enable iterative GRAND decoding of product and GLDPC codes. In the AWGN and fading settings evaluated in \cite{yuan2025soft}, the resulting decoders were reported to outperform selected 5G LDPC baselines. 

Subsequent work uses the SOGRAND missing-list estimate for dynamic list termination and guesswork reduction \cite{rapp2025sogrand}, exploits binary linear-code constraints to improve GRAND-family soft-output estimates beyond the uniform random-codebook approximation \cite{feng2025leveraging}, and specializes SOGRAND to single-parity-check component decoding for LDPC check-node updates \cite{duffy2026sogrand}. The SOGRAND formulas apply to GRAND query orders for which the queried noise-effect probabilities are available \cite{galligan2023upgrade,yuan2025soft}. The present paper develops the finite-memory posterior computations needed to obtain these probabilities, posterior tail masses, and bitwise tail marginals for correlated noise-effect posteriors, and uses those quantities to estimate the unqueried codebook-restricted denominator contribution.

\vspace{-2pt}

\section{System Model}

\vspace{-2pt}

Let $C_n\subseteq\{0,1\}^n$ be a binary codebook with $|C_n|=M=2^k$ and rate $R_{\rm c}=k/n$. The transmitted codeword $X^n$ is drawn uniformly from $C_n$. The receiver observes a real-valued vector $R^n\in\mathbb R^n$, and we denote its realization by $r^n$. Let $h:\mathbb R\to\{0,1\}$ be the hard-decision map, and define the hard-decision vector $Y^n$ by $Y_i=h(R_i),\qquad i=1,\ldots,n $. For the received realization $R^n=r^n$, we write $y_i=h(r_i)$. The hard-decision noise effect associated with the transmitted codeword is $Z^n=Y^n\oplus X^n$. Thus, conditioned on $R^n=r^n$, a candidate word $c^n\in\{0,1\}^n$ corresponds to the candidate noise effect $z^n(c^n)=y^n\oplus c^n $.
Since $X^n$ is uniform on $C_n$, the APP over the codebook is
\begin{equation}\label{eq:codebook_app}
\Pr\{X^n \! = \!c^n\mid R^n \!= \! r^n,C_n\} \!=\!
\frac{\Lambda_{r^n}(z^n(c^n))}{\sum_{\tilde c^n\in C_n}\Lambda_{r^n}(z^n(\tilde c^n))}, c^n \!\in \!C_n,
\end{equation}
where $\Lambda_{r^n}(z^n) \coloneqq p_{R^n|X^n}(r^n\mid X^n=y^n\oplus z^n)$ is the likelihood density of the soft observation under the candidate codeword $y^n\oplus z^n$. Multiplying all likelihoods in \eqref{eq:codebook_app} by the same positive constant does not change the
posterior distribution on $C_n$. We therefore define the ambient (auxiliary) full-space posterior model
\begin{equation}\label{eq:ambient_pi}
\pi_{r^n}(z^n) \coloneqq
\frac{\Lambda_{r^n}(z^n)}{\sum_{u^n\in\{0,1\}^n}\Lambda_{r^n}(u^n)}, \qquad z^n\in\{0,1\}^n .
\end{equation}
The distribution $\pi_{r^n}$ is normalized over all binary noise effects. 
Equivalently we can interpret it as a normalized likelihood over candidate noise effects before the codebook constraint is imposed. It is not the codebook-conditioned APP induced by the uniform prior on $C_n$. The codebook APP in \eqref{eq:codebook_app} is obtained by restricting the normalizing sum to the noise effects associated with codewords, $\{\,z^n(c^n)=y^n\oplus c^n:\ c^n\in C_n\,\}$. Thus $\pi_{r^n}$ is used to order candidate noise effects and to compute ambient tail masses, whereas the codebook-conditioned posterior is obtained by renormalizing the same weights over the codeword-consistent noise effects. The ML codeword can equivalently be written as $\widehat x_{\rm ML} \in \arg\max_{c^n\in C_n}\pi_{r^n}(y^n\oplus c^n)$.

\subsection{Finite-memory ambient posterior energy}

We assume that the ambient full-space posterior model $\pi_{r^n}$ admits a finite-memory energy representation in the hard-decision noise effect. In the first-order case, for strictly positive weights,
\begin{equation}\label{eq:first_order_posterior}
\pi_{r^n}(z^n) = \frac{\exp\{-E_{r^n}(z^n)\}}{\mathsf Z(r^n)},
\end{equation}
where
\begin{equation}\label{eq:first_order_energy}
E_{r^n}(z^n) = \sum_{i=1}^n \alpha_i(z_i;r^n) + \sum_{i=2}^n \beta_i(z_{i-1},z_i;r^n),
\end{equation}
and $\mathsf Z(r^n) = \sum_{u^n\in\{0,1\}^n} \exp\{-E_{r^n}(u^n)\}$. The local potentials $\alpha_i$ and $\beta_i$ may depend on the
received soft vector $r^n$. The memoryless case is obtained by setting $\beta_i\equiv 0$, in which case $\pi_{r^n}$ factorizes across
coordinates. The reliability-ordered logistic weight used in ORBGRAND is an additive ordering cost after reliability sorting; in our notation, it corresponds to a structured memoryless ordering rather than to a finite-memory transition model. First-order Markov models for the hard-decision noise effect are represented by nonzero transition potentials $\beta_i$. Higher-order Markov and finite-state burst models can be treated by replacing the pairwise state $(z_{i-1},z_i)$ with a finite-state trellis state.

A burst-favoring posterior assigns lower excess energy to a contiguous error run than to the corresponding separated single-error contributions. This property can be stated without fixing a particular parameterization of the local potentials. For an interval $I=[\ell,r]\subseteq\{1,\ldots,n\}$, let $\mathbf 1_I$ denote the binary vector that is one on $I$ and zero elsewhere, and define $\Gamma_{r^n}(I) \coloneqq E_{r^n}(\mathbf 1_I)-E_{r^n}(0^n)$. We say that the posterior energy is burst-favoring on $I=[\ell,r]$, with $\ell<r$, if $\Gamma_{r^n}([\ell,r]) < \sum_{i=\ell}^r \Gamma_{r^n}(\{i\})$. Thus, the contiguous error pattern on $I$ has lower excess energy than the sum of the excess energies of isolated single-bit errors at the same positions. Since $\Gamma_{r^n}(I)$ is defined by differences of the full energy $E_{r^n}$, the condition is invariant to additive constants in $E_{r^n}$ and to equivalent local-potential parameterizations that leave $\pi_{r^n}$ unchanged.

\vspace{-5pt}

\section{Tail-Calibrated SOGRAND for Finite-Memory Posteriors}

Let $z^{(1)},z^{(2)},\ldots,z^{(2^n)}$ be an ordering of $\{0,1\}^n$ satisfying
\begin{equation}\label{eq:energy_ordering} 
E_{r^n}(z^{(1)})\le E_{r^n}(z^{(2)})\le\cdots \le E_{r^n}(z^{(2^n)}), 
\end{equation}
with ties resolved by a fixed deterministic rule. For $q\in\{0,\ldots,2^n\}$, define the queried prefix $\mathcal A_q=\{z^{(1)},\ldots,z^{(q)}\}$, with $\mathcal A_0=\emptyset$. At query $j$, the decoder forms the candidate word $c^{(j)}=y^n\oplus z^{(j)}$ and tests whether $c^{(j)}\in C_n$. Let
\begin{equation}
J_q=\{j\le q:y^n\oplus z^{(j)}\in C_n\}
\end{equation}
be the set of successful query indices, and define the corresponding codeword set
\begin{equation}\label{eq:list_def}
\mathcal L_q =\{y^n\oplus z^{(j)}:j\in J_q\}, \quad |\mathcal L_q|=|J_q|.
\end{equation}
When $J_q\neq\emptyset$, define $j_q^{\rm first}=\min J_q$, $c_q^{\rm first}=y^n\oplus z^{(j_q^{\rm first})}$. The word $c_q^{\rm first}$ is the first codeword found by the
membership-query sequence. Since the map $z^n\mapsto y^n\oplus z^n$ is one-to-one on $\{0,1\}^n$, we have $|\mathcal L_q|=|J_q|$. The queried codebook-restricted ambient mass is
\begin{equation}\label{eq:Sq_def}
S_q = \sum_{j\in J_q}\pi_{r^n}(z^{(j)}).
\end{equation}
The ambient tail mass outside the queried prefix is
\begin{equation}\label{eq:Tq_def}
T_q = 1-\sum_{j=1}^q\pi_{r^n}(z^{(j)}) = \sum_{z^n\notin \mathcal A_q}\pi_{r^n}(z^n).
\end{equation}

For the first-order energy model in \eqref{eq:first_order_energy}, the partition function $\mathsf Z(r^n)$ is computed by a finite-state forward recursion. Let $F_i(b)$ denote the unnormalized posterior mass of all prefixes $z_1^i$ ending in $z_i=b$. Initialize $F_1(a)=\exp\{-\alpha_1(a;r^n)\}$, $a\in\{0,1\}$. For $i=2,\ldots,n$, $b\in\{0,1\}$, 
\begin{eqnarray}\label{eq:forward_recursion}
F_i(b) \!
= \! \exp\{-\alpha_i(b;r^n)\}
\!\!\!\!\! \sum_{a\in\{0,1\}}
\!\!\!\!\! F_{i-1}(a)\exp\{-\beta_i(a,b;r^n)\}.
\end{eqnarray}
Thus
\begin{equation}\label{eq:partition_forward}
\mathsf Z(r^n)=F_n(0)+F_n(1).
\end{equation}
Each queried ambient posterior weight is then $\pi_{r^n}(z^{(j)}) = \exp\{-E_{r^n}(z^{(j)})\} / \mathsf Z(r^n)$. Consequently, the tail mass $T_q$ is obtained by subtracting the accumulated queried ambient mass from one. If $\mathsf Z(r^n)$ and the queried weights are computed from the same energy $E_{r^n}$, then $T_q$ is the exact ambient tail mass of $\{0,1\}^n\setminus\mathcal A_q$ under the assumed first-order posterior model. No explicit summation over the unqueried set is required.

\subsection{Random-codebook unqueried denominator}

Using the full-space normalized weights $\pi_{r^n}$, the codebook-restricted normalizing denominator is
\begin{equation}\label{eq:codebook_denominator}
D(C_n,r^n) = \sum_{c^n\in C_n}\pi_{r^n}(y^n\oplus c^n).
\end{equation}
This denominator is the likelihood denominator in \eqref{eq:codebook_app} divided by the common full-space normalizing constant in \eqref{eq:ambient_pi}. After $q$ queries, it decomposes as
\begin{equation}\label{eq:denominator_decomposition}
D(C_n,r^n)=S_q+U_q,
\end{equation}
where
\begin{equation}\label{eq:unqueried_denominator}
U_q = \sum_{z^n\notin \mathcal A_q} \pi_{r^n}(z^n) \mathbf 1\{y^n\oplus z^n\in C_n\}.
\end{equation}
The term $U_q$ is the unqueried codebook-restricted contribution to the normalizing denominator, whereas the ambient tail mass is $T_q=\sum_{z^n\notin\mathcal A_q}\pi_{r^n}(z^n)$. The two quantities differ because the APP denominator receives contributions only from unqueried noise effects whose associated candidate words $y^n\oplus z^n$ are codewords.

Under the fixed-size random-codebook ensemble, $C_n$ is a uniformly chosen subset of $\{0,1\}^n$ with cardinality $M=2^k$. Conditional on the queried prefix, the membership outcomes in that prefix, the listed codewords $\mathcal L_q$, and the received realization $R^n=r^n$, the remaining $M-|\mathcal L_q|$ codewords are uniformly distributed among the $2^n-q$ unqueried candidate words. Hence, for $q<2^n$, the conditional occupancy probability of an unqueried candidate word is $p_q = \frac{M-|\mathcal L_q|}{2^n-q}$. We estimate the unqueried codebook-restricted denominator contribution by $\widehat U_q = p_qT_q$, and define the estimated normalizing denominator $\widehat D_q = S_q+\widehat U_q$.
For each listed codeword $c^{(j)}=y^n\oplus z^{(j)}$, $j\in J_q$, we set
\begin{equation}\label{eq:list_app_estimate}
\widehat P_q\{X^n=c^{(j)}\mid R^n=r^n,C_n\} = \frac{\pi_{r^n}(z^{(j)})}{\widehat D_q}.
\end{equation}
The estimated probability that the transmitted codeword is absent from the current list is
\begin{equation}\label{eq:missing_list_estimate}
\widehat P_q\{X^n\notin\mathcal L_q\mid R^n=r^n,C_n\} = \frac{\widehat U_q}{\widehat D_q}.
\end{equation}

\vspace{-9pt}

\subsection{Noise-pattern enumeration}
\label{subsec:noise_pattern_enumeration}

For finite $n$, exact enumeration in nondecreasing $E_{r^n}(z^n)$ can be implemented by best-first search, or by a $K$-shortest-path procedure, on the finite-state trellis associated with the energy model \cite{seshadri1994list}. Such a procedure produces the exact energy-ordered sequence. However, generating a large prefix of the ordered list can still require exponentially many candidate outputs in $n$. For first-order binary energies, the structure of the energy can be made explicit by decomposing a noise pattern into its maximal runs of ones.

\subsubsection{Burst-atom decomposition}

For a first-order binary energy, define an interval atom $B=[\ell,r]$, $1\le \ell\le r\le n$, with binary indicator vector $\mathbf 1_B$. Its excess energy relative to the all-zero noise effect
is $\Gamma_{r^n}(B)=E_{r^n}(\mathbf 1_B)-E_{r^n}(0^n)$. For $B=[\ell,r]$, expansion of \eqref{eq:first_order_energy} gives
\begin{align}
\Gamma_{r^n}([\ell,r])
& = \sum_{i=\ell}^{r} \bigl(\alpha_i(1;r^n)-\alpha_i(0;r^n)\bigr)
+ \Delta_{\rm in}(\ell) \nonumber \\
& + \sum_{i=\ell+1}^{r} \Delta_{\rm stay}(i)
+ \Delta_{\rm out}(r),
\label{eq:burst_atom_expansion}
\end{align}
where, for $\ell>1$,
\begin{equation}
\Delta_{\rm in}(\ell) = \beta_{\ell}(0,1;r^n)-\beta_{\ell}(0,0;r^n),
\end{equation}
for $\ell<i\le r$,
\begin{equation}
\Delta_{\rm stay}(i) = \beta_i(1,1;r^n)-\beta_i(0,0;r^n),
\end{equation}
and, for $r<n$,
\begin{equation}
\Delta_{\rm out}(r) = \beta_{r+1}(1,0;r^n)-\beta_{r+1}(0,0;r^n).
\end{equation}
The entry term is omitted when $\ell=1$, and the exit term is omitted when $r=n$.

Every binary noise pattern $z^n$ has a unique decomposition into its maximal runs of ones:\vspace{-2pt} $z^n = \bigoplus_{j=1}^{m}\mathbf 1_{B_j}$, $B_j=[\ell_j,r_j]$, where the intervals are ordered and separated by at least one zero, $r_j+1<\ell_{j+1}$, $ j=1,\ldots,m-1$. The case $m=0$ corresponds to $z^n=0^n$. For a first-order energy, the separation condition implies the exact additive decomposition\vspace{-2pt}
\begin{equation}\label{eq:multiburst_additivity}
E_{r^n}(z^n)-E_{r^n}(0^n) = \sum_{j=1}^{m}\Gamma_{r^n}(B_j).
\end{equation}
Therefore, ordering noise patterns by $E_{r^n}(z^n)$ is equivalent to ordering feasible collections of separated interval atoms by their total excess energy. In burst-favoring models, a contiguous run can have lower excess energy than the corresponding collection of isolated single-bit errors. The relation to ORBGRAND is only at the level of additive pattern-enumeration costs: ORBGRAND assigns additive costs to individual reliability-ordered bit flips, whereas the present construction assigns additive costs to interval atoms induced by the finite-memory posterior energy.

\subsubsection{Approximate energy ordering}

Let $E_{r^n}$ denote the posterior energy that defines the target ordering, and let $\widehat E_{r^n}$ denote the energy used by an approximate enumerator. The following calculation bounds the distortion of posterior-weight ratios caused by replacing $E_{r^n}$ with $\widehat E_{r^n}$. For $\delta>0$, let $\mathcal T_n(\delta;r^n)\subseteq\{0,1\}^n$ be a set of noise effects on which the energy approximation error is controlled. One possible choice is\vspace{-2pt}
\begin{equation}
\mathcal T_n(\delta;r^n) = \left\{ z^n:
\left| -\frac1n\log \pi_{r^n}(z^n)-h_n(r^n) \right|
\le \delta \right\},
\end{equation}
where $h_n(r^n)$ is a prescribed centering function. If the random posterior mass $\Pi_n(\delta;R^n) \coloneqq \sum_{z^n\in\mathcal T_n(\delta;R^n)} \pi_{R^n}(z^n)$ satisfies $\Pi_n(\delta;R^n)\to1$ in probability, then $\mathcal T_n(\delta;R^n)$ has asymptotically full posterior mass.  Assume that, on $\mathcal T_n(\delta;r^n)$,\vspace{-2pt}
\begin{equation}\label{eq:energy_approx_uniform}
\sup_{z^n\in\mathcal T_n(\delta;r^n)} \frac1n
\left| \widehat E_{r^n}(z^n)-E_{r^n}(z^n) \right|
\le \xi_n .
\end{equation}
Define\vspace{-2pt}
\begin{equation}
\widehat\pi_{r^n}(z^n) \coloneqq
\frac{\exp\{-\widehat E_{r^n}(z^n)\}}{\sum_{u^n\in\{0,1\}^n}\exp\{-\widehat E_{r^n}(u^n)\}} .
\end{equation}
Then, for any $u^n,v^n\in\mathcal T_n(\delta;r^n)$,
\begin{equation}\label{eq:ratio_distortion}
\frac1n \left| \log \frac{\pi_{r^n}(u^n)/\pi_{r^n}(v^n)}{\widehat\pi_{r^n}(u^n)/\widehat\pi_{r^n}(v^n)} \right|
\le 2\xi_n .
\end{equation}
Indeed, the normalizing constants cancel in the ratio, and the left-hand side equals
\begin{equation}
\frac1n \left| \bigl(\widehat E_{r^n}(u^n)-E_{r^n}(u^n)\bigr) - \bigl(\widehat E_{r^n}(v^n)-E_{r^n}(v^n)\bigr) \right|.
\end{equation}
Thus, for pairs in $\mathcal T_n(\delta;r^n)$, the ratio $\pi_{r^n}(u^n)/\pi_{r^n}(v^n)$ is changed by at most the multiplicative factor $\exp\{2n\xi_n\}$. This factor is subexponential in $n$ when $\xi_n=o(1)$.

The bound in \eqref{eq:ratio_distortion} controls posterior-weight ratios only. It does not imply that the exact ordering induced by $E_{r^n}$ is preserved. Pairwise order preservation on $\mathcal T_n(\delta;r^n)$ would require the relevant energy gaps under $E_{r^n}$ to exceed the corresponding approximation errors. An error exponent analysis for an approximate enumerator would also require a bound
on the posterior mass outside $\mathcal T_n(\delta;r^n)$, together with a comparison between the decoding error bounds induced by $E_{r^n}$ and by $\widehat E_{r^n}$.

\vspace{-5pt}

\subsection{Computation of posterior weights and tails}

The full-space normalized weights $\pi_{r^n}(z^n)$ and the ambient tail mass $T_q$ are computed from the finite-memory posterior energy. Their computation uses the partition function and the queried noise effects, and does not require explicit summation over all $2^n$ noise effects. For the first-order model in \eqref{eq:first_order_energy}, the partition function $\mathsf Z(r^n)$ is computed by the forward recursion in \eqref{eq:forward_recursion}--\eqref{eq:partition_forward}. Hence, once a queried pattern $z^{(j)}$ has been generated, its full-space normalized weight is
\begin{equation}
w_j = \pi_{r^n}(z^{(j)}) = \frac{\exp\{-E_{r^n}(z^{(j)})\}}{\mathsf Z(r^n)} .
\end{equation}
The decoder maintains the accumulated full-space mass of the queried prefix, $B_q^{\rm amb} = \sum_{j=1}^q w_j$. The ambient tail mass is then $T_q=1-B_q^{\rm amb}$. Thus, $T_q$ is exact under the assumed finite-memory posterior model, provided that $\mathsf Z(r^n)$ and the queried weights are computed from the same energy $E_{r^n}$, and that $\mathcal A_q$ contains the distinct noise effects actually queried. No explicit summation over $\{0,1\}^n\setminus\mathcal A_q$ is required.
For a memory-$m$ posterior energy, the same computation is performed on a finite-state trellis whose state records the previous $m$ noise symbols. The forward recursion has $2^m$ states. A generic dense transition implementation has cost $\mathcal{O}(n2^{2m})$, whereas the binary shift-register structure has $\mathcal{O}(n2^{m+1})$ transitions. Therefore, for fixed memory order $m$, the partition function is computable in time linear in $n$. The weights of queried patterns are then obtained by evaluating, or incrementally updating, their energies and normalizing by $\mathsf Z(r^n)$.

The codebook-restricted unqueried denominator contribution is\vspace{-3pt}
\begin{equation}
U_q = \sum_{z^n\notin \mathcal A_q} \pi_{r^n}(z^n) \mathbf 1\{y^n\oplus z^n\in C_n\}.
\end{equation}
Note that this quantity cannot be computed from the posterior energy alone using only codebook membership queries, because it also depends on which unqueried candidate words $y^n\oplus z^n$ are codewords. Computing $U_q$ exactly would require, in general, testing the remaining unqueried memberships or exploiting additional code-specific structure. Under the fixed-size random-codebook model, we estimate $U_q$ by
\begin{equation}
\widehat U_q = p_qT_q,
\qquad
p_q= \frac{M-|\mathcal L_q|}{2^n-q}, \qquad q<2^n.
\end{equation}
Thus, the full-space normalized weights $w_j$ and the ambient tail mass $T_q$ are computed exactly under the assumed finite-memory posterior model. The codebook-restricted unqueried denominator $U_q$ is not computed exactly from the ambient posterior energy and the queried memberships alone. Under the fixed-size random-codebook ensemble, $\widehat U_q=p_qT_q$ is its conditional-mean estimator given the current query transcript. In our random-linear-code experiments, this quantity is used as the corresponding random-codebook occupancy approximation.

\vspace{-3pt}

\subsection{Tail-calibrated stopping}

Let $\eta\in(0,1)$ be a prescribed tolerance, and let $1 \le q_{\max}\le 2^n$ be an implementation-imposed query limit. For $q<2^n$, define the plug-in missing-list estimate\vspace{-3pt}
\begin{equation}\label{eq:pmiss_hat}
\widehat P_{\rm miss}(q) = \frac{\widehat U_q}{S_q+\widehat U_q}.
\end{equation}
The decoder stops at the first $q\le q_{\max}$ for which $|\mathcal L_q|>0$, $S_q+\widehat U_q>0$, and $\widehat P_{\rm miss}(q)\le\eta$, if such a $q$ occurs. Otherwise, it stops when the query limit $q_{\max}$ is reached.
The statistic $\widehat P_{\rm miss}(q)$ is obtained by replacing the unqueried codebook-restricted denominator contribution $U_q$ with its random-codebook estimate $\widehat U_q=p_qT_q$:\vspace{-2pt}
\begin{equation}
\widehat P_{\rm miss}(q) = \widehat P_q\{X^n\notin\mathcal L_q\mid R^n=r^n,C_n\}
= \frac{\widehat U_q}{S_q+\widehat U_q}.
\end{equation}
Thus, unlike fixed-query abandonment rules such as GRANDAB, the proposed rule stops according to a plug-in estimate of the codebook-restricted posterior probability that the transmitted codeword is absent from the current list.

\vspace{-7pt}

\subsection{Bitwise soft output}
\label{subsec:bitwise_soft_output}
 
The decoder also produces bitwise APP estimates. For $a\in\{0,1\}$, define the unqueried ambient tail mass associated with codeword bit value $a$ at position $i$ by
\begin{equation}\label{eq:tail_bit_mass}
T_{q,i}(a) = \sum_{z^n\notin \mathcal A_q} \mathbf 1\{y_i\oplus z_i=a\}\pi_{r^n}(z^n).
\end{equation}
For each $i$, $T_q=T_{q,i}(0)+T_{q,i}(1)$. When $T_q>0$, define the tail-conditioned bit marginal $\mu_{q,i}(a) =  T_{q,i}(a) /T_q$. For the first-order model in \eqref{eq:first_order_energy}, the full-space bit mass
\begin{equation}
T_{0,i}(a) = \sum_{z^n\in\{0,1\}^n} \mathbf 1\{y_i\oplus z_i=a\}\pi_{r^n}(z^n)
\end{equation}
is computed by a finite-state forward--backward recursion \cite{bahl1974optimal}. The unqueried tail mass $T_{q,i}(a)$ is then obtained by subtracting the queried-prefix contribution:\vspace{-3pt}
\begin{equation}
T_{q,i}(a) = T_{0,i}(a) - \sum_{j=1}^{q} \mathbf 1\{y_i\oplus z_i^{(j)}=a\}\pi_{r^n}(z^{(j)}).
\end{equation}
For memory order $m$, the same computation is performed on the $2^m$-state trellis, followed by subtraction of the queried-prefix contributions. Hence the full-space bit masses, and therefore the unqueried tail bit masses, are computed without explicit summation over $\{0,1\}^n$.

Define the listed codebook-restricted bit contribution $S_{q,i}(a) = \sum_{\substack{j\in J_q:\, c^{(j)}_i=a}} \pi_{r^n}(z^{(j)})$. The corresponding unqueried codebook-restricted bit contribution is\vspace{-3pt}
\begin{equation}
U_{q,i}(a) = \!\!\! \sum_{z^n\notin \mathcal A_q}
\!\!\! \mathbf 1\{y_i\oplus z_i=a\} \pi_{r^n}(z^n) \mathbf 1\{y^n\oplus z^n\in C_n\}.
\end{equation}
Under the fixed-size random-codebook model, its conditional-mean estimate given the current query transcript is $\widehat U_{q,i}(a) = p_qT_{q,i}(a)$, $a\in\{0,1\}$. The plug-in bitwise APP estimate is therefore
\begin{equation}\label{eq:bit_app_tail_equiv}
\widehat P_q\{X_i=a\mid R^n=r^n,C_n\} = \frac{S_{q,i}(a)+p_qT_{q,i}(a)}{\widehat D_q},
\, a\in\{0,1\}.
\end{equation}
Equivalently, when $T_q>0$, this estimate can be written as
\begin{equation}\label{eq:bit_app_tail}
\widehat P_q\{X_i=a\mid R^n=r^n,C_n\} =
\frac{S_{q,i}(a)}{\widehat D_q} + \frac{\widehat U_q}{\widehat D_q}\mu_{q,i}(a).
\end{equation}
The two estimated bit probabilities sum to one because $S_{q,i}(0)+S_{q,i}(1)=S_q$ and $p_qT_{q,i}(0)+p_qT_{q,i}(1)=\widehat U_q$.
The corresponding APP log-likelihood ratio is
\begin{equation}\label{eq:app_llr}
L^{\rm APP}_{q,i} = \log \frac{\widehat P_q\{X_i=0\mid R^n=r^n,C_n\}}{\widehat P_q\{X_i=1\mid R^n=r^n,C_n\}},
\end{equation}
where the LLR convention is $0$-over-$1$. If one of the two estimated bit probabilities is zero, $L^{\rm APP}_{q,i}$ is interpreted as an
extended-real value. In numerical implementations it may be clipped to a prescribed finite interval.

\begin{remark}[Extrinsic information]
The APP LLR in \eqref{eq:app_llr} is computed without a priori input LLRs. If a priori LLRs are supplied to the component decoder, they must be included in the candidate metric before extrinsic information is formed. With the $0$-over-$1$ convention, $L_i^{\rm A} = \log\frac{P_{\rm A}\{X_i=0\}}{P_{\rm A}\{X_i=1\}}$, the candidate-codeword prior contributes, up to an additive constant, $\sum_{\ell=1}^n L_\ell^{\rm A} c_\ell$, $c^n=y^n\oplus z^n$, to the energy of candidate $z^n$. If $L^{\rm APP}_{q,i}(L^{\rm A})$ denotes the APP LLR computed with this augmented metric, then the extrinsic LLR is $L^{\rm E}_{q,i} = L^{\rm APP}_{q,i }(L^{\rm A})-L_i^{\rm A}$, where all LLRs use the same $0$-over-$1$ convention.
\end{remark}

\vspace{-8pt}

\subsection{Decoder procedure}
\label{subsec:decoder_procedure}

Algorithm~\ref{alg:tail_calibrated_sogrand} summarizes the tail-stopping Tail-Calibrated SOGRAND procedure after the posterior-tail quantities, bitwise soft-output quantities, and admissible noise-effect enumerators have been defined. The first-hit and list-output variants use the same posterior-weight computation, membership-query loop, and selected enumerator, but replace the tail-stopping condition by the stopping rules specified in Remark~\ref{rem:decoding_modes}. The selected enumerator is any procedure that returns distinct noise effects in the prescribed order. When an exact energy-ordered enumerator is used, this order is nondecreasing in $E_{r^n}$, with ties resolved by the deterministic rule used in \eqref{eq:energy_ordering}.

\begin{algorithm}[!t]
\small
\caption{Tail-Calibrated SOGRAND for Finite-Memory Posteriors}
\label{alg:tail_calibrated_sogrand}
\begin{algorithmic}[1]
\Require Soft observation $r^n$; hard-decision vector $y^n$;
codebook membership oracle $\Phi_n:\{0,1\}^n\to\{0,1\}$;
posterior energy $E_{r^n}$; distinct-noise-effect enumerator;
codebook size $M=|C_n|$; tolerance $\eta\in(0,1)$;
query limit $1 \le q_{\max}<2^n$.
\State Compute $\mathsf Z(r^n)$ by the finite-state forward recursion.
\State Initialize $S\leftarrow0$, $B^{\rm amb}\leftarrow0$,
$\widehat U\leftarrow0$, $\mathcal L\leftarrow[\,]$, and
$\mathcal Q\leftarrow[\,]$.
\State Set $q_\star\leftarrow0$.
\For{$q=1,\ldots,q_{\max}$}
    \State Generate the next distinct noise effect $z^{(q)}$ using the
    selected enumerator.
    \State Compute $w_q \leftarrow
    \exp\{-E_{r^n}(z^{(q)})\}/\mathsf Z(r^n)$.
    \State $B^{\rm amb}\leftarrow B^{\rm amb}+w_q$.
    \State $c^{(q)}\leftarrow y^n\oplus z^{(q)}$.
    \State Append $(z^{(q)},c^{(q)},w_q)$ to $\mathcal Q$.
    \If{$\Phi_n(c^{(q)})=1$}
        \State Append $c^{(q)}$ to $\mathcal L$.
        \State $S\leftarrow S+w_q$.
    \EndIf
    \State $T\leftarrow 1-B^{\rm amb}$.
    \State $p\leftarrow (M-|\mathcal L|)/(2^n-q)$.
    \State $\widehat U\leftarrow pT$.
    \State $q_\star\leftarrow q$.
    \If{$|\mathcal L|>0$, $S+\widehat U>0$, and
    $\widehat U/(S+\widehat U)\le\eta$}
        \State \textbf{break}
    \EndIf
\EndFor
\State Set $\widehat D\leftarrow S+\widehat U$.
\State Form the listed-codeword APP estimates
$\{\widehat P_{q_\star}\{X^n=c\mid R^n=r^n,C_n\}:c\in\mathcal L\}$
and the missing-list estimate
$\widehat P_{\rm miss}(q_\star)$ using
\eqref{eq:list_app_estimate}--\eqref{eq:missing_list_estimate}.
\If{bitwise APP estimates or APP LLRs are required}
    \State Compute the full-space bit masses $T_{0,i}(a)$,
    $i=1,\ldots,n$, $a\in\{0,1\}$, by the finite-state
    forward--backward recursion.
    \State Set $T_{q_\star,i}(a) \leftarrow
    T_{0,i}(a)-\sum_{(z,c,w)\in\mathcal Q}\mathbf 1\{c_i=a\}w$.
    \State Form the bitwise APP estimates using
    \eqref{eq:bit_app_tail_equiv} and the APP LLRs using
    \eqref{eq:app_llr}.
\EndIf
\end{algorithmic}
\end{algorithm}

If the selected enumerator generates distinct noise effects in the nondecreasing-energy order in \eqref{eq:energy_ordering}, then the first-hit version of Algorithm~\ref{alg:tail_calibrated_sogrand} uses the exact posterior-energy ordering induced by the finite-memory posterior model. If an approximate enumerator is used instead, the algorithm remains well defined, but the generated prefix need not coincide with the exact energy-ordered prefix.

\begin{remark}[Decoding modes]
\label{rem:decoding_modes}
The same ordered membership-query sequence is used with different stopping rules. The first-hit rule stops at 
\[
q_{\rm FH} = \inf\{q:\mathcal L_q\neq\emptyset\}\wedge q_{\max}
\]
and returns the first codeword found by the membership-query sequence if $\mathcal L_{q_{\rm FH}}\neq\emptyset$; otherwise it declares a failure. This is the rule used for the main BLER and membership-query comparisons. The list-output rule stops at
\[
q_L = \inf\{q:\ |\mathcal L_q|\ge L\}\wedge q_{\max},
\]
and is used for missing-list calibration with a prescribed list size
$L$. The tail-stopping rule stops at
\[
q_{\eta} \! = \! \inf\left\{ q: |\mathcal L_q|>0,   S_q+\widehat U_q>0,  
\frac{\widehat U_q}{S_q+\widehat U_q}\le\eta\right\}
\wedge q_{\max}.
\]
In our numerical section, this rule is reported as the Exact-Markov-tailStop variant. In all hard-decision summaries reported here, the decoded word is the first listed codeword. Additional queries in the tail-stopping variant are used to estimate the missing-list probability and listed-codeword APPs; they do not change the first-hit hard decision.
\end{remark}

\vspace{-8pt}

\section{Theoretical Guarantees}
\label{sec:theoretical_guarantees}

\vspace{-4pt}

\begin{theorem}[ML property under exact enumeration]\label{thm:ml_property}
Fix a received vector $r^n$, and let $y^n$ be the corresponding hard-decision vector. Assume that the full-space normalized weights satisfy $\pi_{r^n}(z^n)= \exp\{-E_{r^n}(z^n)\} / \mathsf Z(r^n)$  with strictly positive likelihood weights, and that the query order satisfies \eqref{eq:energy_ordering}. Suppose that the decoder continues querying until it first encounters a codeword. Define
\begin{equation}
\tau = \min\{j\in\{1,\ldots,2^n\}:\ y^n\oplus z^{(j)}\in C_n\}.
\end{equation}
Then $\widehat x = y^n\oplus z^{(\tau)}$ is an ML codeword for the likelihood $p_{R^n|X^n}$; that is,
\begin{equation}
\widehat x\in \arg\max_{c^n\in C_n} p_{R^n|X^n}(r^n\mid X^n = c^n).
\end{equation}
\end{theorem}

\vspace{-2pt}

\begin{proof}
For each $c^n\in C_n$, let $z^n(c^n)=y^n\oplus c^n$. By the definitions of $\Lambda_{r^n}$ and $\pi_{r^n}$, $p_{R^n|X^n}(r^n\mid X^n = c^n) =\Lambda_{r^n}(z^n(c^n))$, $\pi_{r^n}(z^n(c^n)) = \Lambda_{r^n}(z^n(c^n)) /(\sum_{u^n\in\{0,1\}^n}\Lambda_{r^n}(u^n))$. The denominator in the second expression is independent of $c^n$. Hence maximizing $p_{R^n|X^n}(r^n\mid c^n)$ over $C_n$ is equivalent to maximizing $\pi_{r^n}(z^n(c^n))$ over $c^n\in C_n$.
Since $\pi_{r^n}(z^n) = \exp\{-E_{r^n}(z^n)\} /\mathsf Z(r^n)$, the ordering in \eqref{eq:energy_ordering} is an ordering of noise effects in nonincreasing $\pi_{r^n}$. Therefore, the first queried noise effect whose associated candidate word is in $C_n$ maximizes $\pi_{r^n}(z^n(c^n))$ among all $c^n\in C_n$. Thus $y^n\oplus z^{(\tau)}$ maximizes $p_{R^n|X^n}(r^n\mid c^n)$ over $C_n$, proving the claim.
\end{proof}

\vspace{-5pt}

\begin{theorem}[Unqueried denominator under fixed-size random coding]
\label{thm:unqueried_denominator}
Fix $r^n$, $y^n$, and a queried set $\mathcal A_q\subseteq\{0,1\}^n$ with $|\mathcal A_q|=q<2^n$. Let $C_n$ be uniformly distributed over all $M$-subsets of $\{0,1\}^n$, and define $\mathcal B_q=\{y^n\oplus z^n:\ z^n\in\mathcal A_q\}$. Let $\ell_q\subseteq\mathcal B_q$, and define the transcript event $\mathcal E_q(\ell_q)=\{C_n\cap\mathcal B_q=\ell_q\}$. Assume that $\mathbb P_C(\mathcal E_q(\ell_q))>0$. On this event, the realized list satisfies $\mathcal L_q=\ell_q$. Let $N_q=2^n-q$,  $K_q=M-|\ell_q|$,  $p_q=\frac{K_q}{N_q}$. Then the unqueried codebook-restricted denominator contribution $U_q = \sum_{z^n\notin \mathcal A_q} \pi_{r^n}(z^n) \mathbf 1\{y^n\oplus z^n\in C_n\}$ satisfies $\mathbb E_C[U_q\mid \mathcal E_q(\ell_q)] = p_qT_q$, where $T_q = \sum_{z^n\notin\mathcal A_q}\pi_{r^n}(z^n)$. If $N_q>1$, then
\begin{align}
\operatorname{Var}_C(U_q\mid \mathcal E_q(\ell_q))
& = \frac{K_q(N_q-K_q)}{N_q(N_q-1)}
\left( \sum_{z^n\notin\mathcal A_q}\pi_{r^n}(z^n)^2 - \frac{T_q^2}{N_q} \right)  \nonumber\\
& \le p_q(1-p_q) \sum_{z^n\notin\mathcal A_q}\pi_{r^n}(z^n)^2 .
\label{eq:Uq_variance}
\end{align}
For $N_q=1$, the conditional variance is zero. If $p_q>0$ and $T_q>0$, then, with
\begin{equation}\label{eq:Neff_def}
N_{\rm eff}(q) =  \frac{T_q^2}{\sum_{z^n\notin \mathcal A_q}\pi_{r^n}(z^n)^2},
\end{equation}
we have, for every $\epsilon>0$,
\begin{equation}\label{eq:Uq_concentration}
\mathbb P_C\left( \left|U_q-p_qT_q\right| \ge \epsilon p_qT_q \mid \mathcal E_q(\ell_q) \right)
\le \frac{1-p_q}{\epsilon^2p_qN_{\rm eff}(q)}.
\end{equation}
\end{theorem}

\vspace{-10pt}

\begin{proof}
Conditional on $\mathcal E_q(\ell_q)$, the remaining $K_q=M-|\ell_q|$ codewords are uniformly distributed over the $N_q=2^n-q$ unqueried candidate words $\{y^n\oplus z^n:\ z^n\notin\mathcal A_q\}$. Therefore, for each $z^n\notin\mathcal A_q$,
\begin{equation}
\mathbb P_C\{y^n\oplus z^n\in C_n\mid \mathcal E_q(\ell_q)\} = \frac{K_q}{N_q} = p_q .
\end{equation}
Linearity of expectation gives
\begin{align}
\mathbb E_C[U_q\mid \mathcal E_q(\ell_q)] & =
\sum_{z^n\notin\mathcal A_q} \pi_{r^n}(z^n)
\mathbb P_C\{y^n\oplus z^n\in C_n\mid \mathcal E_q(\ell_q)\} \nonumber\\
& = p_q \sum_{z^n\notin\mathcal A_q}\pi_{r^n}(z^n) = p_qT_q .
\end{align}

For each $z^n\notin\mathcal A_q$, define $I_{z^n} = \mathbf 1\{y^n\oplus z^n\in C_n\}$. Conditional on $\mathcal E_q(\ell_q)$, the variables $\{I_{z^n}:z^n\notin\mathcal A_q\}$ are the inclusion indicators of a uniformly chosen $K_q$-subset of an $N_q$-element set. Hence, for $N_q>1$,
\begin{equation}
\operatorname{Var}_C(I_{z^n}\mid \mathcal E_q(\ell_q)) = p_q(1-p_q),
\end{equation}
and, for distinct $z^n,u^n\notin\mathcal A_q$,
\begin{equation}
\operatorname{Cov}_C(I_{z^n},I_{u^n}\mid \mathcal E_q(\ell_q)) = -\frac{p_q(1-p_q)}{N_q-1}.
\end{equation}
Since $U_q = \sum_{z^n\notin\mathcal A_q} \pi_{r^n}(z^n)I_{z^n}$, substitution of these variances and covariances gives
\begin{align}
\operatorname{Var}_C(U_q\mid \mathcal E_q(\ell_q))
& = p_q(1-p_q) \sum_{z^n\notin\mathcal A_q}\pi_{r^n}(z^n)^2 \nonumber\\
& \quad - \frac{p_q(1-p_q)}{N_q-1}
\!\!\!\! \sum_{\substack{z^n,u^n\notin\mathcal A_q\\ z^n\ne u^n}} \!\!\!
\pi_{r^n}(z^n)\pi_{r^n}(u^n)
\nonumber\\
& = \frac{K_q(N_q-K_q)}{N_q(N_q-1)} \!
\left( \! \sum_{z^n\notin\mathcal A_q} \!\!\! \pi_{r^n}(z^n)^2 \!-\!  \frac{T_q^2}{N_q} \right).
\end{align}
Because the weights $\pi_{r^n}(z^n)$ are nonnegative, $T_q^2 \ge \sum_{z^n\notin\mathcal A_q}\pi_{r^n}(z^n)^2$. Therefore,
\begin{equation}
\operatorname{Var}_C(U_q\mid \mathcal E_q(\ell_q)) \le
p_q(1-p_q) \sum_{z^n\notin\mathcal A_q}\pi_{r^n}(z^n)^2 .
\end{equation}
When $N_q=1$, the remaining codebook membership is deterministic under the conditioning, so the conditional variance is zero. Finally, Chebyshev's inequality gives
\begin{align}
&\mathbb P_C\left( \left|U_q-p_qT_q\right| \ge \epsilon p_qT_q \mid \mathcal E_q(\ell_q) \right) \nonumber\\
&\qquad\le \frac{ p_q(1-p_q)\sum_{z^n\notin\mathcal A_q}\pi_{r^n}(z^n)^2}{\epsilon^2p_q^2T_q^2 }
= \frac{1-p_q}{\epsilon^2p_qN_{\rm eff}(q)}.
\end{align}
\end{proof}

\vspace{-14pt}

\begin{remark}
For each fixed query index (q), the identity $\mathbb E_C[U_q\mid \mathcal E_q(\ell_q)] = p_qT_q$ is an exact conditional-mean statement under the fixed-size random-codebook ensemble. Conditional on \(\mathcal E_q(\ell_q)\), the remaining \(K_q=M-|\ell_q|\) codewords are uniformly distributed over all \(K_q\)-subsets of the \(N_q=2^n-q\) unqueried candidate words. Consequently, each unqueried candidate word has conditional inclusion probability \(p_q=K_q/N_q\). This conditional uniformity does not hold for a random linear code, because its codeword-membership events are coupled by the linear-subspace structure. Accordingly, in our random-linear-code experiments, we use \(p_qT_q\) as the fixed-size random-codebook occupancy approximation to \(U_q\), rather than as its exact conditional mean.
\end{remark}

\vspace{-4pt}

\begin{theorem}[Perturbation bound for the missing-list estimate]
\label{thm:tail_calibrated_reliability}
Fix a query index $q$. Let
\begin{equation}
P_{\rm miss}(q) = \Pr\{X^n\notin \mathcal L_q\mid R^n=r^n,C_n\} = \frac{U_q}{S_q+U_q}
\end{equation}
be the codebook-restricted posterior probability that the transmitted codeword is not in the current list, and let
\begin{equation}
\widehat P_{\rm miss}(q) = \frac{\widehat U_q}{S_q+\widehat U_q}
\end{equation}
be its plug-in estimate. Assume that $\widehat U_q>0$ and that $|U_q-\widehat U_q|/\widehat U_q\le \epsilon$ for some $\epsilon\in(0,1)$. Then
\begin{equation}\label{eq:miss_prob_error}
\left| P_{\rm miss}(q)-\widehat P_{\rm miss}(q) \right|
\le \frac{\epsilon\,\widehat U_q S_q}{(S_q+(1-\epsilon)\widehat U_q)(S_q+\widehat U_q)}.
\end{equation}
Moreover, if the decoder stops at an index $q$ satisfying $\widehat P_{\rm miss}(q)\le\eta$, then
\begin{equation}\label{eq:stopping_true_miss_bound}
P_{\rm miss}(q) \le \frac{(1+\epsilon)\eta}{1+\epsilon\eta}.
\end{equation}
\end{theorem}

\vspace{-4pt}

\begin{proof}
For $u\ge0$, define $f(u)=u /(S_q+u)$. Then
\begin{equation}
f(U_q)-f(\widehat U_q) = \frac{S_q(U_q-\widehat U_q)}{(S_q+U_q)(S_q+\widehat U_q)}.
\end{equation}
The relative-error assumption implies
\begin{equation}
U_q\ge (1-\epsilon)\widehat U_q,
\qquad |U_q-\widehat U_q|\le \epsilon\widehat U_q .
\end{equation}
Therefore,
\begin{equation}
\left| P_{\rm miss}(q)-\widehat P_{\rm miss}(q) \right|
\le \frac{\epsilon\,\widehat U_q S_q}{(S_q+(1-\epsilon)\widehat U_q)(S_q+\widehat U_q)},
\end{equation}
which proves \eqref{eq:miss_prob_error}.
If $\widehat P_{\rm miss}(q)\le\eta$, then
\begin{equation}
\frac{\widehat U_q}{S_q+\widehat U_q}\le\eta \quad\Longleftrightarrow\quad S_q\ge \frac{1-\eta}{\eta}\widehat U_q .
\end{equation}
Using $U_q\le(1+\epsilon)\widehat U_q$, we obtain
\begin{equation}
P_{\rm miss}(q) =
\frac{U_q}{S_q+U_q} \le \frac{(1+\epsilon)\widehat U_q}{\frac{1-\eta}{\eta}\widehat U_q+(1+\epsilon)\widehat U_q}
= \frac{(1+\epsilon)\eta}{1+\epsilon\eta}.
\end{equation}
This proves \eqref{eq:stopping_true_miss_bound}.
\end{proof}

\vspace{-7pt}

\begin{theorem}[Ambient-tail abandonment bound]
\label{thm:posterior_tail_abandonment}
For each received vector $r^n$, let $\pi_{r^n}$ denote the conditional law of the hard-decision noise effect $Z^n$ given $R^n=r^n$ under the probability measure used to define
$P_{\rm e}^{\rm full}(n)$ and $P_{\rm e}^{\rm ab}(n)$. Let full GRAND query noise effects in nonincreasing $\pi_{r^n}$-order without abandonment. Let the abandoned decoder use the same query order and stop after $q_n(r^n)\in\{0,\ldots,2^n\}$ queries, where $q_n$ is a deterministic measurable function of $r^n$. If the abandoned decoder stops without finding a codeword, this outcome is counted as a block error.
Let $\mathcal A_q(r^n)$ denote the first $q$ noise effects in the $\pi_{r^n}$-ordered query list. Suppose that
\begin{equation}
T_{q_n(r^n)}(r^n) = \sum_{z^n\notin \mathcal A_{q_n(r^n)}(r^n)} \pi_{r^n}(z^n)
\le \epsilon_n
\end{equation}
for $P_{R^n}$-almost every $r^n$, where $\epsilon_n\to0$. Then
\begin{equation}\label{eq:abandonment_error_bound}
P_{\rm e}^{\rm ab}(n) \le P_{\rm e}^{\rm full}(n)+\epsilon_n .
\end{equation}
Consequently, if full GRAND has vanishing block error probability for a given rate, code ensemble, and channel model, then any abandonment rule satisfying the above ambient-tail condition also has vanishing block error probability for the same rate, ensemble, and channel model.
In particular, in any setting where the full-GRAND capacity-achievability result applies to binary additive stationary ergodic hard-decision noise
and random codebooks, the abandoned decoder also has vanishing block error probability for every $R_{\rm c}<1-H(\mathbf Z)$, where $\mathbf Z=\{Z_i\}_{i\ge1}$ and $H(\mathbf Z)$ denotes the entropy rate of the hard-decision noise process, measured in bits per channel use.
\end{theorem}

\begin{proof}
By the definition of $\pi_{r^n}$,
\begin{equation}
\Pr\{Z^n\notin \mathcal A_{q_n(r^n)}(r^n)\mid R^n=r^n\} =T_{q_n(r^n)}(r^n).
\end{equation}
On the event $Z^n\in \mathcal A_{q_n(R^n)}(R^n)$, the abandoned decoder performs the same queries as full GRAND up to the first queried codeword. If a codeword is encountered before the realized noise effect, both decoders stop at that same codeword. If no such codeword is encountered earlier, both decoders query the realized noise effect and recover the transmitted codeword. Hence any block error made by the abandoned decoder that is not also made by full GRAND can occur only on the event $\{Z^n\notin \mathcal A_{q_n(R^n)}(R^n)\}$. Therefore,
\begin{equation}
P_{\rm e}^{\rm ab}(n) \le P_{\rm e}^{\rm full}(n) + \Pr\{Z^n\notin \mathcal A_{q_n(R^n)}(R^n)\}.
\end{equation}
Taking expectation over $R^n$,
\begin{equation}
\Pr\{Z^n\notin \mathcal A_{q_n(R^n)}(R^n)\}
= \mathbb E\!\left[ T_{q_n(R^n)}(R^n) \right] \le \epsilon_n .
\end{equation}
This proves \eqref{eq:abandonment_error_bound}.
If $P_{\rm e}^{\rm full}(n)\to0$ and $\epsilon_n\to0$, then $P_{\rm e}^{\rm ab}(n)\to0$. The final statement follows by applying this implication in any setting where full GRAND is known to achieve all rates below $1-H(\mathbf Z)$ for binary additive stationary ergodic hard-decision noise.
\end{proof}

\section{Numerical Evaluation}
\label{sec:numerical_evaluation}

We evaluate Tail-Calibrated SOGRAND through two codebook ensembles. The first is the fixed-size random-codebook ensemble used in the missing-list analysis; this experiment isolates the occupancy estimate $\widehat U_q=p_qT_q$. The second is the binary random-linear-code membership-oracle setting used in GRAND, ORBGRAND, and ORBGRAND-AI comparisons \cite{duffy2019capacity,duffy2022ordered,duffy2023using}. The random-linear-code experiments compare noise-effect orderings without using code-specific decoding structure beyond parity-check membership testing. At each parameter point in those experiments, all noninterleaved compared methods use the same sampled random linear code, the same channel model, and the same membership oracle. Among the compared methods, the random-interleaved ORBGRAND baseline is the only one that changes the coordinate order before the correlated channel acts. It interleaves the coded positions before transmission and deinterleaves the received soft observations before decoding. Apart from this preprocessing step, the methods differ only in the order in which candidate noise effects are generated and, for
Tail-Calibrated SOGRAND, in the ambient-tail and missing-list quantities computed during the search.

 \vspace{-5pt}

\subsection{Fixed-size random-codebook validation}
\label{subsec:fixed_size_random_codebook_validation}

We first evaluate the missing-list estimate in the codebook ensemble used in the occupancy analysis. For each independent codebook draw, $C_n$ is sampled uniformly from all subsets of $\{0,1\}^n$ with cardinality $M=2^k$, and membership is tested by explicit lookup. The channel model, posterior energy, noise-effect ordering, and tail-mass computation are the same as in the Gauss--Markov experiments below. This experiment isolates the random-codebook occupancy step in the estimate $\widehat U_q=p_qT_q$; it is not intended as a scalable representation of large codebooks.

For calibration, let $\widehat p_t$ denote the missing-list estimate reported on frame $t$, and let $e_t=\mathbf{1}\{X_t^n\notin\mathcal L_{q_t}\}$ be the corresponding missing-list event. The empirical missing-list frequency is $N^{-1}\sum_{t=1}^N e_t$. The expected calibration error is computed from nonempty logarithmic bins $\{\mathcal B_b\}$ as follows \cite{guo2017calibration}:
\begin{equation}
\mathsf{ECE} = \sum_{b:\,|\mathcal B_b|>0}
\frac{|\mathcal B_b|}{N}
\left| \frac{1}{|\mathcal B_b|} \sum_{t\in\mathcal B_b}\widehat p_t - \frac{1}{|\mathcal B_b|} \sum_{t\in\mathcal B_b} e_t \right|.
\end{equation}

Table~\ref{tab:fixed_size_calibration} reports the resulting calibration summary. Note that the fixed-size random-codebook experiment directly evaluates the occupancy approximation in the ensemble used by the conditional-mean analysis. Since calibration estimates are unreliable when the number of missing-list events is small, Table~\ref{tab:fixed_size_calibration} also reports the number of observed missing-list events.
For the two operating points with a non-negligible number of missing-list events, the empirical missing-list frequency is close to $\overline{\widehat P}_{\rm miss}$, and the ECE is small relative to the corresponding missing-list probability.  



Fig.~\ref{fig:fixed_size_calibration_reliability} reports the binned reliability diagrams corresponding to the fixed-size random-codebook validation in Table~\ref{tab:fixed_size_calibration}. In each bin, the horizontal coordinate is the binwise average of $\widehat P_{\rm miss}$ and the vertical coordinate is the empirical frequency of the event $X^n\notin\mathcal L_q$. The dashed line is the ideal calibration line $y=x$. Vertical bars show $95\%$ Wilson confidence intervals for the empirical bin frequencies, and the integer next to each point is the number of frames in the bin. Bins with zero observed missing-list events are included in the ECE calculation but are omitted from the log--log plot.

\begin{figure}[!t]
    \centering
    \begin{subfigure}[t]{0.77\linewidth}
        \centering
        \includegraphics[width=\linewidth]{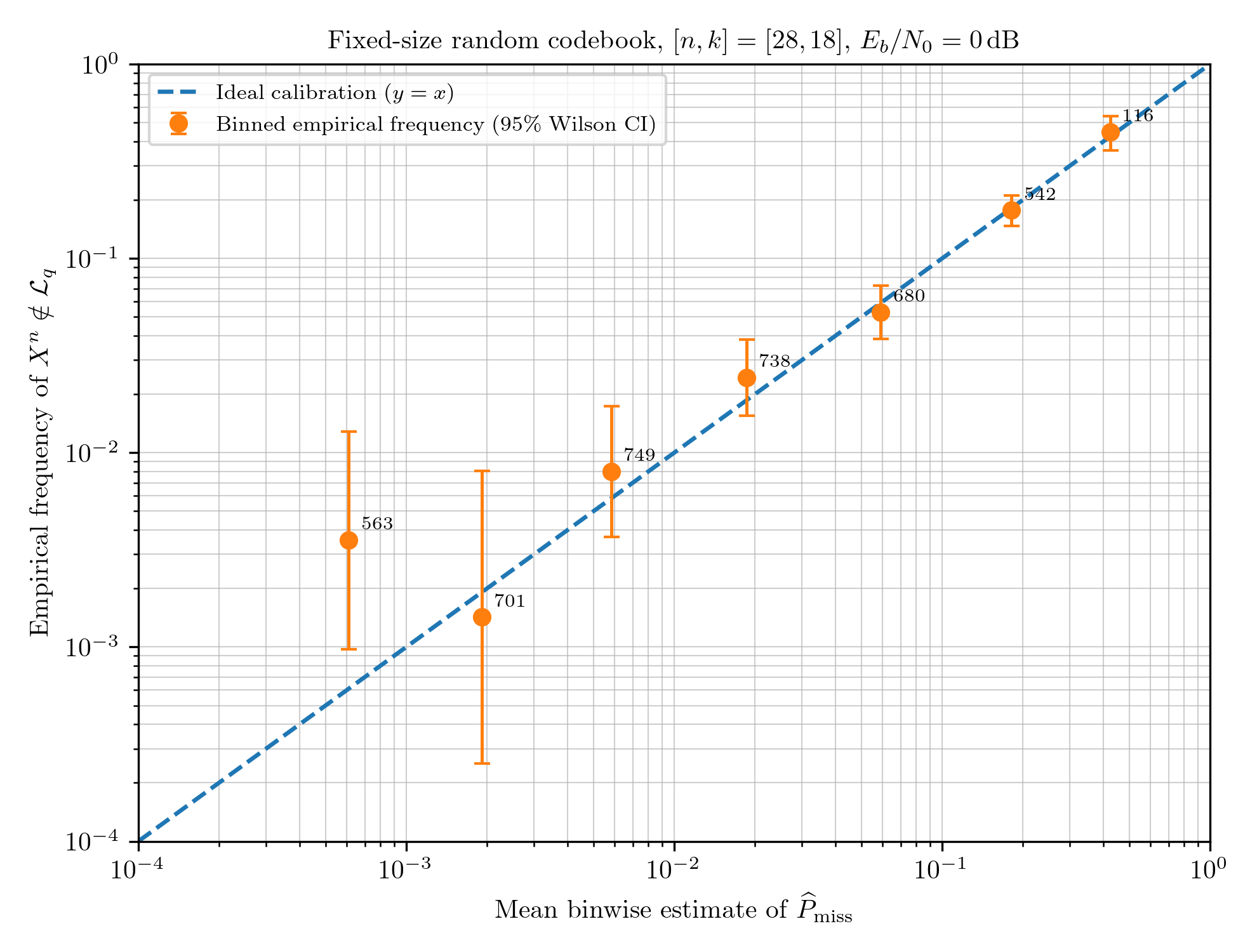}
        \caption{$E_b/N_0=0$ dB.}
        \label{fig:fixed_size_calibration_0db}
    \end{subfigure}
   
    \begin{subfigure}[t]{0.77\linewidth}
        \centering
        \includegraphics[width=\linewidth]{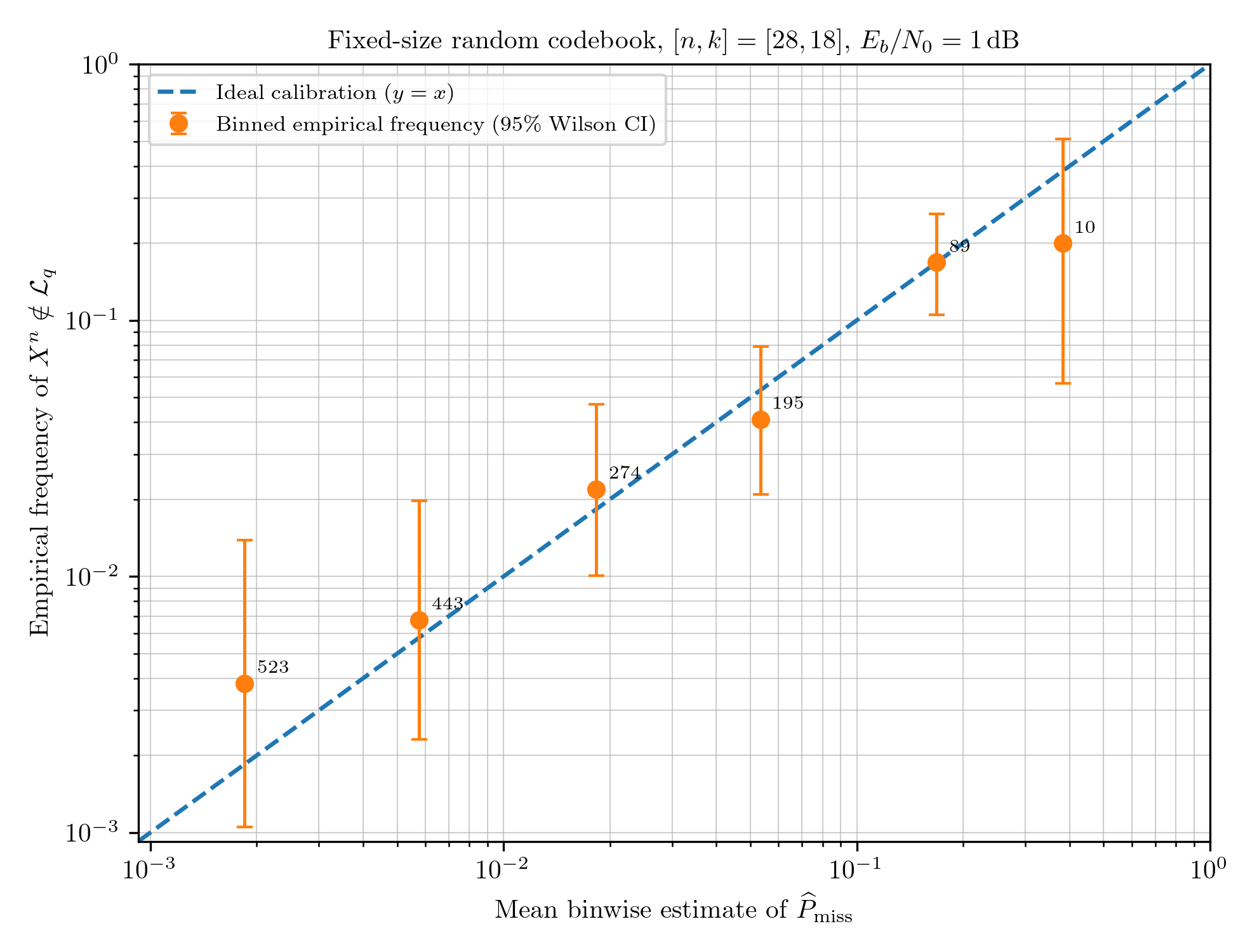}
        \caption{$E_b/N_0=1$ dB.}
        \label{fig:fixed_size_calibration_1db}
    \end{subfigure}
    \caption{Reliability diagrams for the fixed-size random-codebook validation with $[n,k]=[28,18]$, $\rho=0.5$, $q_{\max}=5000$, and list size $4$. Points show binned empirical missing-list frequencies, vertical bars show $95\%$ Wilson confidence intervals, and adjacent integers show bin counts.}
    \label{fig:fixed_size_calibration_reliability}
\end{figure}

     

\begin{table}[h]
\centering
\caption{Fixed-size random-codebook validation of the plug-in missing-list estimate. Codebooks are sampled uniformly from all subsets of $\{0,1\}^n$ with cardinality $2^k$. The experiment uses
$\rho=0.5$, $q_{\max}=5000$, and list size $4$.}
\label{tab:fixed_size_calibration}
\scriptsize
\setlength{\tabcolsep}{3pt}
\resizebox{\linewidth}{!}{
\begin{tabular}{ccccccc}
\toprule
$[n,k]$ & $E_b/N_0$ & Frames & Miss.
& $\overline{\widehat P}_{\rm miss}$
& Emp. miss.
& ECE \\
\midrule
$[28,18]$ & $0$ dB & $5000$ & $211$
& $4.156{\times}10^{-2}$
& $4.220{\times}10^{-2}$
& $3.443{\times}10^{-3}$ \\
$[28,18]$ & $1$ dB & $5000$ & $36$
& $7.690{\times}10^{-3}$
& $7.200{\times}10^{-3}$
& $1.491{\times}10^{-3}$ \\
\bottomrule
\end{tabular}}
\end{table}

\vspace{-2pt}

\begin{figure*}[!t]
\centering
\begin{minipage}{0.32\linewidth}
\centering
\includegraphics[width=\linewidth]{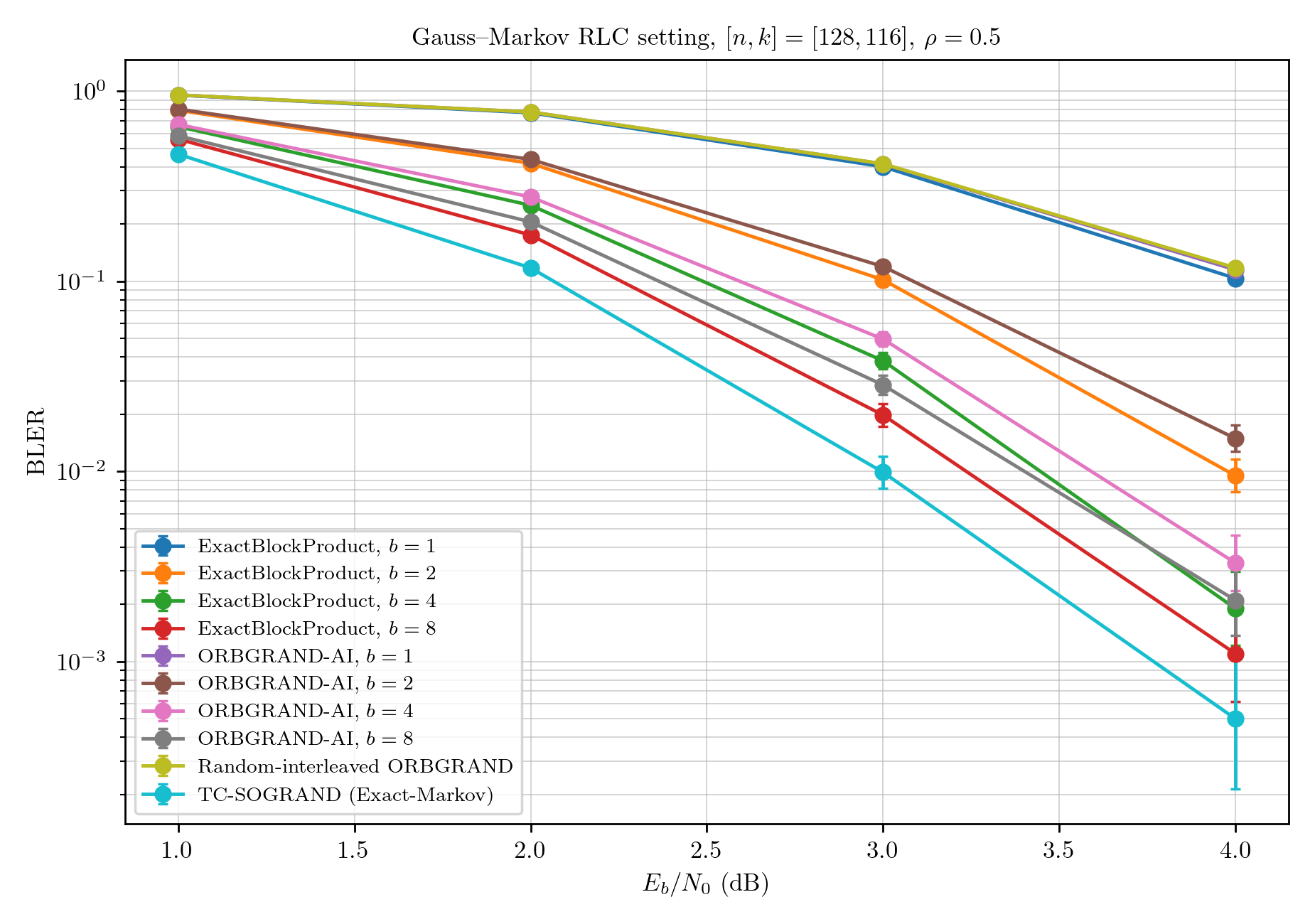}
\end{minipage}
~
\begin{minipage}{0.32\linewidth}
\centering
\includegraphics[width=\linewidth]{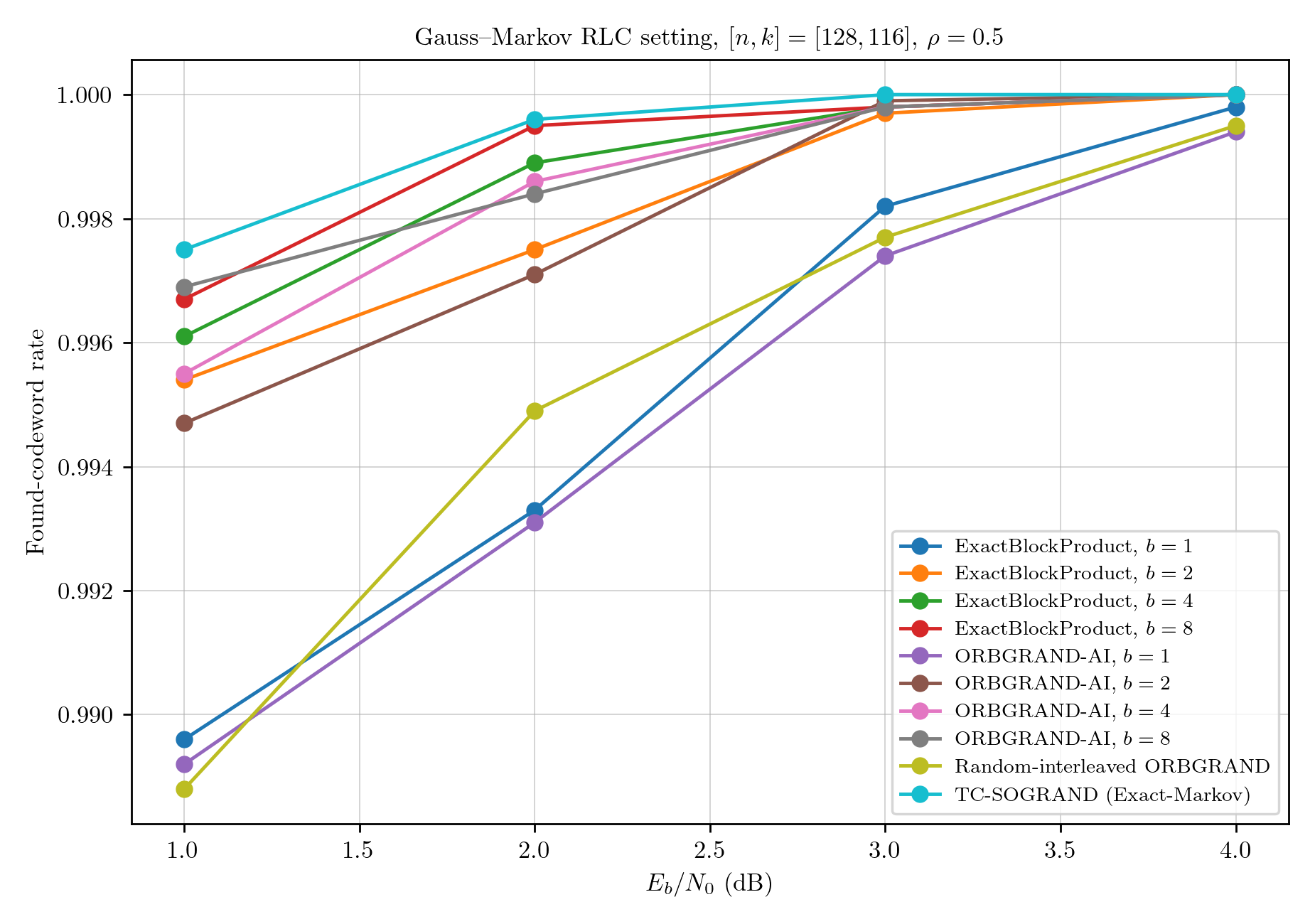}
\end{minipage}
~
\begin{minipage}{0.32\linewidth}
\centering
\includegraphics[width=\linewidth]{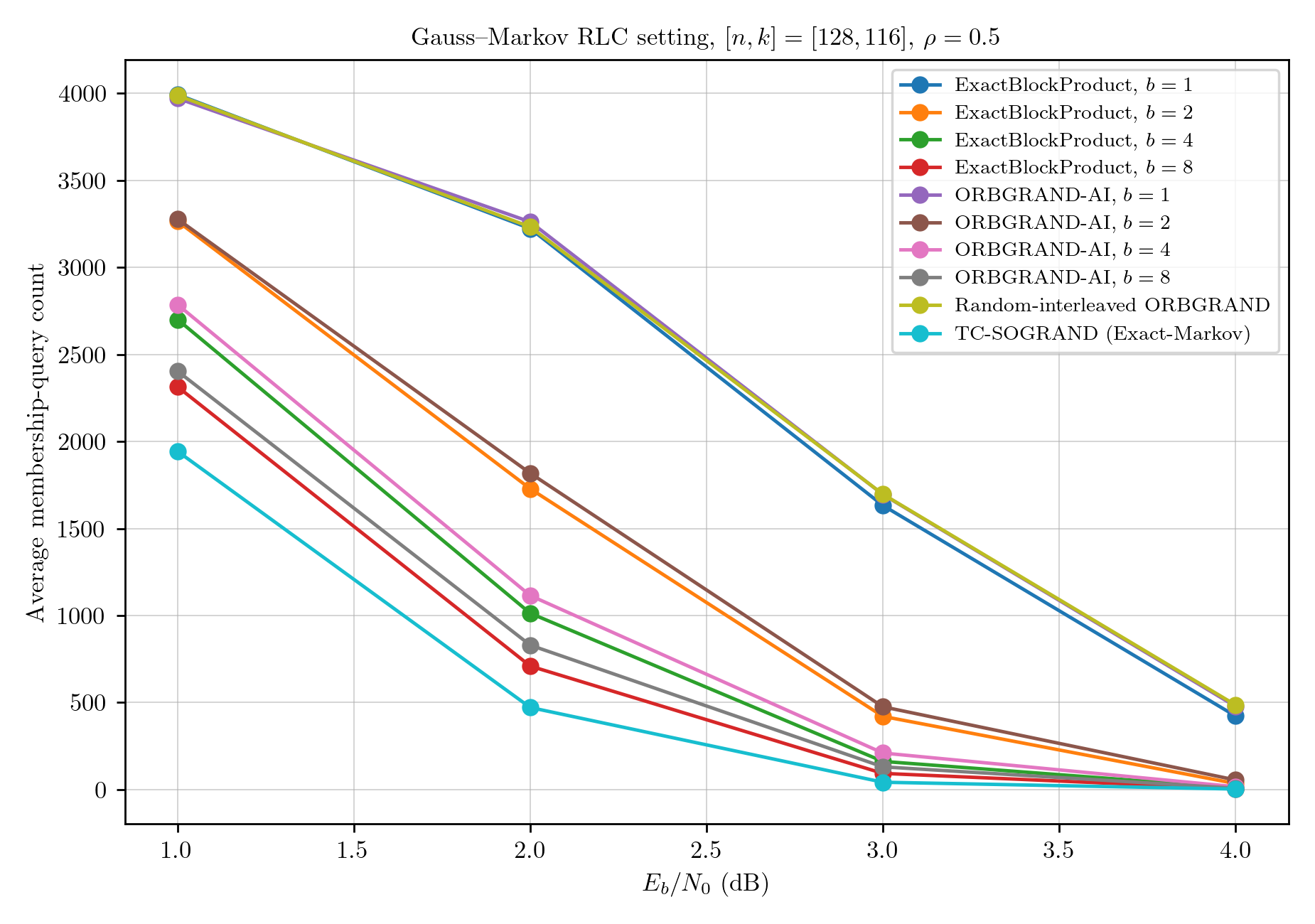}
\end{minipage}

\vspace{1mm}

\begin{minipage}{0.32\linewidth}
\centering
\includegraphics[width=\linewidth]{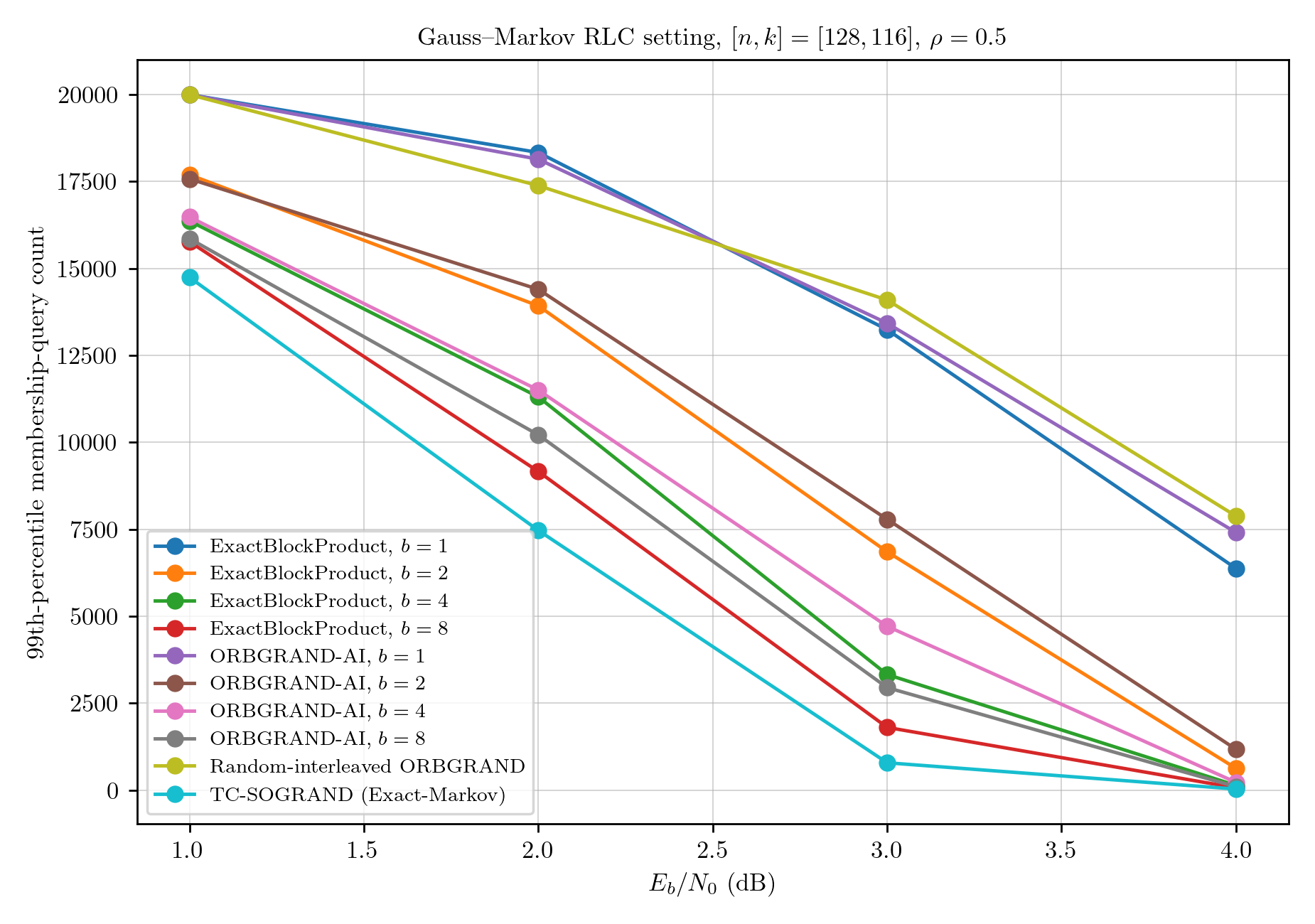}
\end{minipage}
~
\begin{minipage}{0.32\linewidth}
\centering
\includegraphics[width=\linewidth]{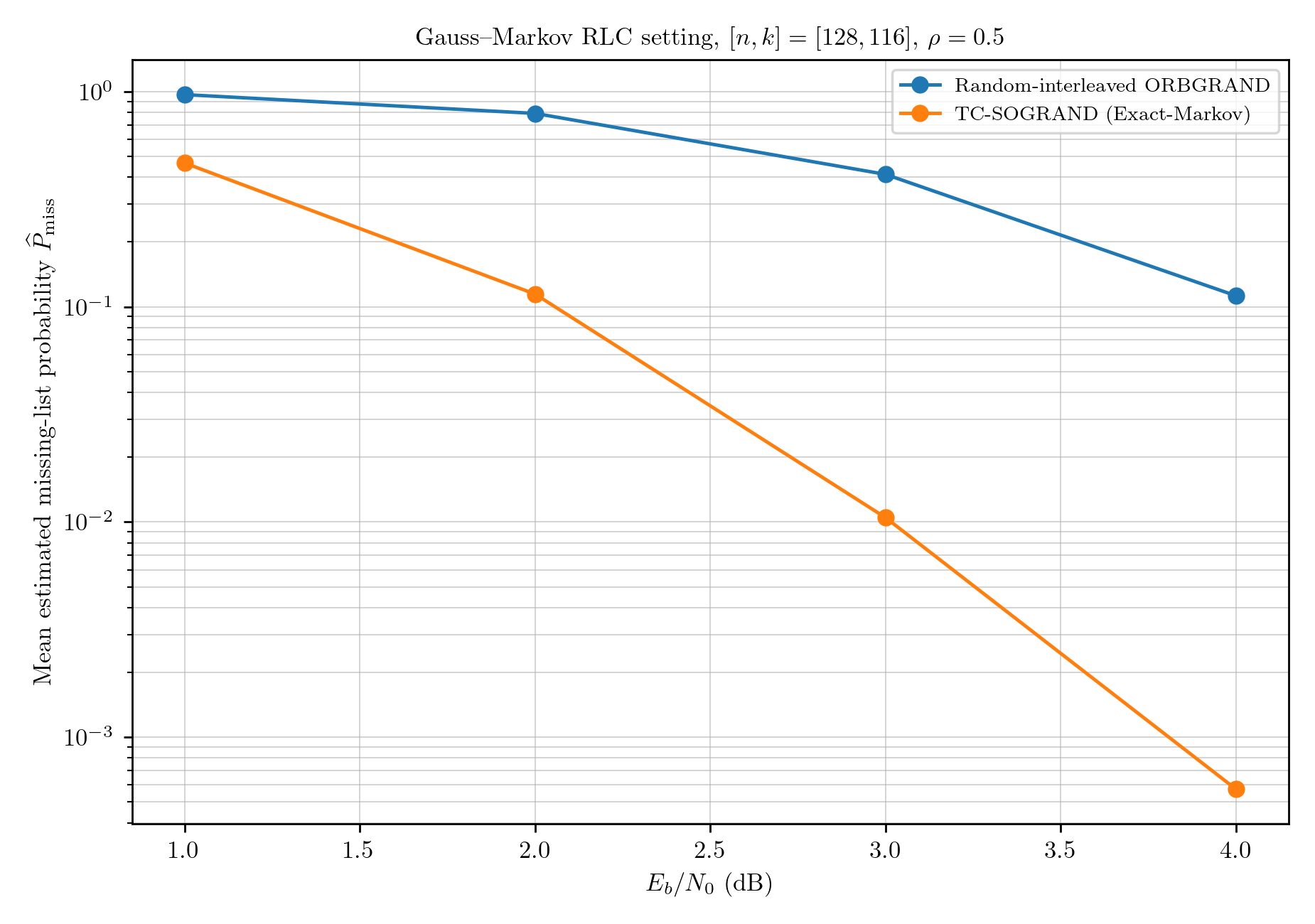}
\end{minipage}
~
\begin{minipage}{0.32\linewidth}
\centering
\includegraphics[width=\linewidth]{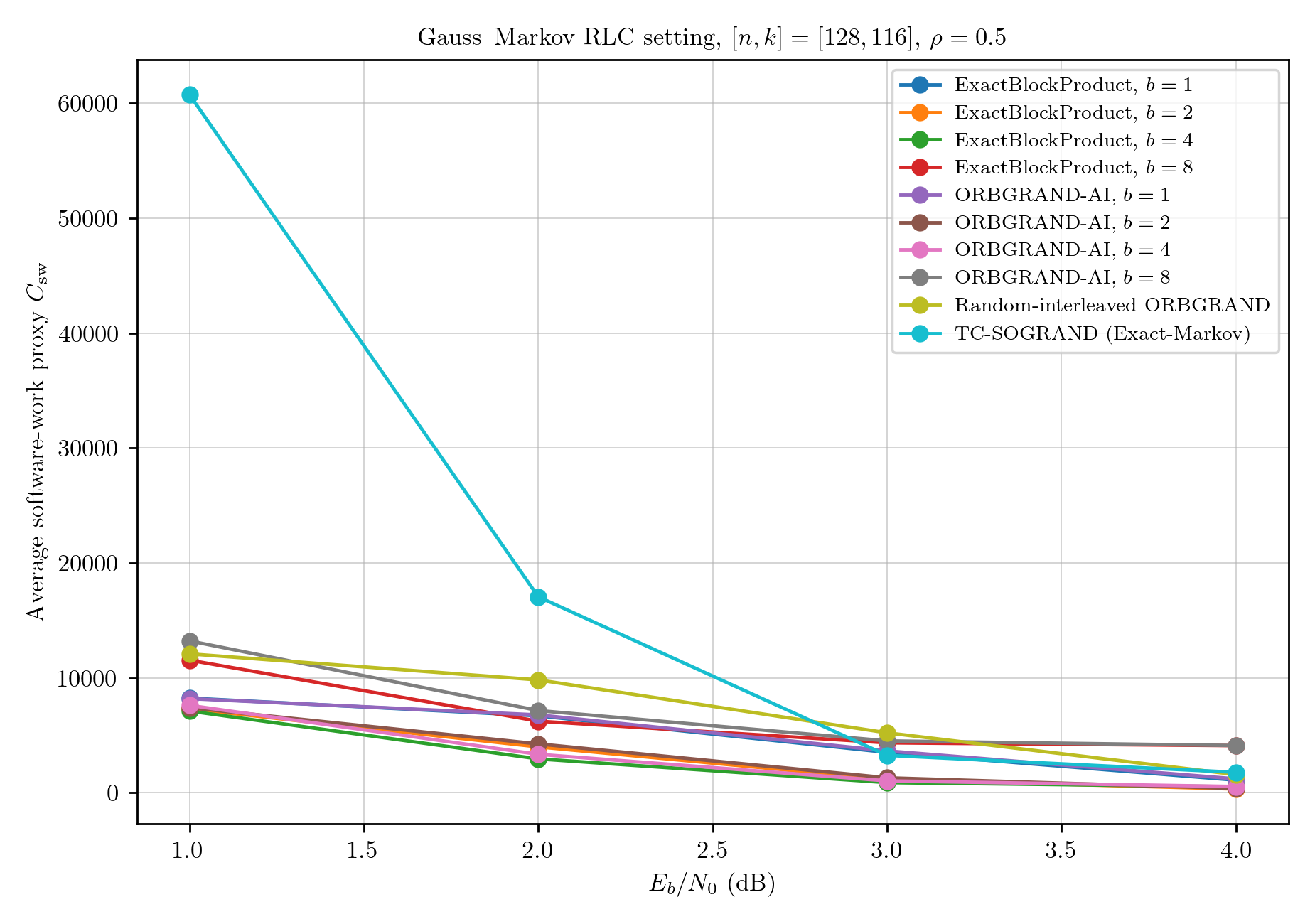}
\end{minipage}
\vspace{-2pt}
\caption{Gauss--Markov RLC experiment with $[n,k]=[128,116]$, $\rho=0.5$, and $q_{\max}=20000$. The tolerance $\eta=10^{-4}$ is active only for the Exact-Markov-tailStop variant and is inactive for first-hit BLER decoding. Top row: BLER, found-codeword rate, and average number of codebook membership queries. Bottom row: 99th-percentile membership-query count, mean estimated missing-list probability, and software-work proxy $C_{\rm sw}$. The mean missing-list panel includes only methods that compute missing-list estimates. All curves except Exact-Markov-tailStop use first-hit stopping. The Exact-Markov-tailStop curve uses the same first listed codeword as its hard decision and continues the search until $\widehat P_{\rm miss}\le10^{-4}$ or $q_{\max}$ is reached.}
\label{fig:gm_ebn0_all}
\end{figure*}

\subsection{Random-linear-code experimental setup}

In the random-linear-code experiments, the code is a binary random linear code (RLC) with parameters $[n,k]=[128,116]$ and $R_{\rm c}=k/n=0.90625$. The codebook is represented by a parity-check membership oracle. Each frame uses a uniformly selected codeword. For the soft-input experiments, BPSK maps bit $0$ to $+1$ and bit $1$ to $-1$. Equivalently, for a binary codeword $X^n$, define the transmitted real symbol $S_i = 1-2X_i$. The received sample is $R_i=S_i+N_i$. The primary soft-input channel is a real stationary Gauss--Markov noise model, $N_1\sim \mathcal N(0,\sigma^2)$, $N_i=\rho N_{i-1}+V_i$, $V_i\sim \mathcal N(0,\sigma^2(1-\rho^2))$, with $|\rho|<1$, where the innovations $\{V_i\}_{i\ge2}$ are mutually independent, independent of $N_1$, and independent of the transmitted codeword. The marginal noise variance is chosen as $\sigma^2 = \frac{1}{2R_{\rm c}10^{E_b/N_0/10}}$, which corresponds to unit-energy BPSK symbols over a real channel. Hard decisions are obtained by thresholding $R_i$ at zero.

Our soft-input comparison uses $\rho=0.5$ and membership-query limit $q_{\max}=20000$. For the random-linear-code experiments, $p_qT_q$ is used as the fixed-size random-codebook occupancy approximation to the unqueried codebook-restricted denominator; its calibration is evaluated empirically. The Exact-Markov first-hit, list-output, and tail-stopping runs use the same posterior weights, ambient-tail computation, membership-query loop, and first-order Gauss--Markov posterior-energy enumerator. They differ only in the stopping rule. The first-hit run stops at the first listed codeword and is used for the main BLER and membership-query comparisons; $\eta$ is inactive in this run. The list-output calibration run stops after the prescribed list size is reached. The Exact-Markov-tailStop run continues after the first listed codeword and stops when $\widehat U_q/(S_q+\widehat U_q)\le10^{-4}$ or when $q_{\max}$ is reached. In all hard-decision summaries, the decoded word is the first listed codeword.

The baselines and ablations are defined as follows. The random-interleaved ORBGRAND baseline uses a fixed random permutation $\pi$ of $\{1,\ldots,n\}$ for each code realization. The transmitted codeword is permuted before the Gauss--Markov channel acts. If $S_i=1-2X_i$ denotes the BPSK symbol in code coordinate $i$, then the interleaved channel output is generated as $\widetilde R_j = S_{\pi(j)}+N_j$, $j=1,\ldots,n$. Before decoding, the received samples are deinterleaved to the original code coordinates, $R_i^{\rm int}=\widetilde R_{\pi^{-1}(i)} = S_i+N_{\pi^{-1}(i)}$. Ordinary memoryless ORBGRAND is then applied to $R^{\rm int}$ and tests membership in the original codebook $C_n$. Thus the codeword coordinates are restored before decoding, whereas the effective noise sequence across code coordinates is a randomly permuted version of the Gauss--Markov noise realization. This baseline is a system-level random-interleaving reference; it is not a decoder applied to the same code-coordinate correlated observation used by the noninterleaved methods.
The ORBGRAND-AI baseline follows the nonoverlapping-block approximate-independence construction in the BPSK RLC setting \cite{duffy2023using}. It uses local blocks of length $b\in\{1,2,4,8\}$, computes block-substitution reliabilities, and applies an ORBGRAND-style ordering to those substitutions. This implementation is a software baseline for the present comparison and is not a circuit-level implementation of ORBGRAND-AI. We also include an ExactBlockProduct ablation using the same nonoverlapping local blocks. It enumerates block substitutions in nondecreasing summed block excess negative log-likelihood rather than by ORBGRAND-style rank weights. This ablation is exact only under the nonoverlapping block-independence approximation; it is not the full finite-memory posterior ordering. The Exact-Markov Tail-Calibrated SOGRAND decoder uses the first-order posterior energy induced by the Gauss--Markov model and enumerates putative noise effects in nondecreasing posterior energy.

We report BLER, found-codeword rate, average and high-percentile membership-query counts, mean estimated missing-list probability, and an unweighted software-work proxy. The found-codeword rate is the probability that at least one codeword is found before the query limit; it is not the probability of correct decoding. For all hard-decision summaries, BLER is computed from the first listed codeword. The additional queries in the tail-stopping variant are used to estimate the missing-list probability and listed-codeword APPs. The ORBGRAND-AI and ExactBlockProduct baselines are hard-output orderings in these experiments, so their missing-list estimates are not reported. Let $N_{\rm mem}$ denote the number of codebook membership queries, $N_{\rm met}$ the number of posterior-metric evaluations, $N_{\rm heap}$ the number of priority-queue removals, and $N_{\rm prep}$ the number of method-specific preprocessing operations. Our reported software-work proxy is $C_{\rm sw}$, defined as
\begin{equation}
\label{eq:C_sw}
C_{\rm sw} = N_{\rm mem} + N_{\rm met} + N_{\rm heap} + N_{\rm prep}.
\end{equation}
This proxy is an implementation-level counter for the present software enumerators. It is not a model of latency, area, throughput, or energy, and it should not be compared with circuit-level
ORBGRAND implementations. BLER uncertainty in the plots is reported using $95\%$ Wilson score confidence intervals \cite{wilson1927probable}. The tables report Monte Carlo point estimates.


\vspace{-2pt}

\subsection{Correlated Gaussian channels}

\subsubsection{Gauss--Markov decoding performance}

Fig.~\ref{fig:gm_ebn0_all} reports the main Gauss--Markov experiment as a function of $E_b/N_0$ for $\rho=0.5$. The six panels show BLER, found-codeword rate, average membership-query count, 99th-percentile membership-query count, mean estimated missing-list probability, and the software-work proxy $C_{\rm sw}$. In this experiment, among the tested ORBGRAND-AI and ExactBlockProduct block sizes $b\in\{2,4,8\}$, larger block sizes give lower BLER point estimates and lower average membership-query counts. For a fixed block size, ExactBlockProduct usually improves on ORBGRAND-AI, which shows that replacing rank weights by exact block-product likelihood scores is useful. The Exact-Markov decoder still gives the lowest BLER point estimates and the smallest average membership-query counts among the tested methods, showing that the nonoverlapping block-product approximation does not recover the full finite-memory posterior ordering.

Table~\ref{tab:gm_bler} gives the BLER point estimates from the Gauss--Markov experiment with $10^4$ frames per point. At $E_b/N_0=3$~dB, ORBGRAND-AI with $b=8$ has BLER $2.85\times10^{-2}$, ExactBlockProduct with $b=8$ has BLER $1.98\times10^{-2}$, and the Exact-Markov decoder has BLER $9.9\times10^{-3}$. At $E_b/N_0=4$ dB, the corresponding values are $2.1\times10^{-3}$, $1.1\times10^{-3}$, and $5.0\times10^{-4}$.

\begin{table}[b]
\centering
\vspace{-4pt}
\caption{BLER point estimates for the Gauss--Markov RLC experiment with
$[n,k]=[128,116]$, $\rho=0.5$, $q_{\max}=20000$, and $10^4$
frames per point.}
\vspace{-3pt}
\label{tab:gm_bler}
\resizebox{\linewidth}{!}{
\begin{tabular}{c|cccccc}
\toprule
$E_b/N_0$ (dB)
& \shortstack{ORBGRAND-AI\\$b=2$}
& \shortstack{ORBGRAND-AI\\$b=4$}
& \shortstack{ORBGRAND-AI\\$b=8$}
& \shortstack{ExactBlockProduct\\$b=8$}
& \shortstack{Random-interleaved\\ORBGRAND}
& \shortstack{Exact-Markov\\~} \\
\midrule
1 & 0.8043 & 0.6683 & 0.5829 & 0.5599 & 0.9555 & 0.4664\\
2 & 0.4379 & 0.2779 & 0.2054 & 0.1752 & 0.7772 & 0.1176\\
3 & 0.1196 & 0.0497 & 0.0285 & 0.0198 & 0.4149 & 0.0099\\
4 & 0.0149 & 0.0033 & 0.0021 & 0.0011 & 0.1178 & 0.0005\\
\bottomrule
\end{tabular}}
\end{table}

Log-linear interpolation of the simulated BLER points as functions of $E_b/N_0$ gives the same ordering among the compared BLER curves. At BLER $10^{-1}$, the Exact-Markov decoder reaches the target at approximately $2.07$ dB, compared with $2.26$ dB for ExactBlockProduct with $b=8$, $2.36$ dB for ORBGRAND-AI with $b=8$, $2.49$ dB for ExactBlockProduct with $b=4$, and $2.59$ dB for ORBGRAND-AI with $b=4$. At BLER $10^{-2}$, the Exact-Markov decoder reaches the target at approximately $3.00$ dB, compared with $3.24$ dB for ExactBlockProduct with $b=8$, $3.40$ dB for ORBGRAND-AI with $b=8$, $3.45$ dB for ExactBlockProduct with $b=4$, and $3.59$ dB for ORBGRAND-AI with $b=4$. These interpolations use only the simulated points in the plotted SNR range. Under the same log-linear interpolation, the Exact-Markov curve reaches BLER $10^{-3}$ at approximately $3.77$ dB; the other tested methods do not attain this target within the plotted range.

\begin{figure*}[!t]
\centering
\begin{minipage}{0.3\linewidth}
\centering
\includegraphics[width=\linewidth]{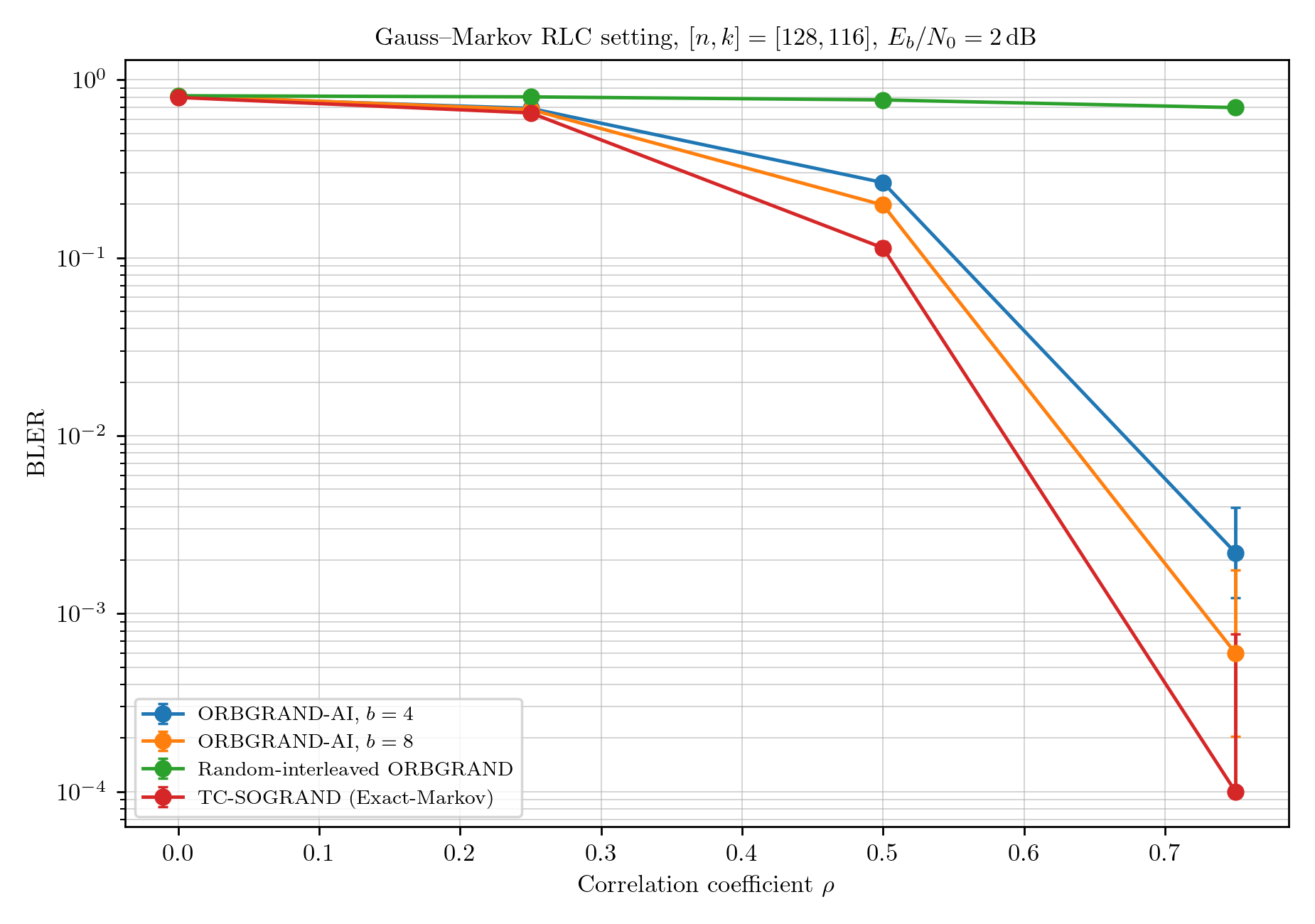}
\end{minipage}
~
\begin{minipage}{0.3\linewidth}
\centering
\includegraphics[width=\linewidth]{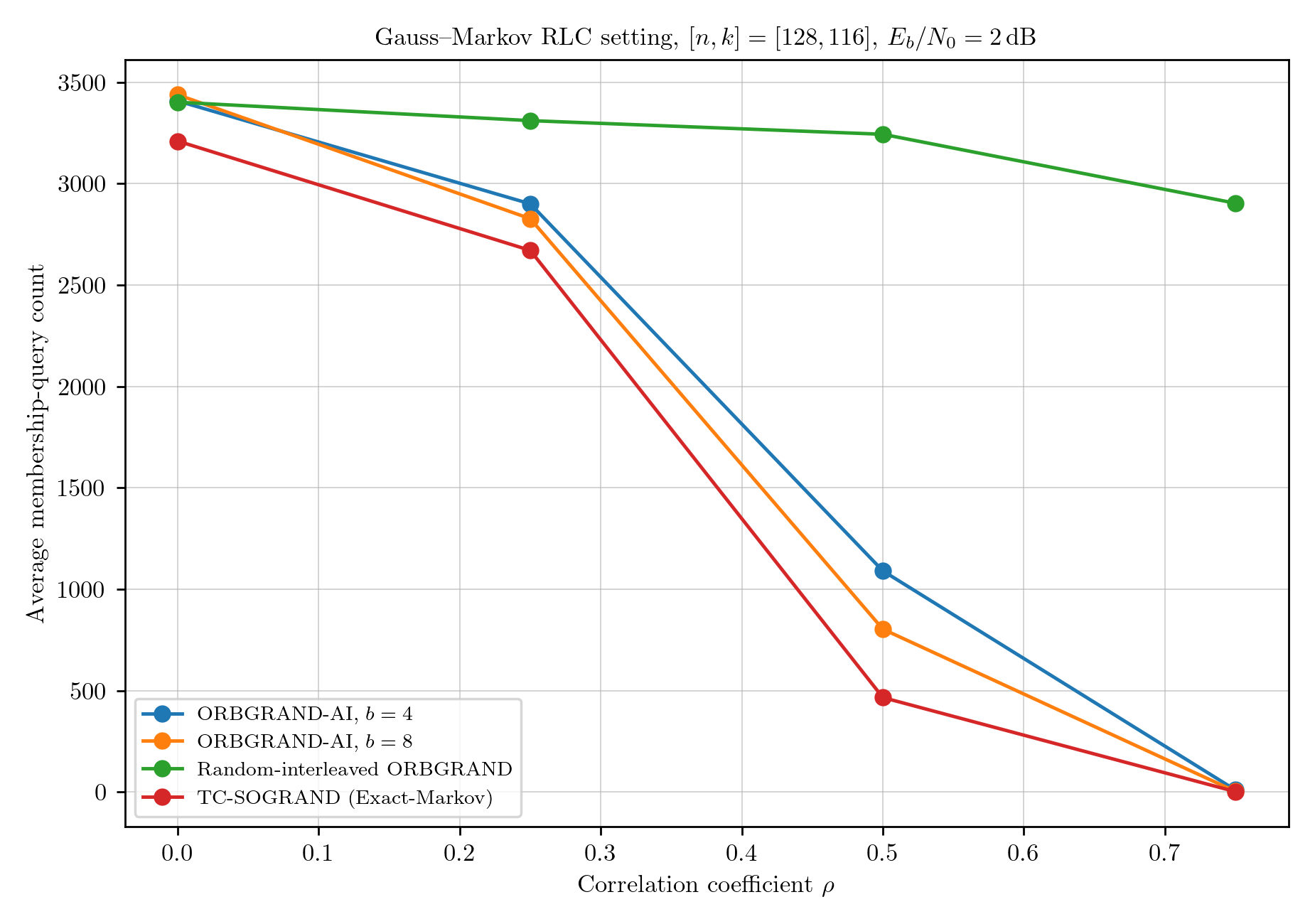}
\end{minipage}
~
\begin{minipage}{0.3\linewidth}
\centering
\includegraphics[width=\linewidth]{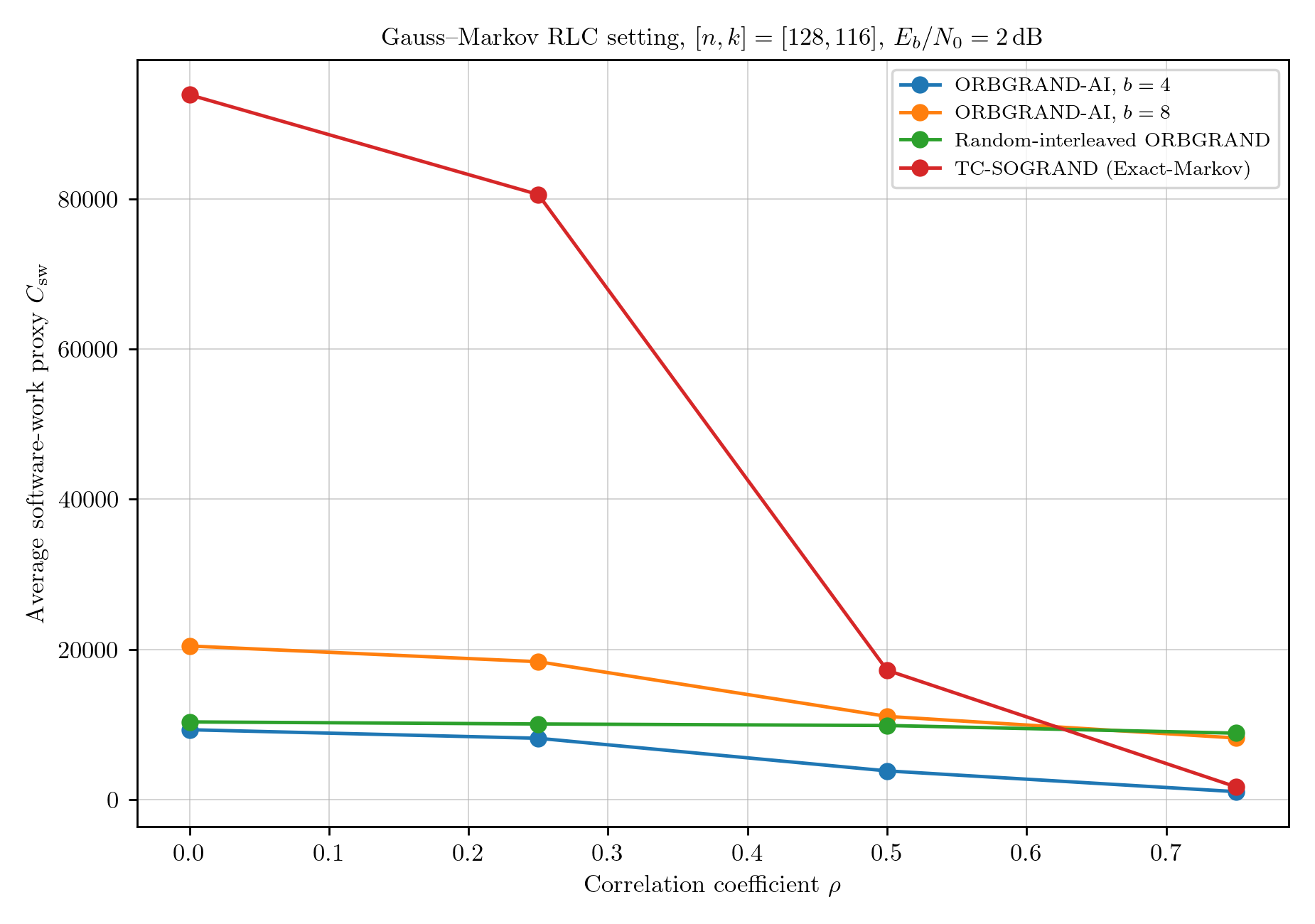}
\end{minipage}
\vspace{-5pt}
\caption{Correlation sweep at $E_b/N_0=2$ dB for the $[128,116]$ RLC,
with $5000$ frames per point. Left: BLER. Middle: average
membership-query count. Right: software-work proxy $C_{\rm sw}$. Zero observed BLER values are displayed at $1/(2N)$ on logarithmic axes,
where $N$ is the number of simulated frames; the corresponding CSV entries
remain equal to zero.}
\vspace{-4pt}
\label{fig:rho_sweep_all}
\end{figure*}

\subsubsection{Membership-query count and enumeration cost}

Table~\ref{tab:gm_queries} reports the average number of codebook membership queries for ORBGRAND-AI with $b=8$, ExactBlockProduct with $b=8$, and the Exact-Markov decoder in the same Gauss--Markov experiment. The ExactBlockProduct ablation has fewer average membership queries than ORBGRAND-AI with the same block size at each tested value of $E_b/N_0$. The Exact-Markov decoder has the smallest average membership-query count at each tested value of $E_b/N_0$. This observation is consistent with the guesswork ordering principle: among deterministic query orders, ordering noise effects in nonincreasing conditional probability minimizes the expected rank of the realized noise effect \cite{christiansen2012guesswork, duffy2019capacity}.
Membership decoding, however, stops at the first queried noise effect whose associated candidate word is a codeword. Therefore, the search may stop before the realized noise effect if an earlier queried noise effect $z^n$ satisfies $y^n\oplus z^n\in C_n$ and $y^n\oplus z^n\ne X^n$. Such an event is a decoding error, not a successful recovery.

\begin{table}[b]
\centering
\vspace{-2pt}
\caption{Average codebook membership-query counts for the Gauss--Markov RLC
experiment with $[n,k]=[128,116]$, $\rho=0.5$, and
$q_{\max}=20000$.}
\vspace{-2pt}
\label{tab:gm_queries}
\footnotesize
\resizebox{\linewidth}{!}{
\begin{tabular}{c|ccc}
\toprule
$E_b/N_0$ (dB)
& \shortstack{ORBGRAND-AI\\$b=8$}
& \shortstack{ExactBlockProduct\\$b=8$}
& \shortstack{Exact-Markov\\~} \\
\midrule
1 & 2403.6 & 2313.6 & 1942.9\\
2 & 829.7  & 709.5  & 471.6\\
3 & 130.6  & 93.0   & 42.2\\
4 & 7.8    & 5.1    & 3.1\\
\bottomrule
\end{tabular}}
\end{table}

Membership-query count measures the number of codebook-membership tests, which is the standard GRAND search metric. It does not include the software work required to generate the candidate noise effects. To report that overhead for the present implementation, we use the unweighted software-work proxy defined in \eqref{eq:C_sw}. The four terms are reported with unit weight only to give a reproducible software-work proxy for the present Python implementation. The quantity $C_{\rm sw}$ is not an asymptotic complexity bound and is not a hardware model.


The software-work proxy gives a different comparison from the membership-query count. At $E_b/N_0=1$ dB, $C_{\rm sw}$ is approximately $6.08\times10^4$ for Exact-Markov, $1.32\times10^4$ for ORBGRAND-AI with $b=8$, and $1.15\times10^4$ for ExactBlockProduct with $b=8$. At $E_b/N_0=3$ dB, the corresponding values are approximately $3.26\times10^3$, $4.52\times10^3$, and $4.35\times10^3$. Thus, in this software implementation, Exact-Markov has the smallest average membership-query count among the tested orderings, but its software-work proxy is larger at low SNR because the exact finite-memory
enumerator performs more posterior-metric and priority-queue operations. The result separates the cost of codebook membership tests from the cost of generating the ordered candidate noise effects.

\subsubsection{Dependence on the correlation coefficient}

Fig.~\ref{fig:rho_sweep_all} summarizes the Gauss--Markov correlation sweep at fixed $E_b/N_0=2$ dB, while Table~\ref{tab:rho_sweep} reports the corresponding BLER point estimates. The three panels show BLER, average membership-query count, and the software-work proxy $C_{\rm sw}$. When $\rho=0$, the Gauss--Markov noise samples are independent, and the BLER point estimates of the tested methods are close. As $\rho$ increases over the tested values, the gap between the BLER point estimates of the correlation-aware orderings and those of the memoryless or interleaved baselines increases. At $\rho=0.5$, ORBGRAND-AI with $b=8$ has BLER point estimate $0.2004$, whereas the Exact-Markov decoder has BLER point estimate $0.1126$. At $\rho=0.75$, no errors were observed for ORBGRAND-AI with $b=8$ or for the Exact-Markov decoder in $5000$ frames. The corresponding $95\%$ Wilson upper confidence endpoint is $7.68\times10^{-4}$; therefore, the zero table entries should be interpreted as finite-sample Monte Carlo point estimates, not as evidence of zero error probability.

\begin{table}[!b]
\centering
\caption{BLER point estimates versus Gauss--Markov correlation coefficient
at $E_b/N_0=2$ dB. Each point uses $5000$ frames. A zero entry means
that no errors were observed in the simulated frames.}
\vspace{-4pt}
\label{tab:rho_sweep}
\resizebox{\linewidth}{!}{
\begin{tabular}{c|cccc}
\toprule
$\rho$
& \shortstack{Random-interleaved\\ORBGRAND}
& \shortstack{ORBGRAND-AI\\$b=4$}
& \shortstack{ORBGRAND-AI\\$b=8$}
& \shortstack{Exact-Markov\\~} \\
\midrule
0.00 & 0.8050 & 0.8098 & 0.8110 & 0.8020\\
0.25 & 0.8002 & 0.6986 & 0.6854 & 0.6532\\
0.50 & 0.7656 & 0.2682 & 0.2004 & 0.1126\\
0.75 & 0.6928 & 0.0030 & 0 & 0\\
\bottomrule
\end{tabular}}
\end{table}

\subsubsection{Model mismatch}

Fig.~\ref{fig:model_mismatch_all} and Table~\ref{tab:model_mismatch} evaluate the sensitivity of the Exact-Markov decoder to mismatch in the assumed posterior model. The true
channel has $\rho=0.5$ and $E_b/N_0=3$ dB, while the decoder uses $\rho_{\rm model}$ in its assumed first-order posterior energy. The figure shows the BLER and average membership-query trends as functions of $\rho_{\rm model}$, and the table reports the corresponding BLER, average query count, 99th-percentile query count, and software-work proxy $C_{\rm sw}$.

Among the tested values, the matched value $\rho_{\rm model}=0.5$ gives the lowest BLER point estimate and the smallest average membership-query count. The mismatched values give higher BLER point estimates and larger average membership-query counts, with the largest degradation occurring for the memoryless model $\rho_{\rm model}=0$. For example, $\rho_{\rm model}=0.25$ gives a BLER point estimate of $0.0400$, whereas the memoryless model gives a BLER point estimate of $0.4044$. Thus, within this tested grid, matching the posterior correlation parameter improves both BLER and membership-query count for the Exact-Markov ordering.

\begin{table}[b]
\centering
\caption{Model-mismatch sweep for the Exact-Markov decoder with true $\rho=0.5$, $E_b/N_0=3$ dB, and $5000$ frames per point.}
\vspace{-4pt}
\label{tab:model_mismatch}
\resizebox{\linewidth}{!}{
\begin{tabular}{c|cccc}
\toprule
$\rho_{\rm model}$
& BLER
& Average membership-query count
& 99th-percentile query count
& $C_{\rm sw}$\\
\midrule
0.00 & 0.4044 & 1656.4 & 13419.2 & 53176\\
0.25 & 0.0400 & 143.1  & 3319.2  & 6720\\
0.50 & 0.0112 & 37.3   & 656.3   & 3102\\
0.75 & 0.0154 & 71.1   & 1523.1  & 4297\\
0.90 & 0.0250 & 110.2  & 2382.2  & 5555\\
\bottomrule
\end{tabular}}
\end{table}

\begin{figure}[!t]
\centering
\begin{minipage}{0.482\linewidth}
\centering
\includegraphics[width=\linewidth]{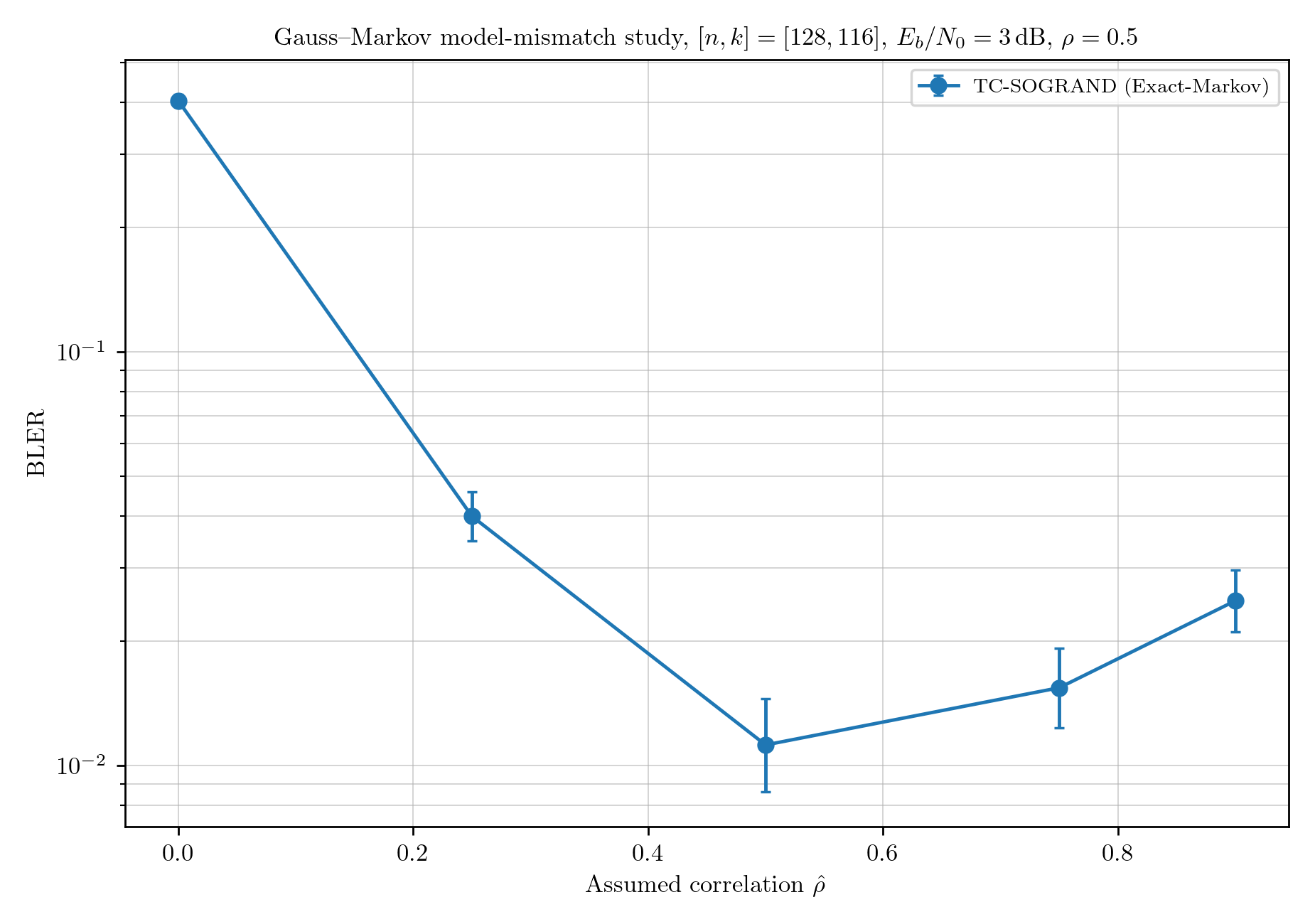}
\end{minipage}
\begin{minipage}{0.482\linewidth}
\centering
\includegraphics[width=\linewidth]{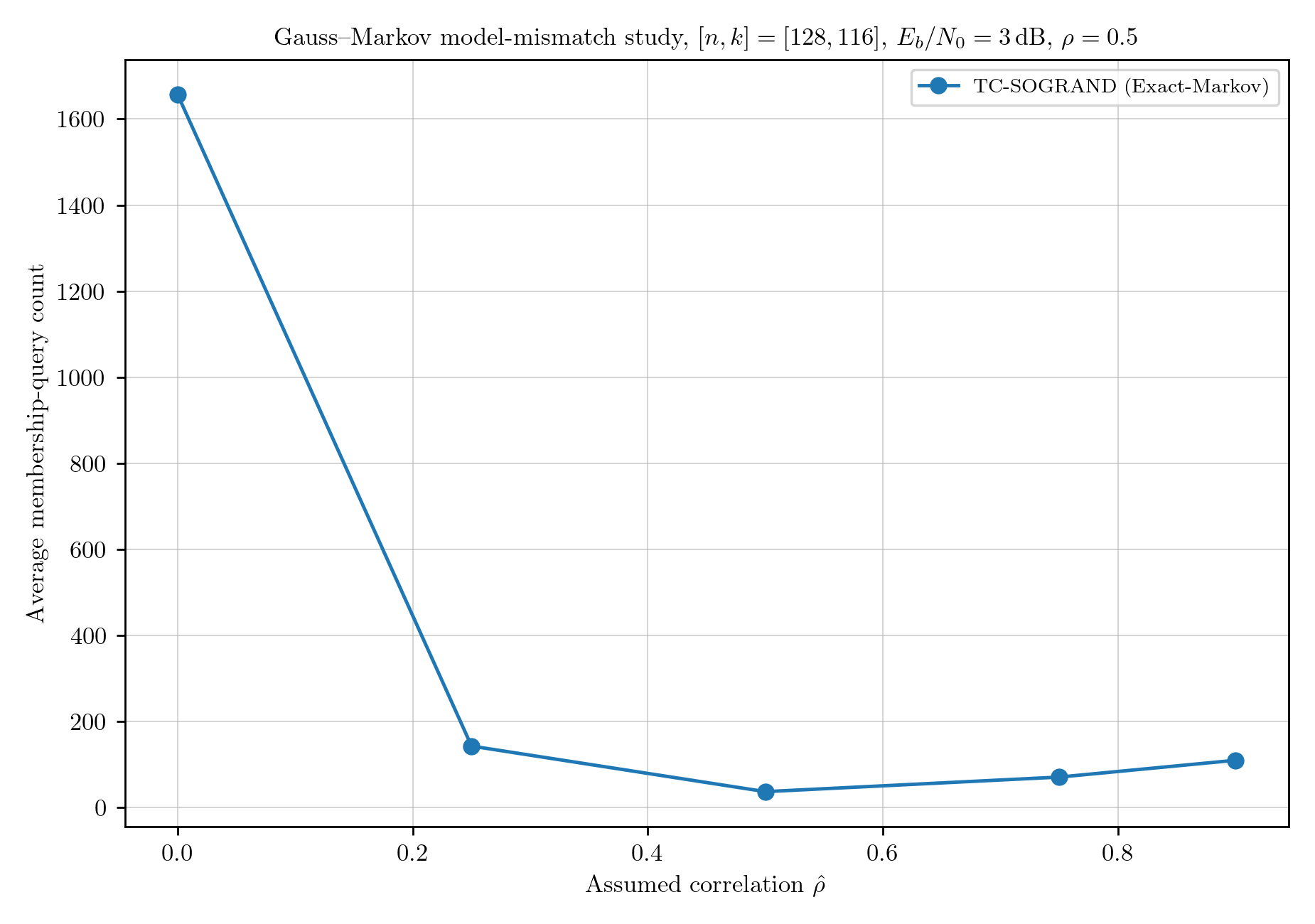}
\end{minipage}
\vspace{-8pt}
\caption{Model-mismatch sweep for the Exact-Markov decoder. The true channel has $\rho=0.5$ and $E_b/N_0=3$ dB. Left: BLER. Right: average membership-query count.}
\vspace{-3pt}
\label{fig:model_mismatch_all}
\end{figure}

\begin{figure*}[!h]
\centering
\begin{minipage}{0.31\linewidth}
\centering
\includegraphics[width=\linewidth]{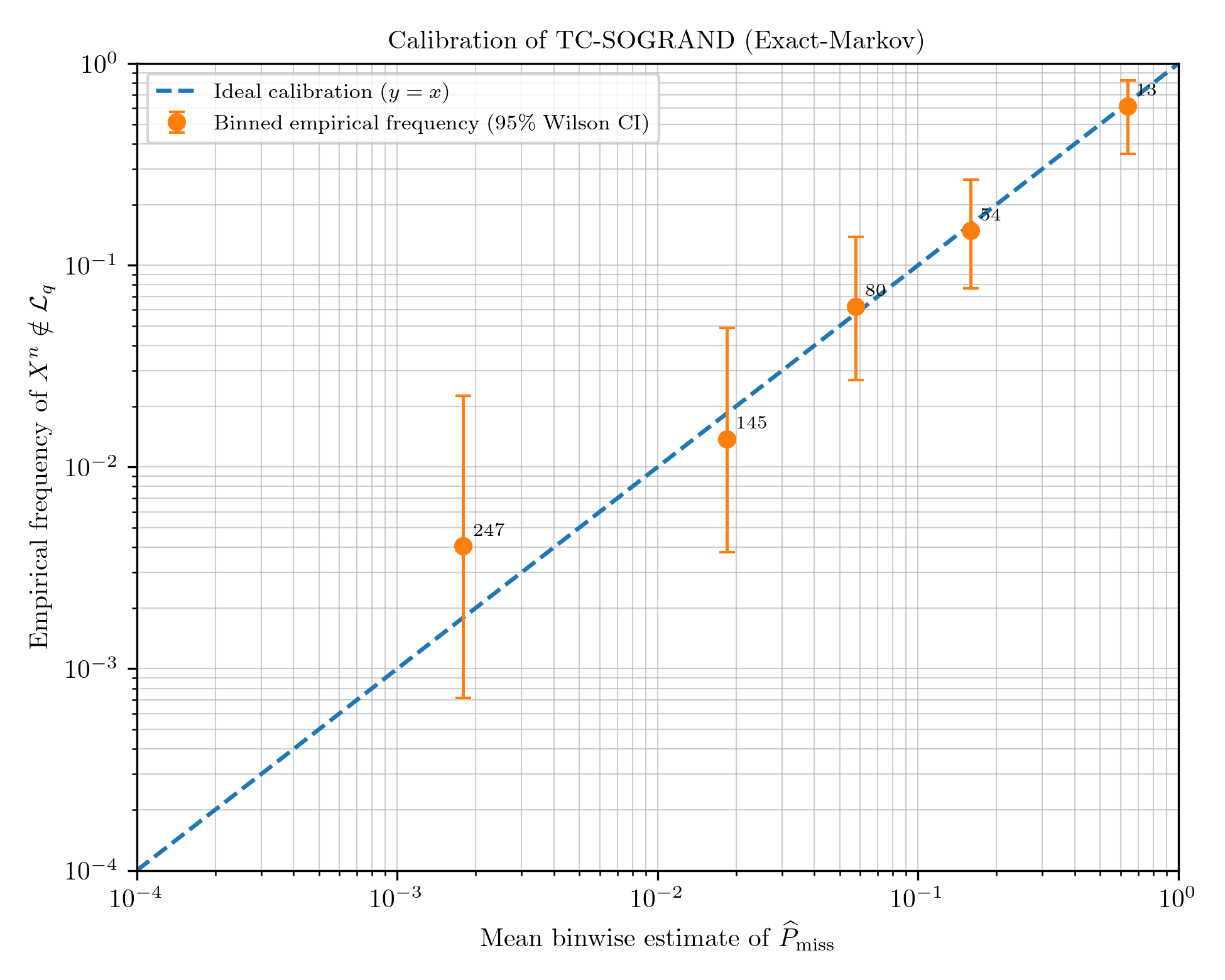}
\end{minipage}
~
\begin{minipage}{0.31\linewidth}
~
\includegraphics[width=\linewidth]{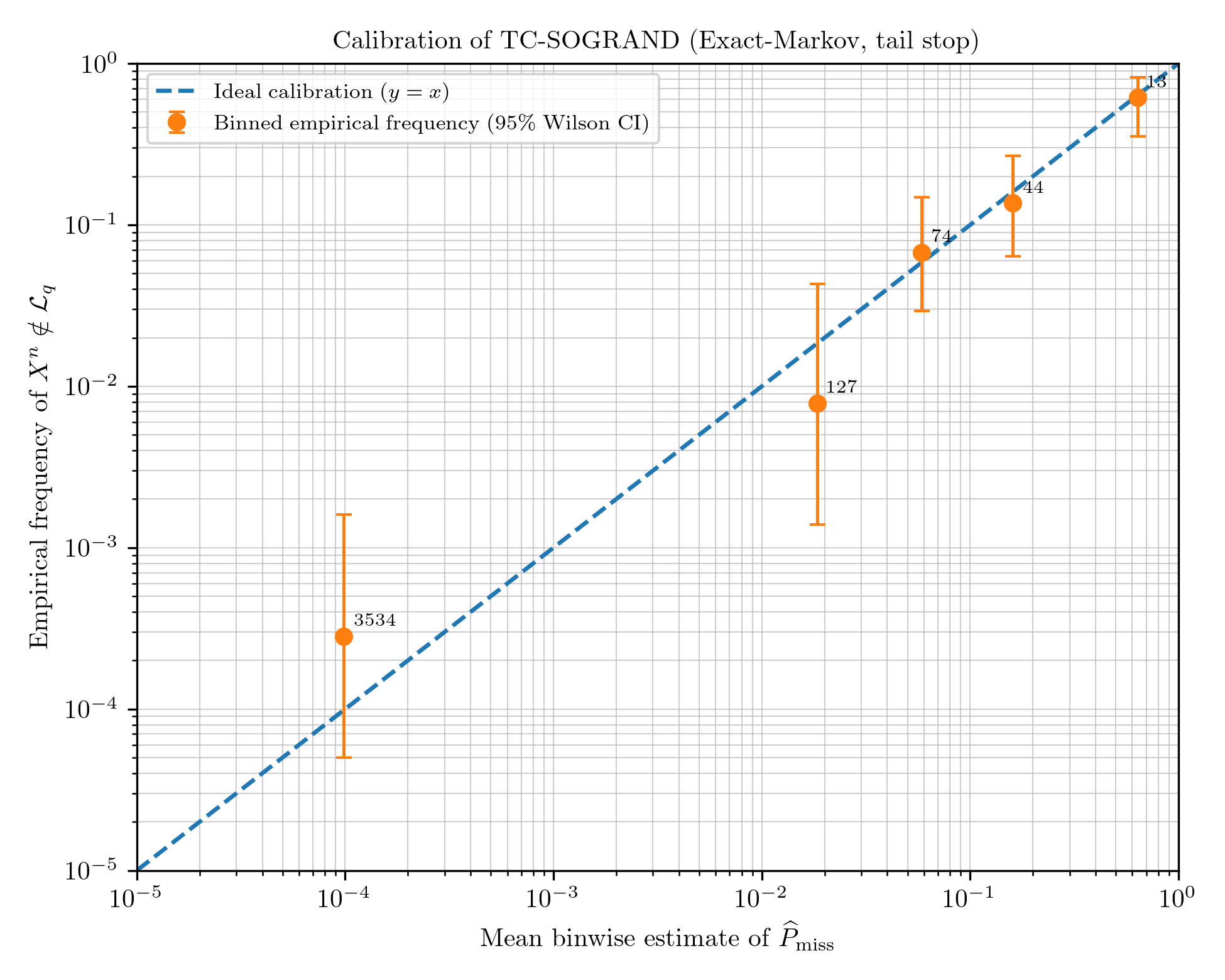}
\end{minipage}
\hfill
\begin{minipage}{0.31\linewidth}
\centering
\includegraphics[width=\linewidth]{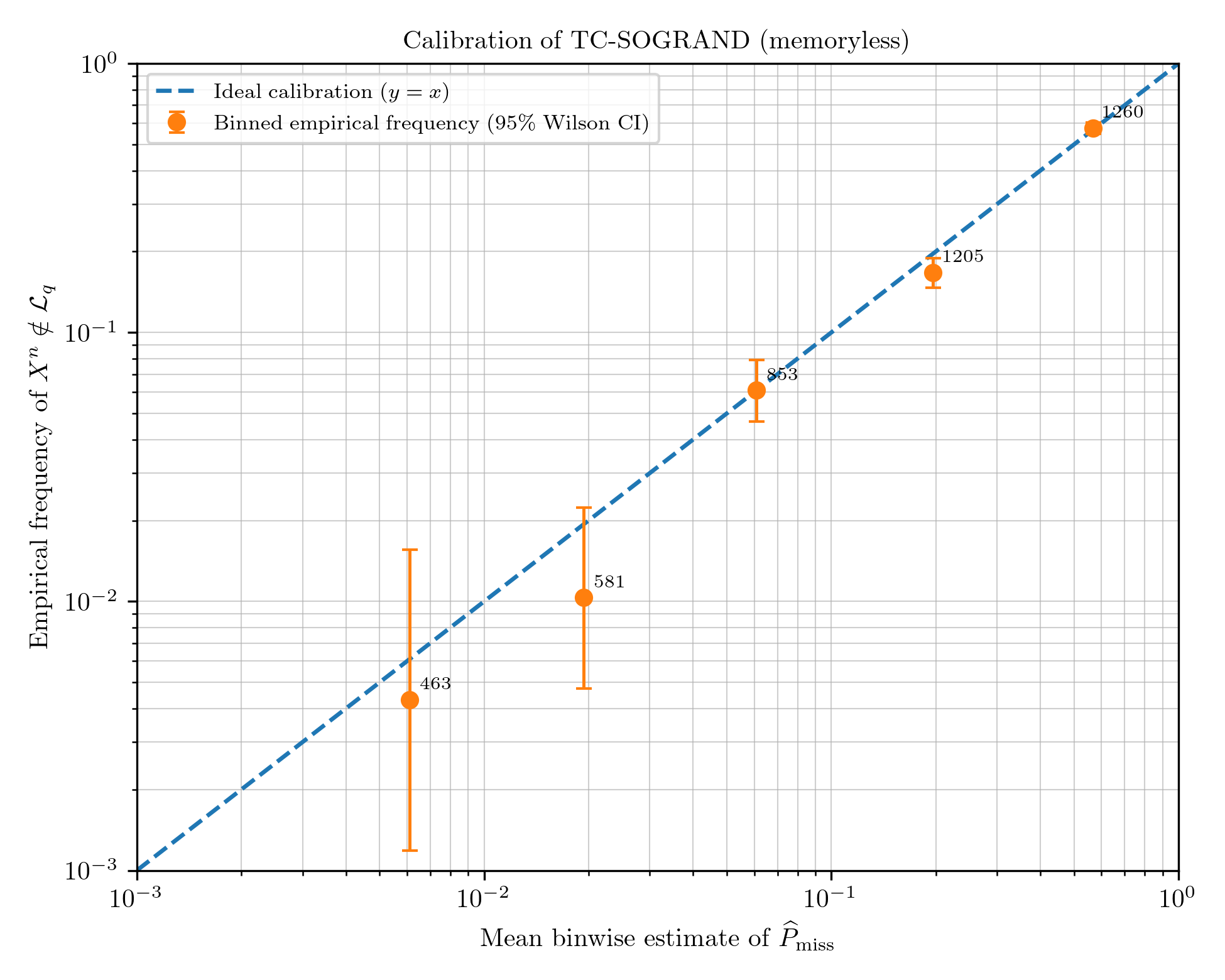}
\end{minipage}
\caption{Missing-list reliability diagrams. Left: Exact-Markov list-output decoder. Middle: Exact-Markov tail-stopping variant. Right: memoryless SOGRAND reference on a memoryless AWGN channel. Markers show binned empirical missing-list frequencies, error bars show Wilson confidence intervals, and labels indicate bin counts.}
\label{fig:calibration_all}
\end{figure*}

\subsection{Calibration and sensitivity analyses}

\subsubsection{Missing-list calibration}

In the RLC experiments, Tail-Calibrated SOGRAND uses the random-codebook occupancy approximation to estimate the codebook-restricted missing-list probability as $\widehat P_{\rm miss}(q) = \widehat U_q / (S_q+\widehat U_q)$. To assess the empirical calibration of this estimate, the values of $\widehat P_{\rm miss}(q)$ at the decoder output are grouped into logarithmic bins. For each bin, we compare the mean estimated missing-list probability with the empirical frequency of the event $X^n\notin \mathcal L_q$. The calibration experiment uses
$[n,k]=[64,52]$, $\rho=0.5$, $E_b/N_0=2$ dB, $q_{\max}=10000$, maximum list size $4$, and $5000$ frames.

Fig.~\ref{fig:calibration_all} reports three reliability diagrams: the Exact-Markov list-output decoder, the Exact-Markov tail-stopping variant, and a memoryless SOGRAND reference on a memoryless AWGN channel. For the Exact-Markov list-output decoder, the count-weighted mean estimated missing-list probability is $5.28\times10^{-3}$, the empirical missing-list frequency is $5.00\times10^{-3}$, and the count-weighted expected calibration error is $9.33\times10^{-4}$. For the Exact-Markov tail-stopping variant, the corresponding quantities are $4.85\times10^{-3}$, $4.20\times10^{-3}$, and $1.17\times10^{-3}$, respectively. These results show small empirical calibration discrepancies in the populated probability range. The lowest-probability bins contain too few missing-list events to assess calibration below $10^{-4}$; resolving that range would require more frames or a rare-event simulation method.

\subsubsection{Missing-list stopping tolerance}

\begin{table}[!b]
\centering
\caption{Missing-list stopping tolerance sweep for the Exact-Markov-tailStop decoder at $E_b/N_0=3$ dB, $\rho=0.5$, $q_{\max}=20000$, and $5000$ frames per point.}
\vspace{-4pt}
\label{tab:eta_sweep}
\resizebox{\linewidth}{!}{
\begin{tabular}{c|cccc}
\toprule
$\eta$
& BLER
& Average membership-query count
& Mean output $\widehat P_{\rm miss}$
& $C_{\rm sw}$\\
\midrule
$10^{-1}$ & 0.0104 & 112.8  & $8.17\times10^{-3}$ & 5232\\
$10^{-2}$ & 0.0120 & 442.4  & $3.43\times10^{-3}$ & 14674\\
$10^{-3}$ & 0.0104 & 1255.4 & $1.43\times10^{-3}$ & 39629\\
$10^{-4}$ & 0.0088 & 3099.0 & $5.89\times10^{-4}$ & 98204\\
$10^{-5}$ & 0.0112 & 6232.5 & $5.03\times10^{-4}$ & 190899\\
\bottomrule
\end{tabular}}
\end{table}

Fig.~\ref{fig:eta_sweep_all} summarizes the effect of the stopping tolerance $\eta$ for the Exact-Markov-tailStop decoder, while Table~\ref{tab:eta_sweep} reports the corresponding BLER point estimates, average membership-query counts, mean output values of $\widehat P_{\rm miss}$, and software-work proxy $C_{\rm sw}$. The experiment uses $E_b/N_0=3$ dB, $\rho=0.5$, and $q_{\max}=20000$. In this variant, the decoder stops when the estimated missing-list probability $\widehat P_{\rm miss}(q)$ is at most $\eta$, or when the query limit $q_{\max}$ is reached. Hence the reported output value of $\widehat P_{\rm miss}$ can exceed $\eta$ on frames that terminate at the query limit. Over the tested values, decreasing $\eta$ increases the average membership-query count and decreases the mean output value of $\widehat P_{\rm miss}$. The BLER point estimates vary within the Monte Carlo uncertainty of this $5000$-frame experiment and should not be interpreted as a monotone function of $\eta$.

\begin{figure*}[t]
\centering
\begin{minipage}{0.3\linewidth}
\centering
\includegraphics[width=\linewidth]{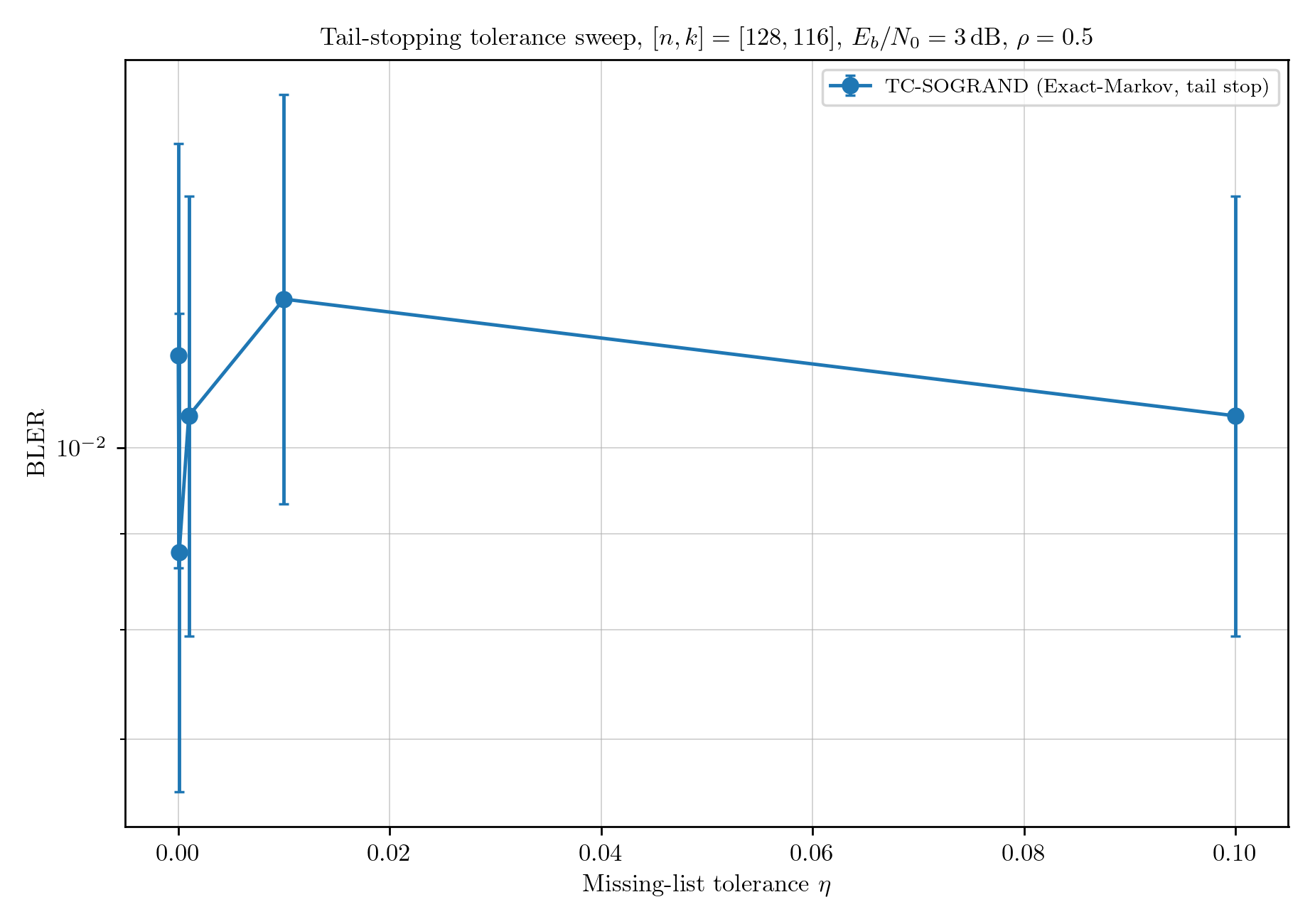}
\end{minipage}
~
\begin{minipage}{0.3\linewidth}
\centering
\includegraphics[width=\linewidth]{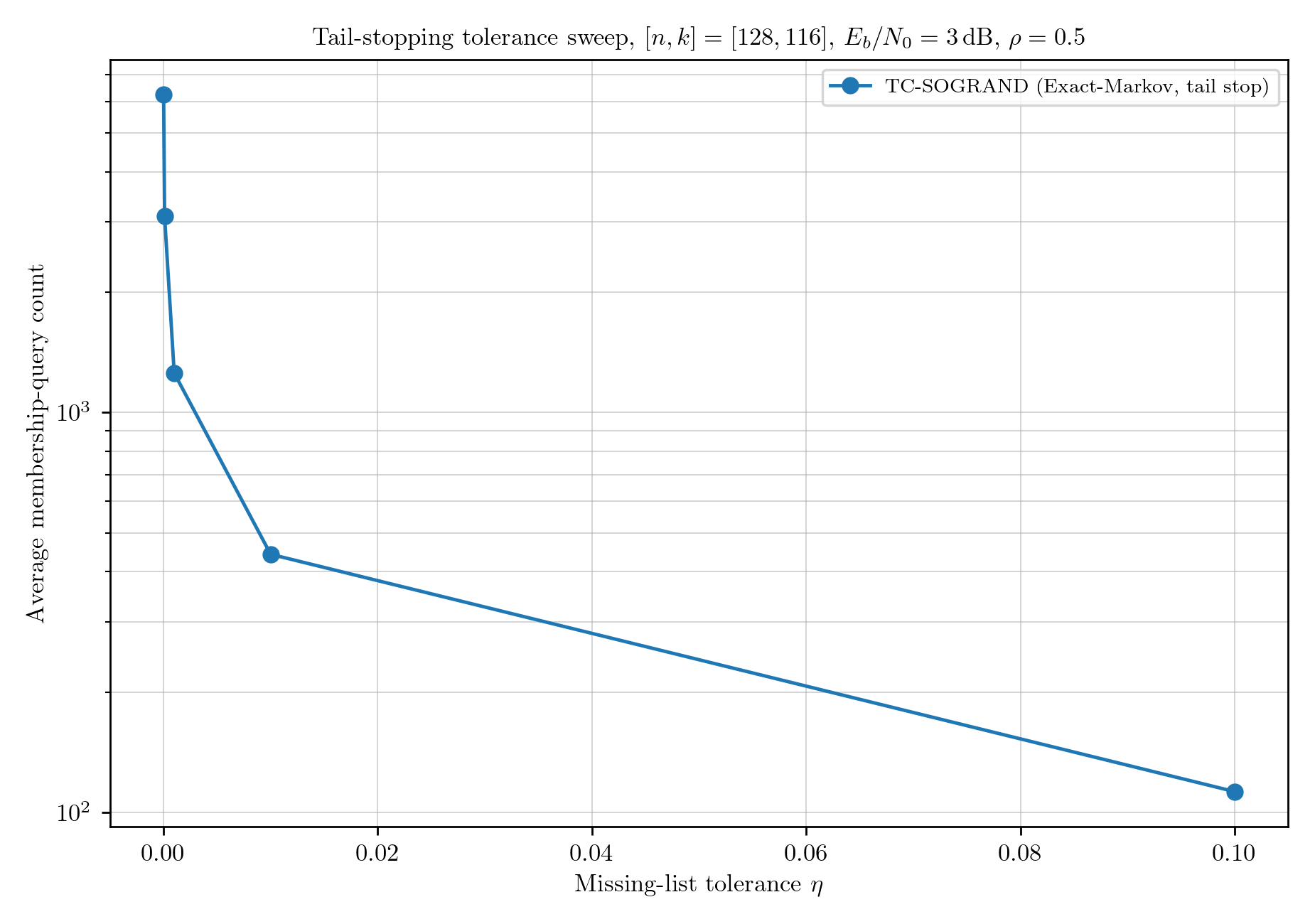}
\end{minipage}
~
\begin{minipage}{0.3\linewidth}
\centering
\includegraphics[width=\linewidth]{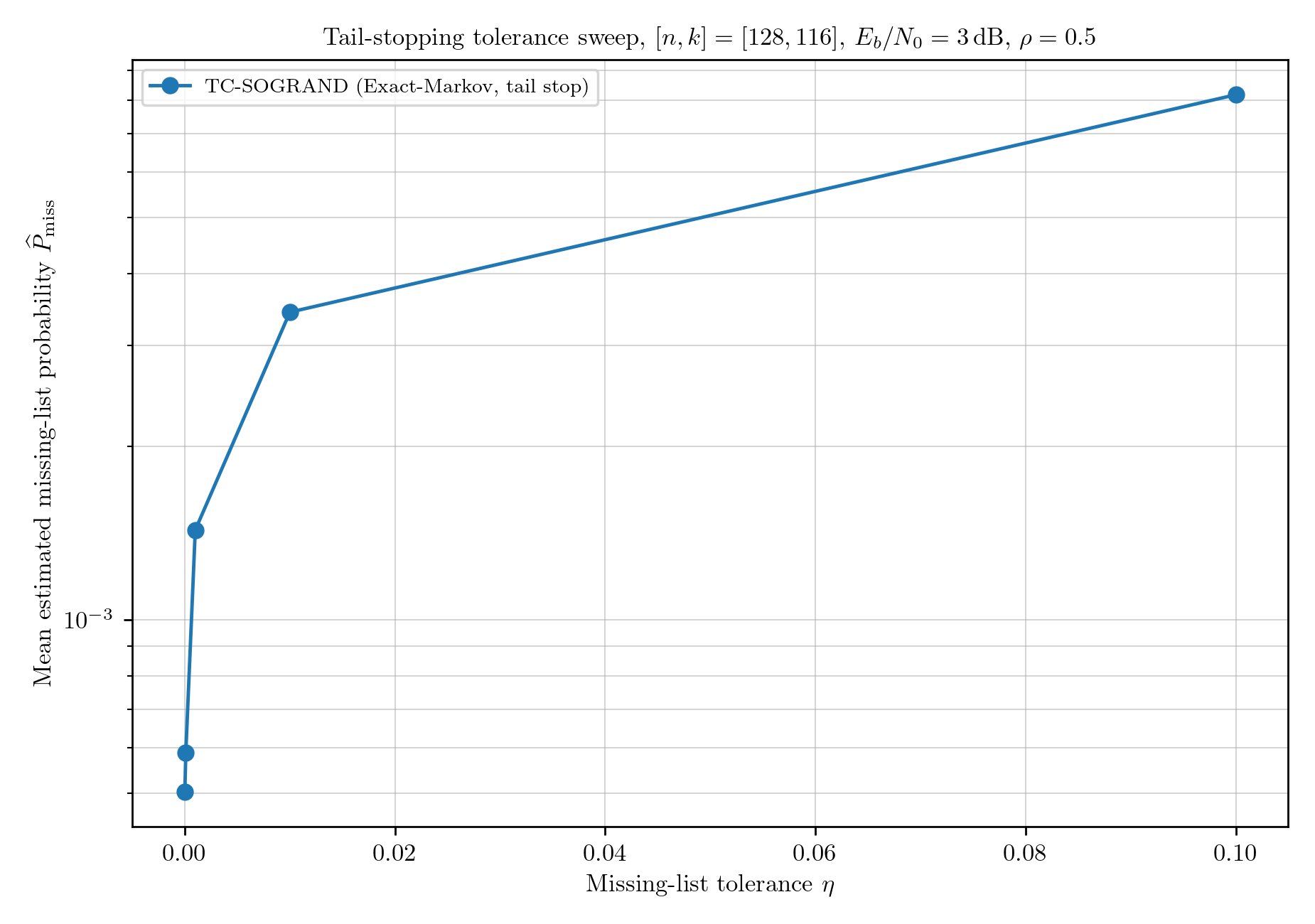}
\end{minipage}
\vspace{-5pt}
\caption{Missing-list stopping tolerance sweep for the
Exact-Markov-tailStop decoder with $[n,k]=[128,116]$, $\rho=0.5$,
$E_b/N_0=3$ dB, and $q_{\max}=20000$. Left: BLER. Middle: average
membership-query count. Right: mean output value of the estimated
missing-list probability.}
\label{fig:eta_sweep_all}
\end{figure*}


\subsubsection{Membership-query limit}

Fig.~\ref{fig:qmax_sweep_all} reports a sweep over the membership-query
limit $q_{\max}$ at $E_b/N_0=3$ dB and $\rho=0.5$. For the
Exact-Markov-tailStop decoder, the missing-list stopping tolerance is
$\eta=10^{-5}$. The Exact-Markov-tailStop decoder has a lower BLER
point estimate than ORBGRAND-AI with $b=8$ at each tested query limit.
At $q_{\max}=1000$, the BLER point estimates are $0.0386$ for
ORBGRAND-AI with $b=8$ and $0.0128$ for Exact-Markov-tailStop. At
$q_{\max}=20000$, the corresponding values are $0.0290$ and
$0.0108$.

Over the tested query limits, the found-codeword-rate point estimates increase with $q_{\max}$, and the Exact-Markov-tailStop decoder has a higher found-codeword-rate point estimate at the smaller tested query limits. This sweep uses independent Monte Carlo samples for different values of $q_{\max}$. Therefore, small nonmonotone changes in BLER across $q_{\max}$ should be interpreted as Monte Carlo variation, not as deterministic properties of the decoders.

\begin{figure*}[t]
\centering
\begin{minipage}{0.3\linewidth}
\centering
\includegraphics[width=\linewidth]{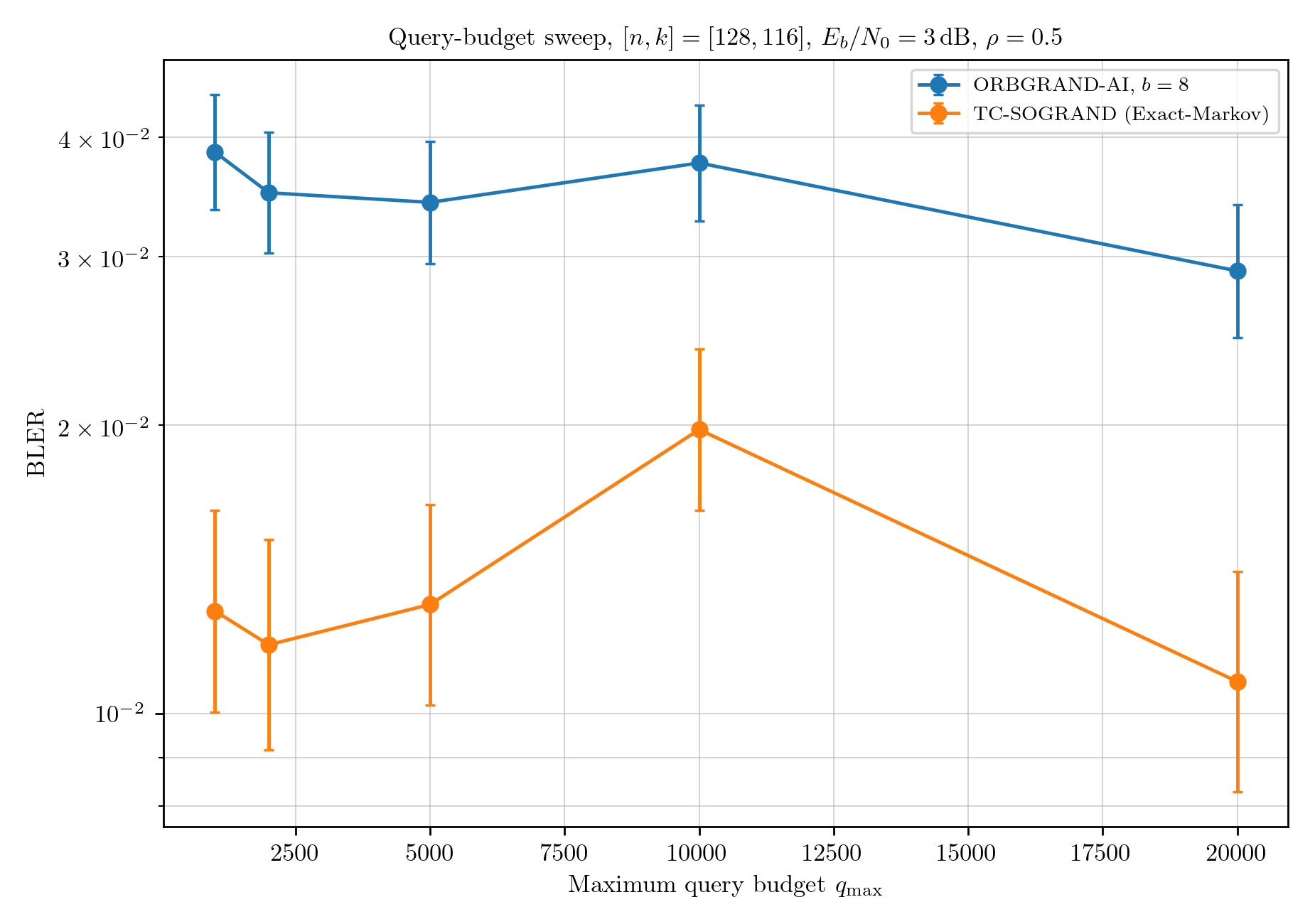}
\end{minipage}
~
\begin{minipage}{0.3\linewidth}
\centering
\includegraphics[width=\linewidth]{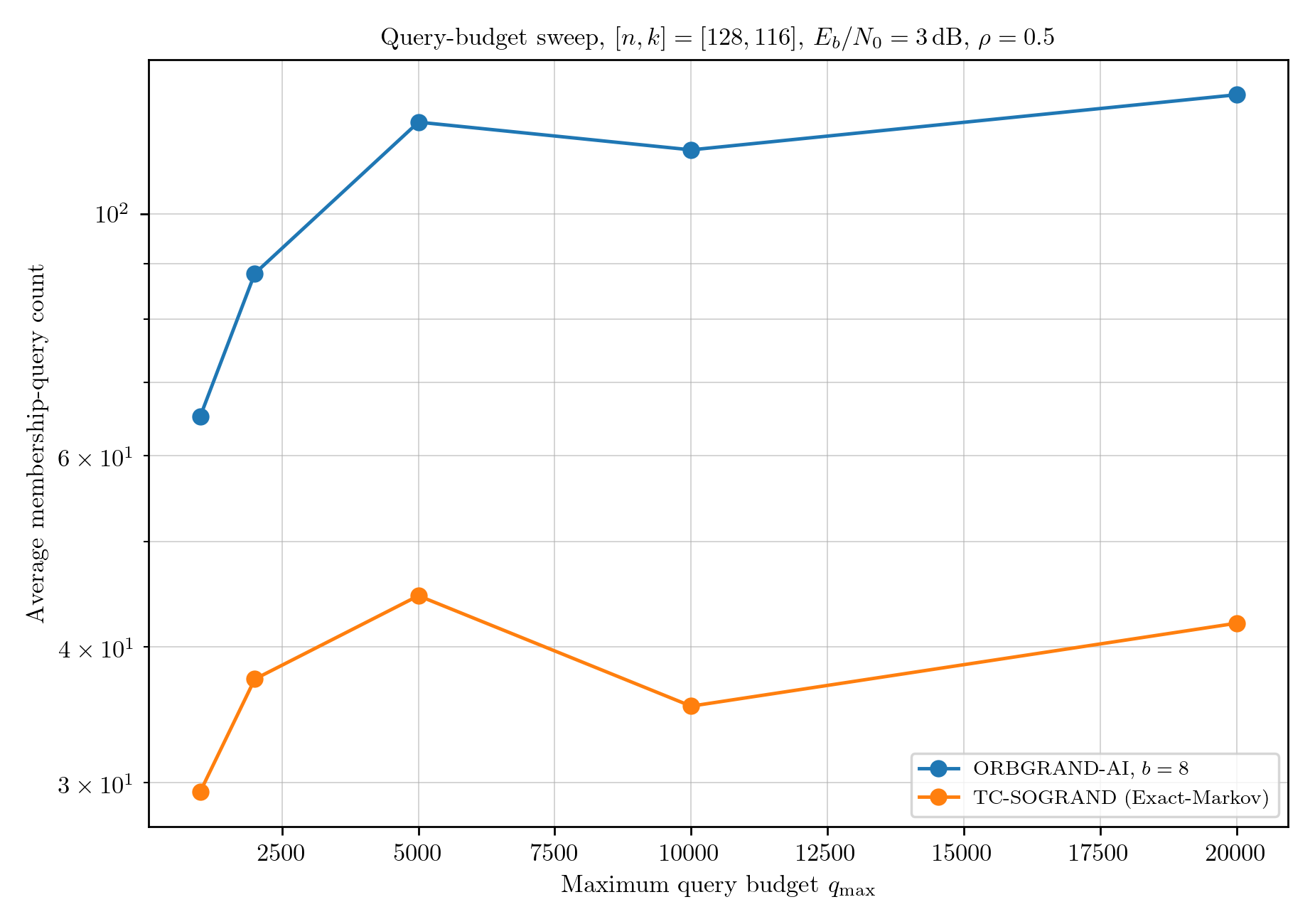}
\end{minipage}
~
\begin{minipage}{0.3\linewidth}
\centering
\includegraphics[width=\linewidth]{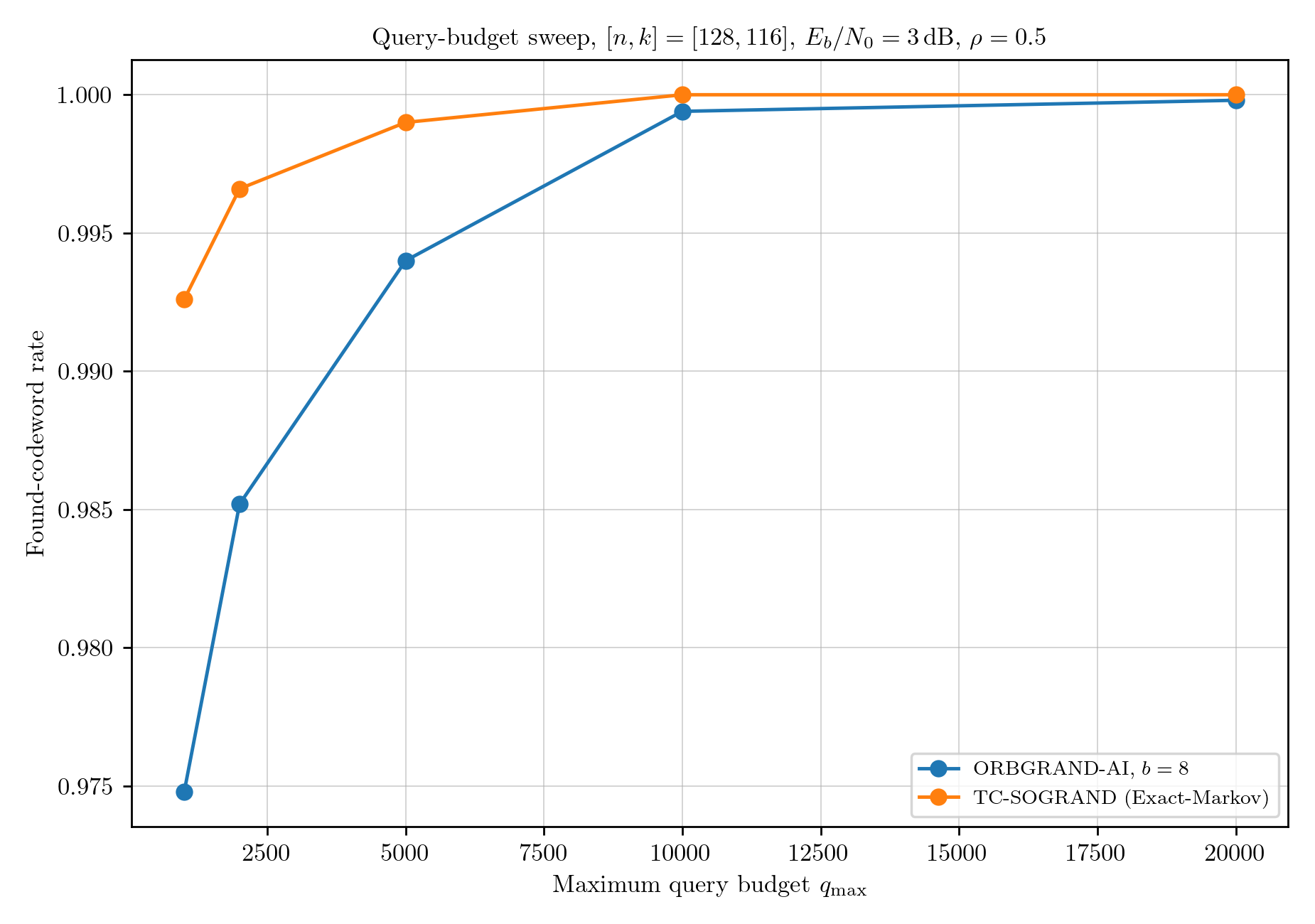}
\end{minipage}
\vspace{-5pt}
\caption{Membership-query-limit sweep at $E_b/N_0=3$ dB for the $[128,116]$ RLC with $\rho=0.5$. For Exact-Markov-tailStop, $\eta=10^{-5}$. Left: BLER. Middle: average membership-query count.
Right: found-codeword rate.}
\vspace{-5pt}
\label{fig:qmax_sweep_all}
\end{figure*}

\vspace{10pt}

\subsubsection{Fixed membership-query-limit rate sweep}

\begin{table}[!b]
\centering
\caption{Fixed-query-limit rate sweep for length-$128$ RLCs at $E_b/N_0=2$ dB, $\rho=0.5$, $q_{\max}=20000$, and $5000$ frames per point.}
\vspace{-4pt}
\label{tab:rate_sweep_exact}
\resizebox{\linewidth}{!}{
\begin{tabular}{c|c|cccc}
\toprule
$k$
& $R_{\rm c}$
& BLER
& Found-codeword rate
& Average membership-query count
& 99th-percentile query count\\
\midrule
 96 & 0.7500 & 0.2098 & 0.7902 & 5801.4 & 20000.0\\
104 & 0.8125 & 0.1046 & 0.8954 & 3372.0 & 20000.0\\
110 & 0.8594 & 0.0650 & 0.9460 & 1984.9 & 20000.0\\
116 & 0.9063 & 0.1200 & 0.9998 & 476.6  & 7131.7\\
120 & 0.9375 & 0.2732 & 1.0000 & 71.3   & 688.0\\
\bottomrule
\end{tabular}}
\end{table}

Fig.~\ref{fig:rate_sweep_all} summarizes the fixed-query-limit rate sweep, while Table~\ref{tab:rate_sweep_exact} reports the corresponding BLER, found-codeword rate, average membership-query count, and 99th-percentile query count for the Exact-Markov decoder. The experiment uses a fixed $E_b/N_0=2$ dB and $q_{\max}=20000$. Note that this is a fixed-query-limit experiment, not an asymptotic rate-threshold experiment. At lower rates, the codebook is sparser and, under the fixed-$E_b/N_0$ normalization, the channel noise variance is larger. Consequently, the decoder may fail to find any codeword within the query limit. At higher rates, the codebook is denser and the channel noise variance is smaller. The denser codebook increases the probability that the search encounters an incorrect codeword before the transmitted codeword. The BLER point estimates in this sweep are therefore affected jointly by the rate-dependent channel variance,
codebook density, and query-limit truncation.

For the Exact-Markov decoder and $k\in\{96,104,110,116,120\}$, the BLER point estimates are $0.2098$, $0.1046$, $0.0650$, $0.1200$, and $0.2732$, respectively.
The corresponding found-codeword rates are $0.7902$, $0.8954$, $0.9460$, $0.9998$, and $1.0000$. The found-codeword rate records only whether the membership search returns
at least one codeword before the query limit; it is not a correctness probability. Thus, the value $1.0000$ at $k=120$ indicates that the search found a codeword in every simulated frame, not that every decoded word was correct. The nonmonotonic BLER trend reflects the finite-query regime: as $k$ increases, the code rate changes the channel noise variance used at fixed $E_b/N_0$, while the larger codebook also increases the probability of an earlier incorrect codeword. These competing effects are observed under the fixed query limit $q_{\max}$.

\begin{figure*}[!t]
\centering
\begin{minipage}{0.32\linewidth}
\centering
\includegraphics[width=\linewidth]{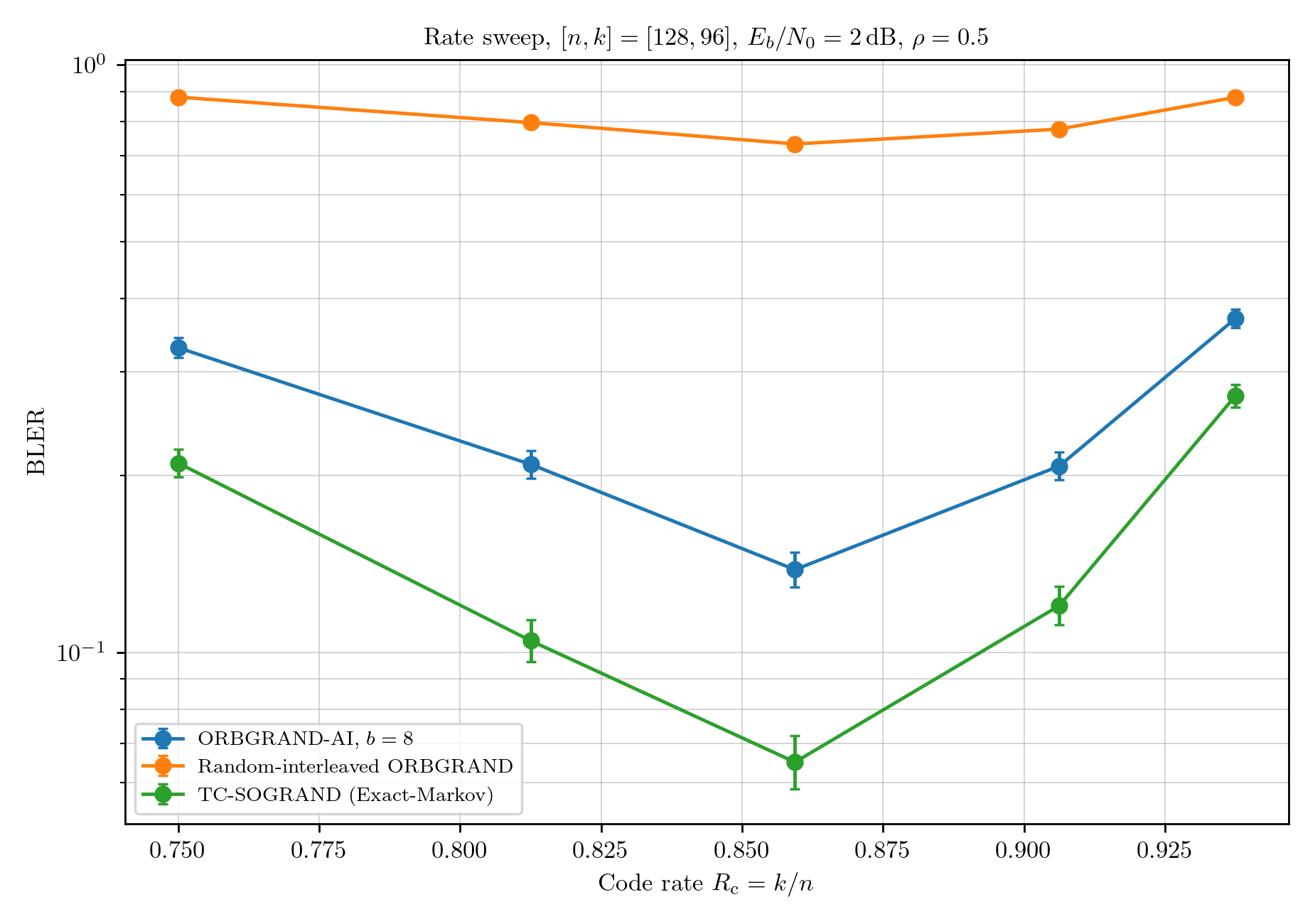}
\end{minipage}
\hfill
\begin{minipage}{0.32\linewidth}
\centering
\includegraphics[width=\linewidth]{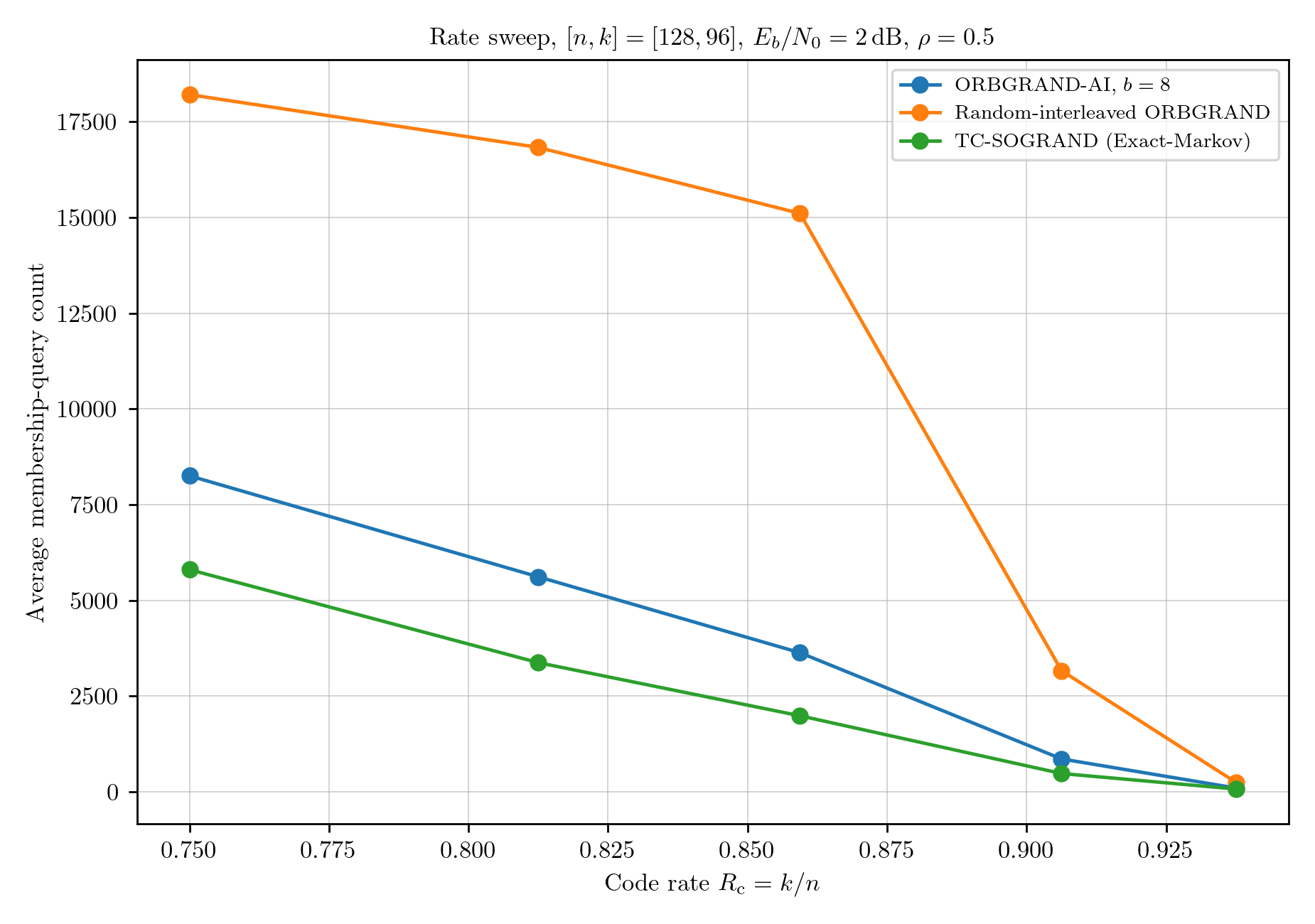}
\end{minipage}
\hfill
\begin{minipage}{0.32\linewidth}
\centering
\includegraphics[width=\linewidth]{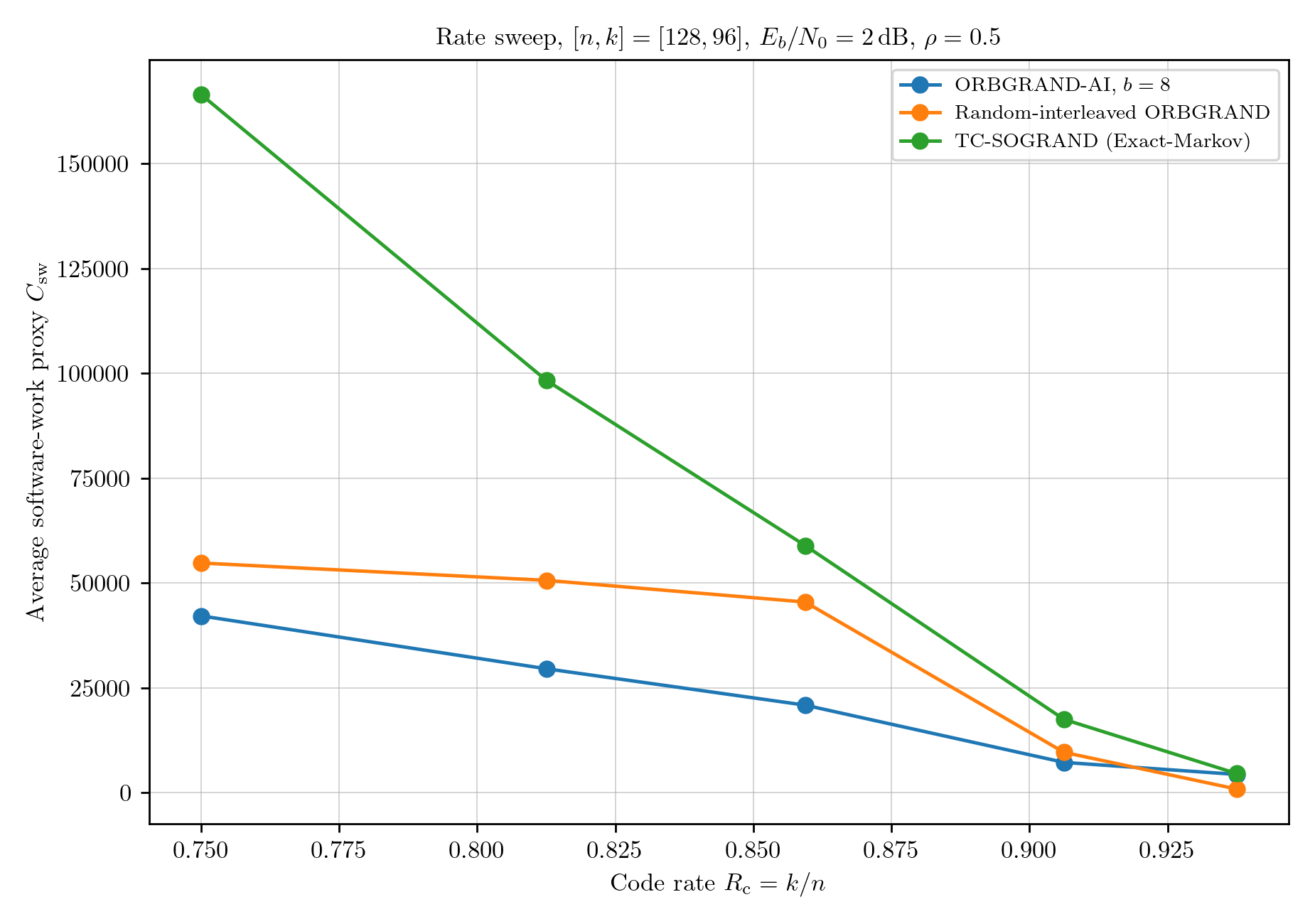}
\end{minipage}
\caption{Fixed-query-limit rate sweep for length-$128$ RLCs at $E_b/N_0=2$ dB, $\rho=0.5$, and $q_{\max}=20000$. Left: BLER. Middle: average membership-query count. Right: software-work proxy $C_{\rm sw}$.}
\label{fig:rate_sweep_all}
\end{figure*}

\vspace{10pt}

\subsubsection{Hard-decision binary Markov noise}

We also evaluate binary additive hard-decision noise models with temporal correlation. These experiments replace the real-valued soft-observation model by a binary additive channel and compare a finite-memory noise-effect ordering with Hamming-weight GRAND when the binary noise process has memory. 
For the binary Markov noise model, the transition probabilities are
\begin{align}
\Pr\{Z_i=1\mid Z_{i-1}=0\} &= p_{01},\\
\Pr\{Z_i=0\mid Z_{i-1}=1\} &= p_{10},
\end{align}
with $p_{10}=0.3$. The initial noise state is drawn from the stationary distribution of this two-state Markov chain. The baseline is Hamming-weight GRAND. The matched Markov Tail-Calibrated SOGRAND decoder uses the first-order Markov noise prior. Fig.~\ref{fig:binary_markov_all} and Table~\ref{tab:binary_markov} show that, for all tested values of $p_{01}$, the matched finite-memory ordering gives lower BLER point estimates and smaller average membership-query counts than Hamming-weight
GRAND.

\begin{figure*}[!t]
\centering
\begin{minipage}{0.32\linewidth}
\centering
\includegraphics[width=\linewidth]{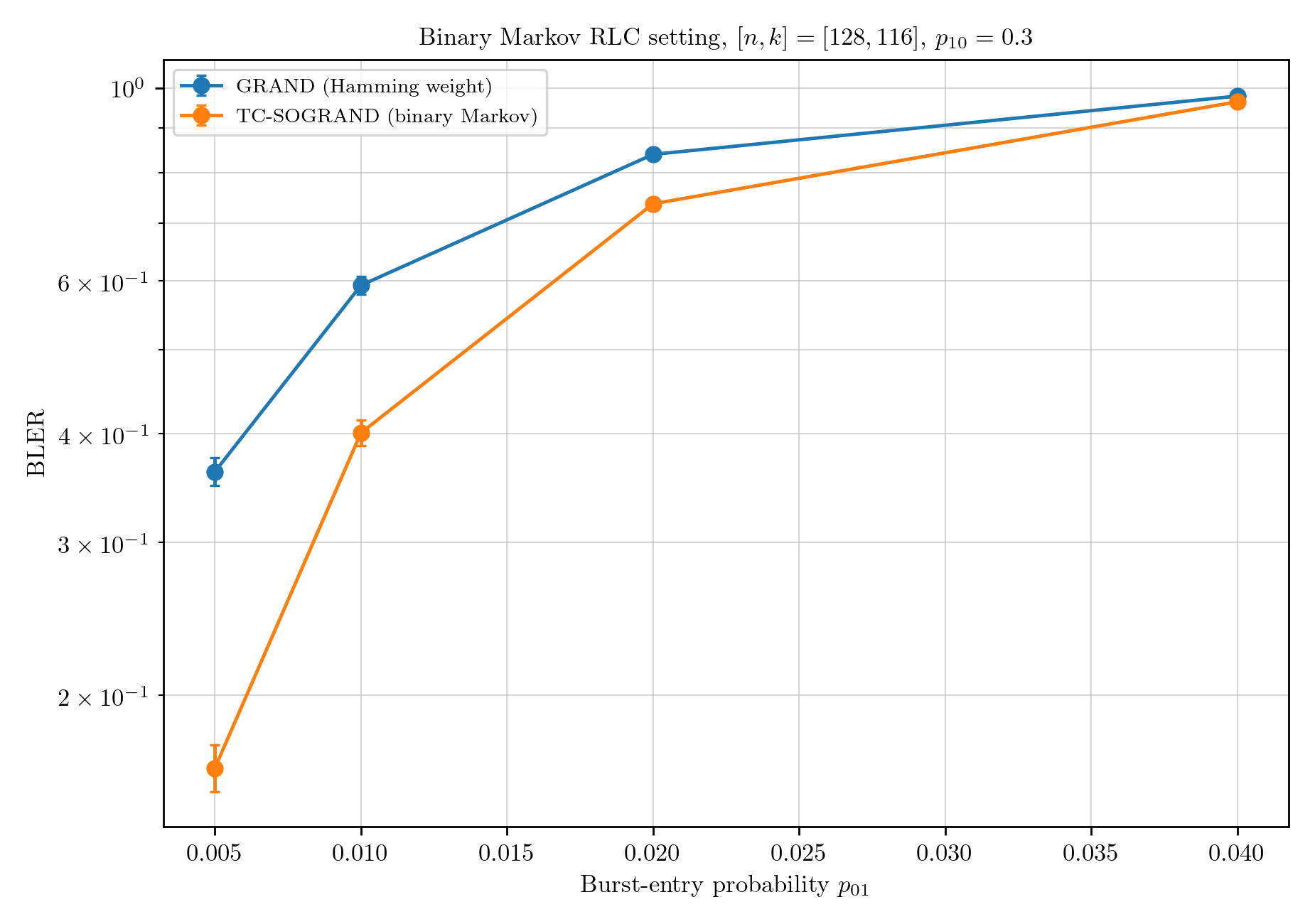}
\end{minipage}
\hfill
\begin{minipage}{0.32\linewidth}
\centering
\includegraphics[width=\linewidth]{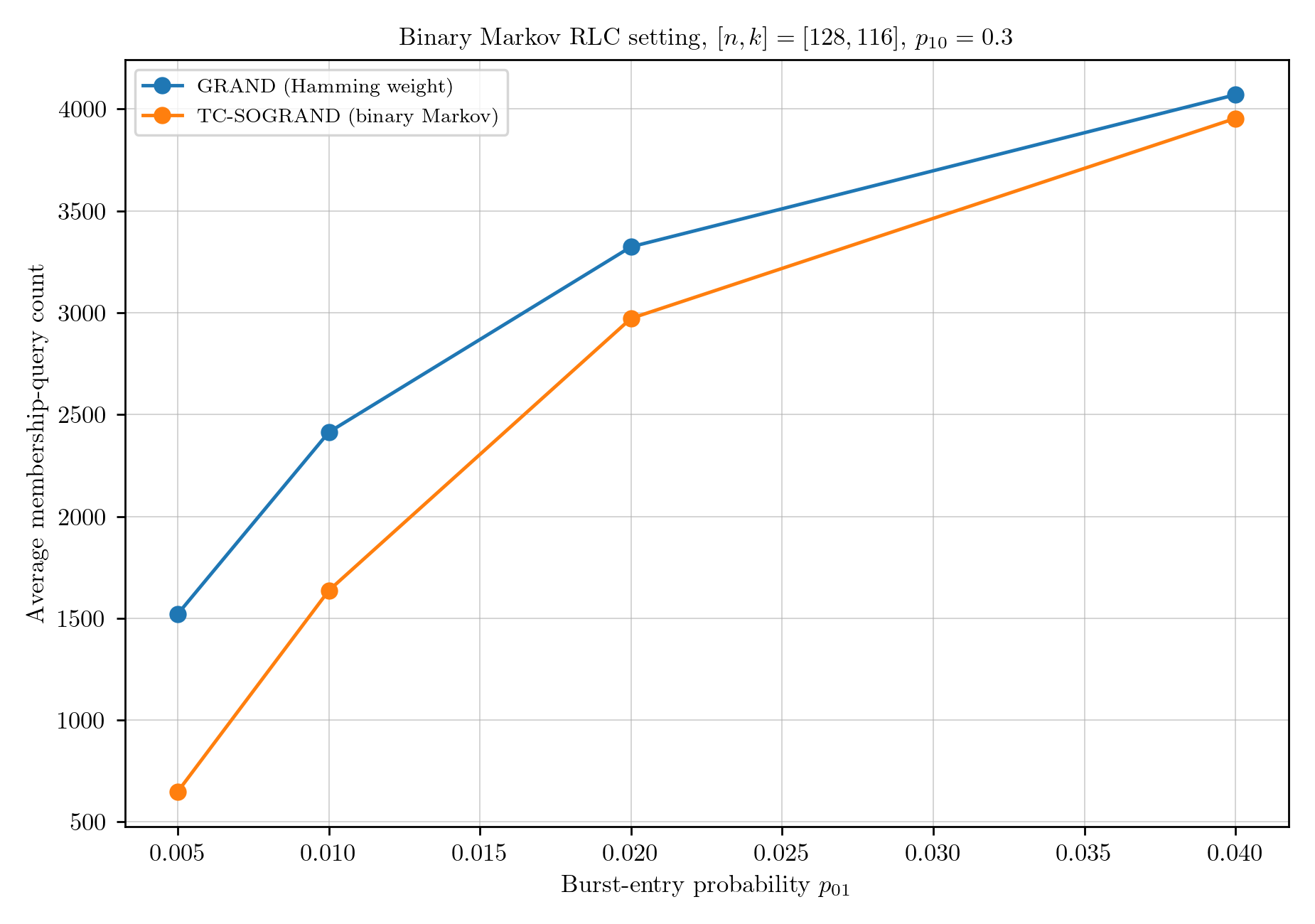}
\end{minipage}
\hfill
\begin{minipage}{0.32\linewidth}
\centering
\includegraphics[width=\linewidth]{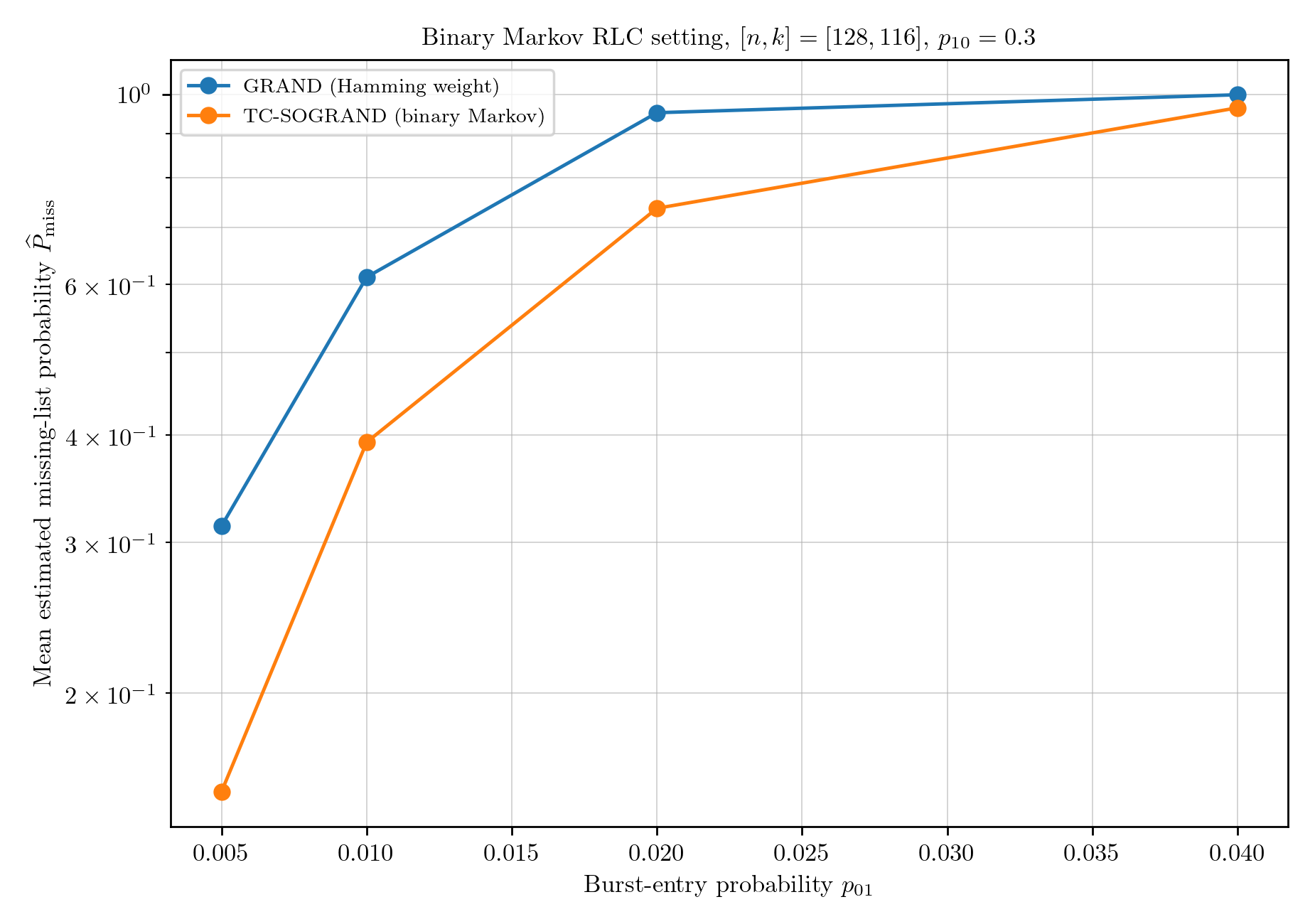}
\end{minipage}
\caption{Hard-decision binary Markov noise experiment for the $[128,116]$ RLC, with $p_{10}=0.3$, $q_{\max}=20000$, and $5000$ frames per point. Left: BLER. Middle: average membership-query count.
Right: mean estimated missing-list probability.}
\label{fig:binary_markov_all}
\end{figure*}

\begin{table}[h]
\centering
\caption{Hard-decision binary Markov noise experiment for the
$[128,116]$ RLC, with $p_{10}=0.3$, $q_{\max}=20000$, and $5000$
frames per point.}
\vspace{-4pt}
\resizebox{\linewidth}{!}{
\label{tab:binary_markov}
\begin{tabular}{c|cc|cc}
\toprule
& \multicolumn{2}{c|}{Hamming-weight GRAND}
& \multicolumn{2}{c}{Matched Markov} \\
$p_{01}$ & BLER & Average query count & BLER & Average query count\\
\midrule
0.005 & 0.3618 & 1520.6 & 0.1648 & 648.5\\
0.010 & 0.5932 & 2413.5 & 0.4010 & 1636.2\\
0.020 & 0.8396 & 3323.8 & 0.7364 & 2972.2\\
0.040 & 0.9800 & 4070.4 & 0.9654 & 3954.3\\
\bottomrule
\end{tabular}}
\end{table}

\vspace{10pt}

\subsubsection{Gilbert--Elliott hard-decision noise}

\begin{figure}[!b]
\centering
\begin{minipage}{0.75\linewidth}
\centering
\includegraphics[width=\linewidth]{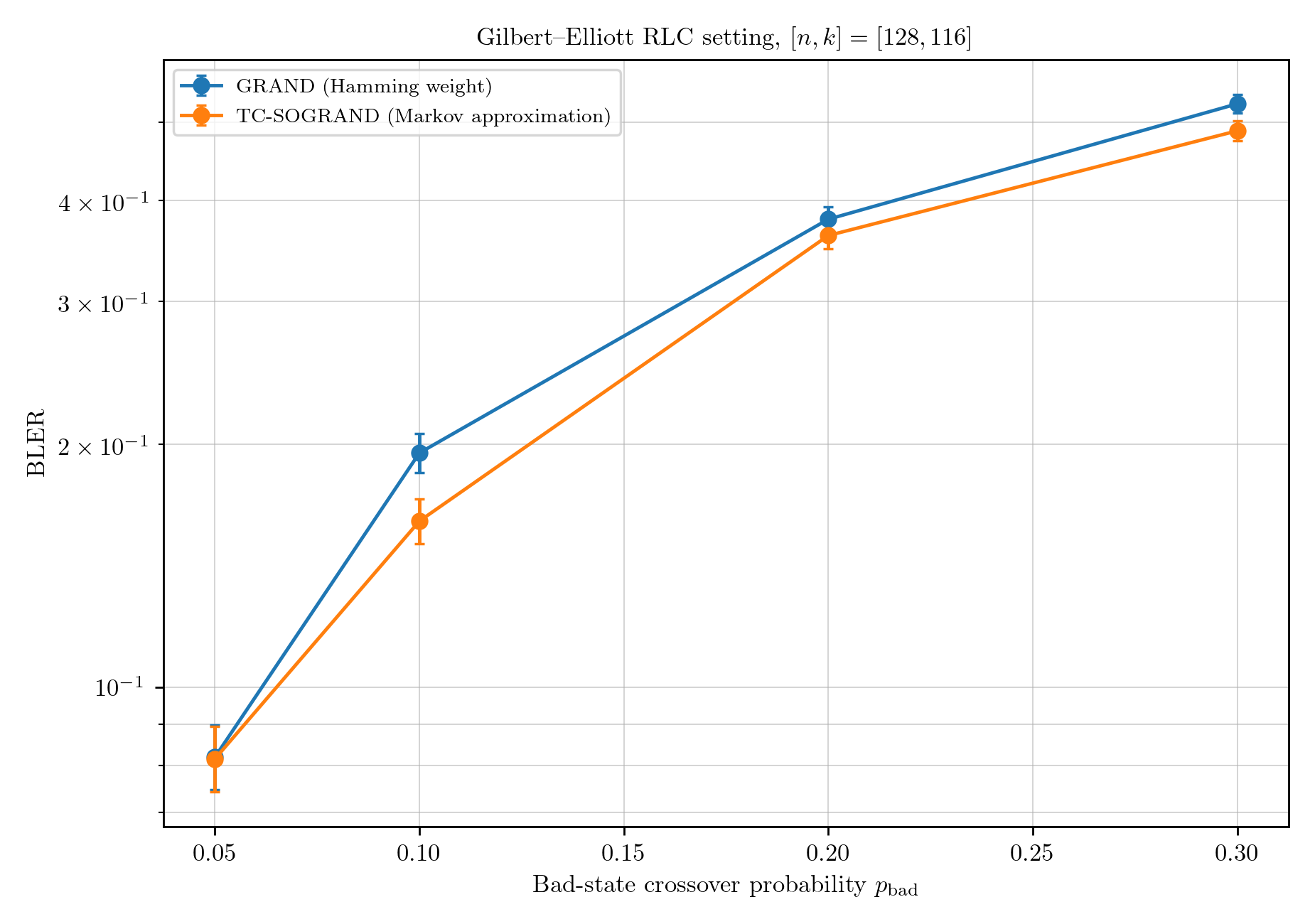}
\end{minipage}

\begin{minipage}{0.75\linewidth}
\centering
\includegraphics[width=\linewidth]{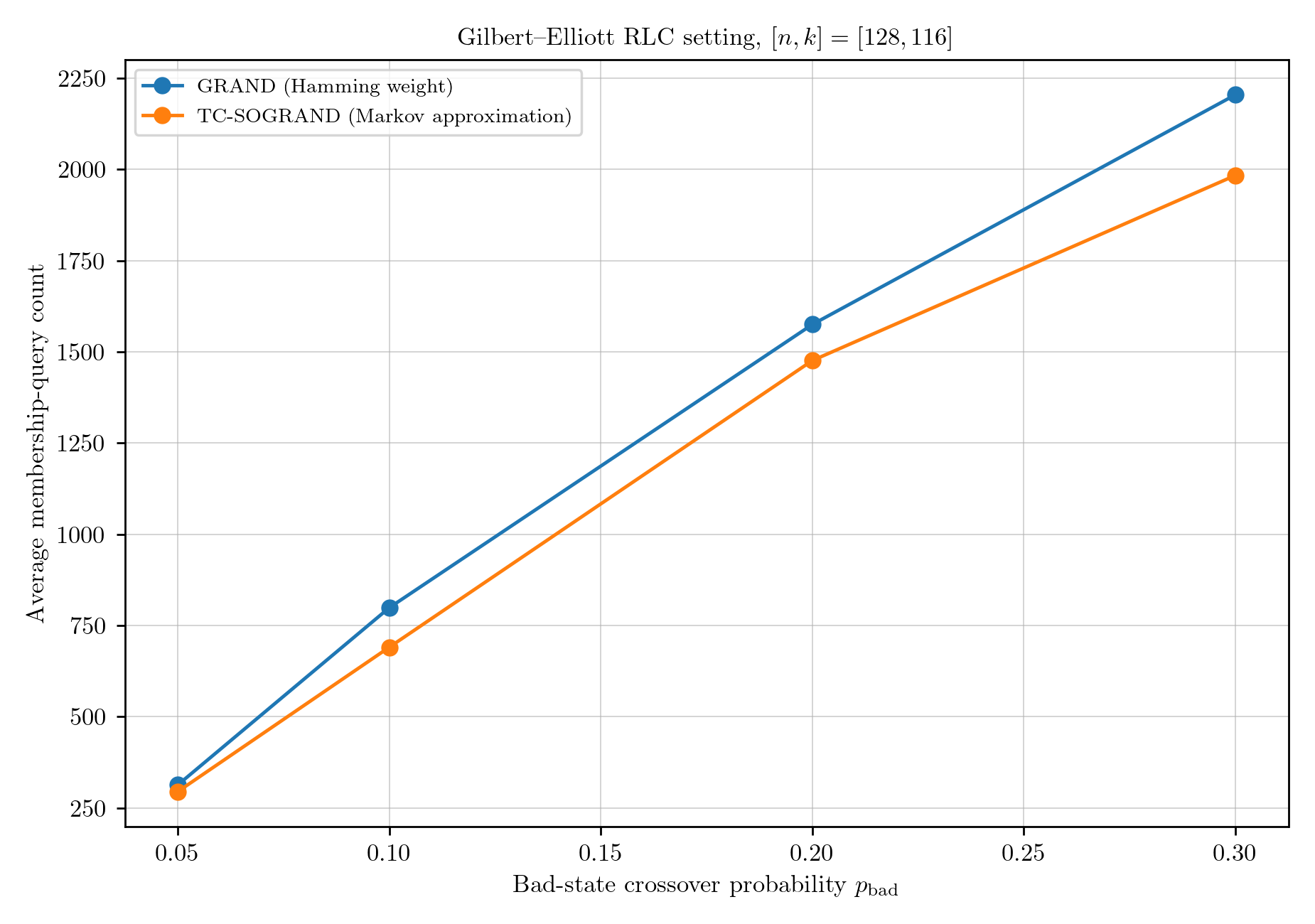}
\end{minipage}
\caption{Gilbert--Elliott hard-decision experiment for the $[128,116]$ RLC with $q_{\max}=20000$ and $5000$ frames per point. The decoder uses a first-order Markov approximation to the marginal binary noise-effect process. Up: BLER. Bottom: average membership-query count.}
\label{fig:ge_all}
\end{figure}

For the Gilbert--Elliott experiment \cite{gilbert1960capacity, elliott1963estimates}, the hidden state takes values $G$ and $B$, with transition probabilities
\[
\Pr\{G\to B\}=0.02, \qquad
\Pr\{B\to G\}=0.30 .
\]
The initial hidden state is drawn from the stationary distribution of this two-state Markov chain. Conditional on the hidden state, the binary noise symbol has crossover probability $10^{-3}$ in state $G$ and $p_{\rm bad}$ in state $B$, where $p_{\rm bad}$ is varied. The decoder uses a first-order Markov approximation to the marginal binary noise-effect process. For each tested value of $p_{\rm bad}$, the transition probabilities of this approximation are estimated from an independent simulated noise sequence of length $2\times10^5$. This decoder does not use the exact hidden-state posterior; it retains only first-order marginal transition statistics of the observed binary noise sequence.

Fig.~\ref{fig:ge_all} and Table~\ref{tab:ge} compare this Markov approximation with Hamming-weight GRAND. For each tested value of $p_{\rm bad}$, the Markov approximation gives a smaller average membership-query count and a lower BLER point estimate, although the BLER difference is very small at $p_{\rm bad}=0.05$. The gains are smaller than in the matched binary Markov experiment, which is consistent with the use of a marginal first-order approximation rather than the exact hidden-state posterior.

\begin{table}[h]
\centering
\caption{Gilbert--Elliott hard-decision experiment for the $[128,116]$ RLC with $q_{\max}=20000$ and $5000$ frames per point. The decoder uses a first-order Markov approximation to the marginal binary
noise-effect process.}
\label{tab:ge}
\resizebox{\linewidth}{!}{
\begin{tabular}{c|cc|cc}
\toprule
& \multicolumn{2}{c|}{Hamming-weight GRAND}
& \multicolumn{2}{c}{Markov approximation} \\
\hline
$p_{\rm bad}$ & BLER & Average query count & BLER & Average query count\\
\midrule
0.05 & 0.0820 & 312.5  & 0.0816 & 294.3\\
0.10 & 0.1950 & 798.4  & 0.1606 & 690.2\\
0.20 & 0.3794 & 1574.9 & 0.3622 & 1476.1\\
0.30 & 0.5270 & 2204.6 & 0.4880 & 1983.6\\
\bottomrule
\end{tabular}}
\end{table}


\section{Conclusion}

We introduced a finite-memory soft-output GRAND decoder whose stopping rule is based on a plug-in estimate of the missing-list probability rather than on a fixed membership-query limit. Under exact nondecreasing-energy enumeration, the first codeword encountered in the query order is an ML codeword for the likelihood model that defines the posterior energy, provided that the search is not abandoned before this codeword is reached. Under the fixed-size random-codebook model, the unqueried codebook-restricted denominator contribution has conditional expectation equal to the ambient posterior tail mass multiplied by the remaining random-codebook occupancy probability. This separates the ambient posterior tail from the codebook-restricted APP denominator and gives the occupancy correction used in the missing-list estimate.
Our formulation extends the random-codebook missing-list denominator calculation of SOGRAND to finite-memory noise-effect posterior models. It also gives blockwise and bitwise APP estimates under the same random-codebook occupancy approximation. The bitwise estimates define APP LLRs for the component decoder. With a priori input LLRs included in the candidate metric, extrinsic LLRs are obtained by subtracting the input LLRs from the resulting APP LLRs. Numerical results in the tested Gauss--Markov random-linear-code setting show lower BLER point estimates and smaller average membership-query counts relative to the tested memoryless and approximate-independence GRAND orderings, while also showing an enumeration-cost tradeoff. Designing lower-complexity enumerators for finite-memory posterior energies and evaluating the resulting soft output inside full iterative product-code and GLDPC
decoders remain directions for future work.

\bibliographystyle{IEEEtran}
\bibliography{references}

\end{document}

%% file: preamble_IEEE.tex
\makeatletter
\@ifundefined{ifisdraft}{\newif\ifisdraft}{}
\isdrafttrue   
\makeatother

\usepackage[utf8]{inputenc} 
\usepackage[T1]{fontenc}    

\usepackage{lipsum}         
\usepackage{url}            
\usepackage{xcolor}         
\usepackage{microtype}      

\definecolor{headercolor}{gray}{0.82}  
\definecolor{zebracolor}{gray}{0.95}   
\definecolor{lightgray}{gray}{0.93}     
\definecolor{darkred}{rgb}{0.6, 0, 0}    

\usepackage{cuted}       
\usepackage{multicol}    
\usepackage{setspace}    
\usepackage{footmisc}    
\usepackage{chngcntr}    
\usepackage{enumitem}    
\usepackage{sansmath}    
\usepackage{comment}
\usepackage{derivative}
\usepackage{ifthen}

\usepackage{amsmath, amssymb, amsfonts, amsbsy} 
\usepackage{mathtools}    
\usepackage{mathrsfs}     
\usepackage{bm}           
\usepackage{dsfont}       
\usepackage{relsize}      
\usepackage{stmaryrd}     
\usepackage{bbm}          

\usepackage{booktabs}     
\usepackage{array}        
\usepackage{tabularx}     
\usepackage{longtable}    
\usepackage{makecell}     

\usepackage{algorithm}
\usepackage{algpseudocode}

\usepackage{graphicx}     
\usepackage{epsfig}       
\usepackage{subcaption}   
\usepackage{caption}      

\usepackage{tikz}
\usetikzlibrary{arrows.meta,positioning,fit,shapes.misc}
\usetikzlibrary{calc}

\usepackage{doi}         
\usepackage{cite}        

\usepackage{hyperref}
\hypersetup{
    colorlinks   = true,
    linkcolor    = {red!50!black},
    citecolor    = {blue!50!black},
    urlcolor     = {black}
}
\usepackage{cleveref}    

\renewcommand\thesection{\arabic{section}}
\renewcommand\thesubsection{\arabic{section}.\arabic{subsection}}
\def\thesectiondis{\thesection.}
\def\thesubsectiondis{\thesectiondis\arabic{subsection}.}

\makeatletter
\def\@seccntformatinl#1{\csname the#1dis\endcsname\hskip 1em\relax}
\makeatother

\usepackage{tocloft}
\usepackage{amsthm}
\newtheorem{theorem}{Theorem}[section]

\newtheorem{remark}[theorem]{Remark}

\def\BibTeX{{\rm B\kern-.05em{\sc i\kern-.025em b}\kern-.08em
    T\kern-.1667em\lower.7ex\hbox{E}\kern-.125emX}}

\newcommand{\orcid}[1]{%
  \href{https://orcid.org/#1}{\includegraphics[width=10pt]{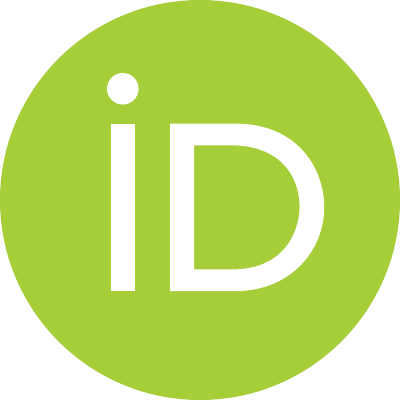}}%
}

\def\markov{\hbox{$\--$}\kern-1.5pt\hbox{$\circ$}\kern-1.5pt\hbox{$\--$}}

\makeatletter
\newcommand*{\centernot}{%
  \mathpalette\@centernot
}
\def\@centernot#1#2{%
  \mathrel{%
    \rlap{%
      \settowidth\dimen@{$\m@th#1{#2}$}%
      \kern.5\dimen@
      \settowidth\dimen@{$\m@th#1=$}%
      \kern-.5\dimen@
      $\m@th#1\not$%
    }%
    {#2}%
  }%
}
\makeatother

\makeatletter
\newcommand*{\rom}[1]{\expandafter\@slowromancap\romannumeral #1@}
\makeatother

